\documentclass{IEEEtran}
\usepackage{graphicx}
\usepackage{subfigure}
\usepackage{amsmath}
\usepackage{amsthm}
\usepackage{amssymb}
\usepackage{mathrsfs}
\usepackage{slashbox}
\usepackage{float}
\usepackage[vlined,ruled,linesnumbered]{algorithm2e}
\usepackage{xcolor}

\newtheorem{thm}{Theorem}

\newtheorem{lem}{Lemma}
\newtheorem{cor}{Corollary}
\newtheorem{definition}{Definition}

\newcount\DraftStatus  
\usepackage{color}
\definecolor{darkgreen}{rgb}{0,0.5,0}
\definecolor{purple}{rgb}{1,0,1}
\newcommand{\draftnote}[2]{\ifnum\DraftStatus=1
	\marginpar{
		\tiny\raggedright
		\hbadness=10000
        \def\baselinestretch{0.8}
        \textcolor{#1}{\textsf{\hspace{0pt}#2}}}
     \fi}

\begin{document}

\title{A Differentially Private Game Theoretic Approach for Deceiving Cyber Adversaries}

\author{Dayong Ye, Tianqing Zhu*, Sheng Shen and Wanlei Zhou
\thanks{*Tianqing Zhu is the corresponding author. D. Ye, T. Zhu, S. Shen and W. Zhou are with the Centre for Cyber Security and Privacy and the School of Computer Science, University of Technology, Sydney, Australia. Email: \{Dayong.Ye, Tianqing.Zhu, 12086892, Wanlei.Zhou\}@uts.edu.au}}

\maketitle
\pagestyle{empty}
\thispagestyle{empty}

\begin{abstract}
Cyber deception is one of the key approaches used to mislead attackers
by hiding or providing inaccurate system information.
There are two main factors limiting the real-world application of existing cyber deception approaches.
The first limitation is that the number of systems in a network is assumed to be fixed.
However, in the real world, the number of systems may be dynamically changed.
The second limitation is that attackers' strategies are simplified in the literature.
However, in the real world, attackers may be more powerful than theory suggests.
To overcome these two limitations, 
we propose a novel differentially private game theoretic approach to cyber deception.
In this proposed approach, a defender adopts differential privacy mechanisms to
strategically change the number of systems and obfuscate the configurations of systems,
while an attacker adopts a Bayesian inference approach to infer the real configurations of systems.
By using the differential privacy technique,
the proposed approach can 1) reduce the impacts on network security resulting from changes in the number of systems
and 2) resist attacks regardless of attackers' reasoning power.
The experimental results demonstrate the effectiveness of the proposed approach.
\end{abstract}


\section{Introduction}\label{sec:introduction}
Network security is one of the most important problems faced by enterprises and countries today \cite{Goel16}. 
Before launching a network attack,
malicious attackers often scan systems in a network to identify vulnerabilities
that can be exploited to intrude into the network \cite{Carroll11}.
The aim of this scanning is to understand the configurations of these systems,
including the operating systems they are running,
and their IP/MAC addresses on the network.
Once these questions are answered,
attackers can efficiently formulate plans to attack the network.
In order to prevent attackers from receiving true answers to these questions 
and thus reduce the likelihood of successful attacks,
cyber deception techniques are employed \cite{Alme16}.

Instead of stopping an attack or identifying an attacker,
cyber deception techniques aim to mislead an attacker and induce him to attack non-critical systems
by hiding or lying about the configurations of the systems in a network \cite{Albanese16}.
For example, an important system $s$ with configuration $k$ is obfuscated by the defender to appear as a less important system with configuration $k'$.
Thus, when an attacker scans the network,
he observes system $s$ with configuration $k'$ rather than $k$, as shown in Fig. \ref{fig:simpleExample}.
Since configuration $k'$ is less important than $k$,
the attacker may skip over system $s$.
\begin{figure}[ht]
\vspace{-2mm}
\centering
	\includegraphics[scale=0.7]{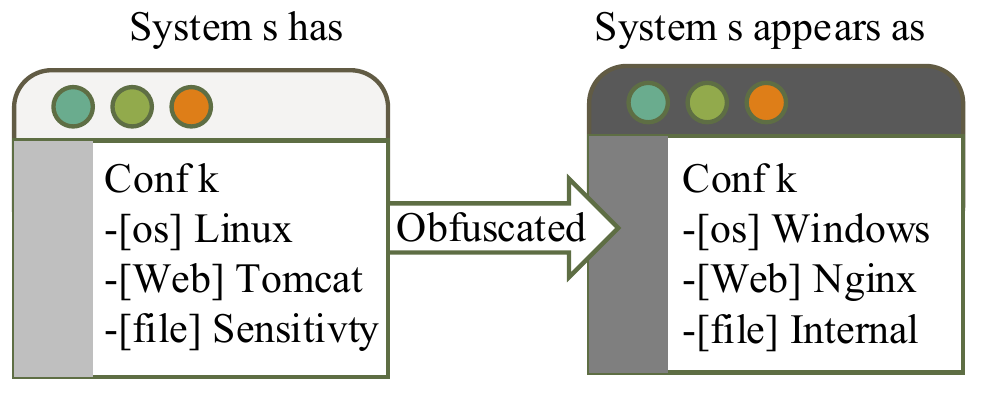}
	\vspace{-2mm}
	\caption{System $s$ with configuration $k$ is obfuscated to and appears as configuration $k'$}
	\label{fig:simpleExample}
	\vspace{-1mm}
\end{figure}

To model such an interaction between an attacker and a defender,
game theory has been adopted as a means of studying cyber deception \cite{Kiennert18,Shaer19}.
Game theory is a theoretical framework to study the decision-making strategies of competing players,
where each player aims to maximize her or his own utility.
In cyber deception, defenders and attackers can be modelled as players.
Game theory, thus, can be used to investigate how a defender reacts to an attacker and vice versa.
Game theoretic formulation overcomes traditional solutions to cyber deception in many aspects,
such as proven mathematics, reliable defense mechanisms and timely actions \cite{Do17}.

Existing game theoretic approaches, however, have two common limitations.
The first limitation is that
the number of systems in a network is assumed to be fixed.
This assumption hinders the applicability of existing approaches to real-world applications,
because in the real world, the number of systems in a network may be dynamically changed.
Although some existing approaches, like \cite{Jajodia17}, are computationally feasible to recompute the strategy
when new systems are added, this change in the number of systems may affect the security of a network \cite{Modi13,Zarp17}.
For example, a network has three systems: $s_1$, $s_2$ and $s_3$.
The three systems are associated with two configurations,
where $s_1$ and $s_2$ are associated with configuration $k_1$,
and $s_3$ is associated with configuration $k_2$.
The attacker knows that the network has three systems and two configurations.
He also knows that configuration $k_1$ has two systems and configuration $k_2$ has one system.
He, however, is unaware of which system is associated with which configuration.
Now when system $s_3$ is removed by the defender,
the attacker knows that configuration $k_1$ has two systems and configuration $k_2$ has zero system.
By comparing the two results,
the attacker can deduce that system $s_3$ is associated with configuration $k_2$.
After deducing system $s_3$'s configuration,
the attacker can also deduce that systems $s_1$ and $s_2$ are associated with configuration $k_1$.
Similarly, when a new system $s_4$ is added and associated with configuration $k_2$,
the attacker can immediately deduce the configuration of system $s_4$
because he realizes that the number of systems in configuration $k_2$ increases by $1$.
Therefore, the problem of how to strategically change the number of systems in a network without sacrificing its security is challenging.

The second limitation is that
attackers' approaches are simplified in the literature through the use of only greedy-like approaches.
Greedy-like approaches create deterministic strategies which are highly predictable to defenders.
Hence, using these approaches may underestimate the ability of real-world attackers
who can use more complex approaches that are much less predictable to defenders.
If a defender's approach is developed based on simplified attackers' approaches,
the defender may be easily compromised by real-world attackers.
Thus, the problem of how to develop an efficient defender's approach against powerful attackers is also challenging \cite{Ding2018}.

Accordingly, to overcome the above two limitations,
in this paper, we propose a novel differentially private game theoretic approach to cyber deception.
By using the differential privacy technique \cite{Dwork06},
the change in the number of systems does not affect the security of a network,
as an attacker cannot determine whether a given system is inside or outside of the network.
Thus, the attacker cannot deduce each system's true configuration. 
Moreover, by using the differential privacy-based approach,
a defender's strategy is unpredictable to an attacker, irrespective of his reasoning power.

In summary, the contribution of this paper is two-fold.

1) To the best of our knowledge, we are the first to apply differential privacy to the cyber deception game in order to overcome the two common limitations mentioned above.
	

2) We theoretically illustrate the properties of the proposed approach,
  and experimentally demonstrate its effectiveness.


\section{Related Work}\label{sec:related work}
In this section, we first review related works about game theory for cyber deception
and next discuss their limitations.
A thorough survey on game theory for general network and cyber security can be found in \cite{Man13,Liang13,Do17,Kiennert18}.


\subsection{Review of related works}
Albanese et al. \cite{Alba15} propose a deceptive approach to defeat an attacker's effort to fingerprint operating systems and services.
For operating system fingerprinting, they manipulate the outgoing traffic to make it resemble traffic generated by a different system.
For service fingerprinting, they modify the service banner by
intercepting and manipulating certain packets before they leave the network.

Albanese et al. \cite{Albanese16} present a graph-based approach to manipulate the attacker's view of a system's attack surface.
They formalize system views and define the distance between these different views.
Based on this formalization and definition, they develop an approach to confuse the attacker by manipulating responses to his probes
in order to induce an external view of the system, which can minimize the costs incurred by the defender.

Jajodia et al. \cite{Jajodia17} develop a probabilistic logic of deception
and demonstrate that these associated computations are NP-hard.
Specifically, they propose two algorithms that allow the defender to generate faked scan results
in different states to minimize the damage caused by the attacker.

Wang and Zeng \cite{Wang18c} propose a two-stage deception game for network protection.
In the first stage, the attacker scans the whole network
and the defender decides how to respond these scan queries
to distract the attacker from high value hosts.
In the second stage, the attacker selects some hosts for further probing
and the defender decides how to answer these probes.
The defender's decisions are formulated as optimization problems solved by heuristic algorithms.

Schlenker et al. \cite{Schlenker18} propose a game theoretic approach to deceive cyber adversaries.
Their approach involves two types of attackers: a rational attacker and a naive attacker.
The rational attacker knows the defender's deception strategy,
while the naive attacker does not.
Two optimal deception algorithms are developed to counter these two types of attackers.

Bilinski et al. \cite{Bilinski19} propose a simplified cyber deception game model.
In their model, there are only two machines,
where one is real and the other is fake.
At each round, the attacker is allowed to ask one machine about its type,
and the machine can either lie or tell the truth about its type.
After $N$ rounds, the attacker chooses one machine to attack.
The attacker will receive a positive payoff if he attacks the real machine.
Otherwise, he will receive a negative payoff.

Other game theoretic approaches have also been developed \cite{Huang19,Cho19}.
These, however, are not closely related to our work,
as they do not particularly focus on obfuscating configurations of systems.
Some game theoretic approaches \cite{Nguyen19,Pawlick19} mainly focus on dealing with the interactions between defenders and attackers,
while other approaches \cite{Thakoor19,Gan19} aim at investigating the behaivors of defenders and attackers.
Another branch of related work deals with the repeated interactions of defender and attacker
over time following the ``FlipIt'' model \cite{Dijk13,Bowers12,Laszka14}.
Their research focuses on stealthy takeovers,
where both the defender and the attacker want to take control of a set of resources by flipping them.
By comparison, our research focuses on how a defender modifies the configuration of systems to deceive an attacker.

\subsection{Discussion of related works}
The above-reviewed works have two common limitations:
1) the number of systems in a network is assumed to be fixed,
and 2) the attacker's approaches are often simplified.

\subsubsection{Fixed number of systems}
In the real-world networks, the number of systems is often changed in a dynamic manner.
Most of existing works, however, are conducted on the basis of fixed number of systems,
though some works are computationally feasible to accommodate the increase of the number of systems.
An intuitive solution is to extend existing works by enabling systems to be dynamically introduced or removed.
However, 
arbitrarily introducing or removing systems
may compromise the security of a network.

\subsubsection{Simplified attackers}
In existing works, attackers are considered to be either naive \cite{Jajodia17}
or to use only greedy-like approaches \cite{Schlenker18}.
A naive attacker always attacks the observed configuration which has the highest utility.
A greedy attacker knows the defender's deception scheme
and attacks the configuration which has the expected highest utility.
The drawback of naive and greedy-like approaches is that
the selection of a configuration is deterministic
and highly predictable to a defender.
As real-world attackers are more powerful than the simplified naive and greedy attackers,
defenders' approaches developed based on simplified attackers' approaches
may not be applicable to the real world.

In this paper, we develop a novel differentially private game theoretic approach
which can strategically change the number of systems in a network without sacrificing its security.
Moreover, to model a powerful attacker, we adopt a Bayesian inference approach.
By using Bayesian inference, the attacker can infer the probability
with which an observed configuration $k'$ could be a real configuration $k$,
and selects a configuration to attack based on the probability distribution over the observed configurations.
Therefore, the selection of a configuration using the Bayesian inference approach is non-deterministic
and hardly predictable to the defender.

\section{Preliminaries}\label{sec:preliminaries}
\subsection{Cyber deception game}\label{sub:game}
The cyber deception game (CDG) is an imperfect-information Stacklberg game
between a defender and an attacker \cite{Schlenker18},
where a player cannot accurately observe the actions of the other player \cite{Do17}.
The defender takes the first actions.
She decides how the systems should respond to scan queries from an attacker by obfuscating the configurations of these systems.
The attacker subsequently follows by choosing which systems to attack based on his scan query results.
The game continues in this alternating manner, 
until a pre-defined number of rounds is reached.
We use the imperfect information game because in this cyber deception game,
the aim of the defender is to hide the configurations of systems from the attacker
rather than stopping an attack or identifying the attacker.
If we use a perfect information game,
the defender’s previous strategies of obfuscating configurations can be accurately observed by the attacker.
Then, the attacker can immediately deduce the real configurations of all the systems.

In this game, we use $N$ to denote the set of systems in a network protected by the defender.
The number of systems is denoted by $|N|$.
Each system has a set of attributes: an operating system, a file system, services hosted, and so on.
These attributes constitute a system configuration.
We use $K$ to denote the set of configurations,
and the number of configurations is denoted by $|K|$.
According to the configurations, the system set $N$ can be partitioned into a number of subsets: $N_1,...,N_{|K|}$,
where $N_i\cap N_j=\emptyset$ for any $i\not= j$, and $N_1\cup...\cup N_{|K|}=N$.
Each of the systems in subset $N_k$ has configuration $k$ and an associated utility $u_k$.
If a system with configuration $k$ is attacked by the attacker,
the attacker receives utility $u_k$ and the defender loses utility $u_k$.
The utility of a system depends on not only the importance of the system
but also on its security level.
A well-chronicled finding of information security is that
attackers may go for the weakest link, i.e., the system with the weakest security level, to have a foothold
and then try to progress from there via privilege escalation \cite{Arce03}.
Detailed theoretical discussion on the weakest link can be found in \cite{Varian03,Bohme10}.

To mislead the attacker, the defender may obfuscate the configurations of these systems.
Obfuscating a system with configuration $k$ to appear as $k'$ incurs the defender a cost $c(k',k)$.
When scanning, the attacker observes the obfuscated configurations of these systems.
If a system with configuration $k$ is obfuscated to appear as $k'$,
the system's real configuration is still $k$ rather than $k'$.
Thus, when this system is attacked by the attacker,
the attacker receives utility $u_k$ instead of $u_{k'}$.

After obfuscation, the system set $N$ becomes $N'$,
which can be partitioned into: $N'_1,...,N'_{|K|}$.
Since $|N|$ may not be equal to $|N'|$,
there may be void systems ($|N'|>|N|$), or some systems may be taken offline ($|N'|<|N|$).
Void systems can be interpreted as honeypots\footnote{Although our work also uses honeypots, these do not permanently exist in the network, but are rather randomly introduced by our differentially private approach. Our main aim is still to focus on the cyber deception.}.
Deploying a honeypot with configuration $k$ incurs the defender a cost $h_k$.
If the attacker attacks the honeypot,
he receives a negative utility $-u_k$.
Usually, $u_k>h_k$, as the defender would not otherwise have the motivation to deploy a honeypot.
By taking a system offline, the defender loses utility $l_k$,
but this system will not be attacked by the attacker.
Again, generally, $u_k>l_k$, as the defender would not otherwise have the motivation to take a system offline.
While the defender aims to minimize her utility loss,
the attacker aims to maximize his utility gain.

The attacker is aware of the defender's obfuscation approach,
and knows the statistical information pertaining to the network:
the number of systems $|N|$, the number of configurations $|K|$, and the number of the systems associated with each configuration $|N_1|,...,|N_{|K|}|$.
The attacker, however, does not know which system is associated with which configuration.
Moreover, as the aim of the defender is to hide the configurations of systems,
the defender's strategies, i.e., when and how to make $k$ appear as $k'$,
are exactly what the defender wants to hide
and cannot be accurately observed by the attacker.
However, the attacker can partially observe the defender’s previous actions.
As in the above example, during each round, the attacker attacks a system which appears as configuration $k'$,
but after the attack, the attacker receives utility $u_k$.
The attacker, thus, knows that the attacked system was obfuscated from $k$ to $k'$ in this round.
The attacker accumulates this information and
uses Bayesian inference to make decisions in future rounds.
Specifically, in Bayesian inference, the posterior probability $q(k'|k)$ takes this partial observation into account by computing how many times $k$ appears as $k'$ in previous rounds.
The details will be given in Section \ref{sec:method}.

\subsection{Differential privacy}\label{sub:DP}
Differential privacy is a prevalent privacy model capable of guaranteeing that
any individual record being stored in or removed from a dataset makes little difference to the analytical output of the dataset \cite{Dwork06}.
Differential privacy has been successfully applied to cyber physical systems \cite{Zhu20}, machine learning \cite{Ye19} and artificial intelligence \cite{Zhu19,Zhu20TKDE}.

In differential privacy, two datasets $D$ and $D'$ are neighboring datasets if they differ by only one record.
A query $f$ is a function that maps dataset $D$ to an abstract range $\mathbb{R}$, $f: D\rightarrow\mathbb{R}$.
The maximal difference in the results of query $f$ is defined as sensitivity $\Delta S$,
which determines how much perturbation is required for the privacy-preserving answer.
The formal definition of differential privacy is presented as follows.



\begin{definition}[$\epsilon$-Differential Privacy \cite{Dwork14}]\label{Def-DP}
A mechanism $\mathcal{M}$ provides $\epsilon$-differential privacy for any pair of neighboring datasets $D$ and $D'$, and for every set of outcomes $\Omega$, if $\mathcal{M}$ satisfies:
\begin{equation}
Pr[\mathcal{M}(D) \in \Omega] \leq \exp(\epsilon) \cdot Pr[\mathcal{M}(D') \in \Omega]
\end{equation}
\end{definition}


\begin{definition}[Sensitivity \cite{Dwork14}]\label{Def-GS} For a query $f:D\rightarrow\mathbb{R}$, the sensitivity of $f$ is defined as
\begin{equation}
\Delta S=\max_{D,D'} ||f(D)-f(D')||_{1}
\end{equation}
\end{definition}

Two of the most widely used differential privacy mechanisms are the Laplace mechanism and the exponential mechanism \cite{Zhu17}.
The Laplace mechanism adds Laplace noise to the true answer.
We use $Lap(b)$ to represent the noise
sampled from the Laplace distribution with scaling $b$.
\begin{definition}[The Laplace Mechanism \cite{Dwork14}]\label{Def-LA}
Given a function $f: D \rightarrow \mathbb{R}$ over a dataset $D$, Equation~\ref{eq-lap} is the Laplace mechanism.
\begin{equation}
\widehat{f}(D)=f(D)+Lap(\frac{\Delta S}{\epsilon})
\end{equation}\label{eq-lap}
\end{definition}
\begin{definition}[The Exponential Mechanism \cite{Dwork14}]\label{Def-Ex}
The exponential mechanism $\mathcal{M}_E$ selects and outputs an element $r\in\mathcal{R}$
with probability proportional to $exp(\frac{\epsilon u(D,r)}{2\Delta u})$,
where $u(D,r)$ is the utility of a pair of dataset and output,
and $\Delta u=\max\limits_{r\in\mathcal{R}} \max\limits_{D,D':||D-D'||_1\leq 1} |u(D,r)-u(D',r)|$
is the sensitivity of utility.
\end{definition}

Table \ref{tab:notation} presents the notations and terms used in this paper.
\begin{table}[!ht]\scriptsize
\vspace{-3mm}
\newcommand{\tabincell}[2]{\begin{tabular}{@{}#1@{}}#2\end{tabular}}
	\centering
	\caption{The meaning of notations used in this paper}
\begin{tabular} {|p{1cm}|p{5cm}|} \hline
\textbf{Notations}&\textbf{Meaning}\\ \hline
$N$ & a set of systems\\\hline
$K$ & a set of configurations\\\hline
$N'$ & a set of obfuscated systems\\\hline
$u_k$ & the utility of a system with configuration $k$\\\hline
$c(k',k)$ & cost to the defender of obfuscating a system from configuration $k$ to appear as $k'$\\\hline
$h_k$ & cost to the defender's of deploying a honeypot with configuration $k$\\\hline
$l_k$ & utility loss incurred by the defender to take a system with configuration $k$ offline\\\hline
$U_d$ & the defender's total expected utility loss\\\hline
$C_d$ & the defender's total cost\\\hline
$B_d$ & the defender's deploy budget\\\hline
$U_a$ & the attacker's total expected utility gain\\\hline
$\Delta S$ & the sensitivity of a query \\\hline
$\Delta u$ & the sensitivity of utility \\\hline
$\epsilon$ & privacy budget \\\hline
\end{tabular}
	\label{tab:notation}
\vspace{-2mm}
\end{table}
\vspace{-2mm}
\section{The differentially private approach}\label{sec:method}
\subsection{Overview of the approach}
We provide an example here to describe the workflow of our approach.
In a network, there are three systems $N=\{s_1,s_2,s_3\}$ and three configurations $K=\{k_1,k_2,k_3\}$.
System $1$, $2$ and $3$ are associated with configuration $1$, $2$ and $3$, respectively:
$N_1=\{s_1\}$, $N_2=\{s_2\}$ and $N_3=\{s_3\}$.
Each configuration is associated with a utility,
such that systems with the same configuration have the same utility.
Once a system is attacked, the defender loses a corresponding utility,
while the attacker gains this utility.
The aim of our approach is to minimize the defender's utility loss
by using the differential privacy technique to hide the real configurations of systems.

At each round of the game, our approach consists of four steps:
Steps 1 and 2 are carried out by the defender,
while Steps 3 and 4 are conducted by the attacker.

\emph{Step 1}: The defender obfuscates the configurations. 
The obfuscation is conducted on the statistical information of the network according to Algorithm \ref{alg:DP}.
In this example, the number of systems in each configuration is $1$, i.e., $|N_1|=|N_2|=|N_3|=1$,
while the total number of systems in the network is $|N|=3$.
After obfuscation, the number of systems in each configuration could be: $|N'_1|=2$, $|N'_2|=1$, $|N'_3|=1$,
while the total number of systems in the network becomes $|N'|=|N'_1|+|N'_2|+|N'_3|=4$.

\emph{Step 2}: As $|N'|-|N|=1$, an extra system is introduced into the network.
This extra system is interpreted as a honeypot, denoted as $hp_4$.
The defender uses Algorithm \ref{alg:deployment2} to deploy the systems and the honeypot to the configurations.
The deployment result could be: $N'_1=\{s_3,hp_4\}$, $N'_2=\{s_1\}$ and $N'_3=\{s_2\}$,
which is shown to the attacker as demonstrated in Fig. \ref{fig:obfuscation}.
\begin{figure}[ht]
\vspace{-3mm}
	\includegraphics[width=0.49\textwidth, height=2.5cm]{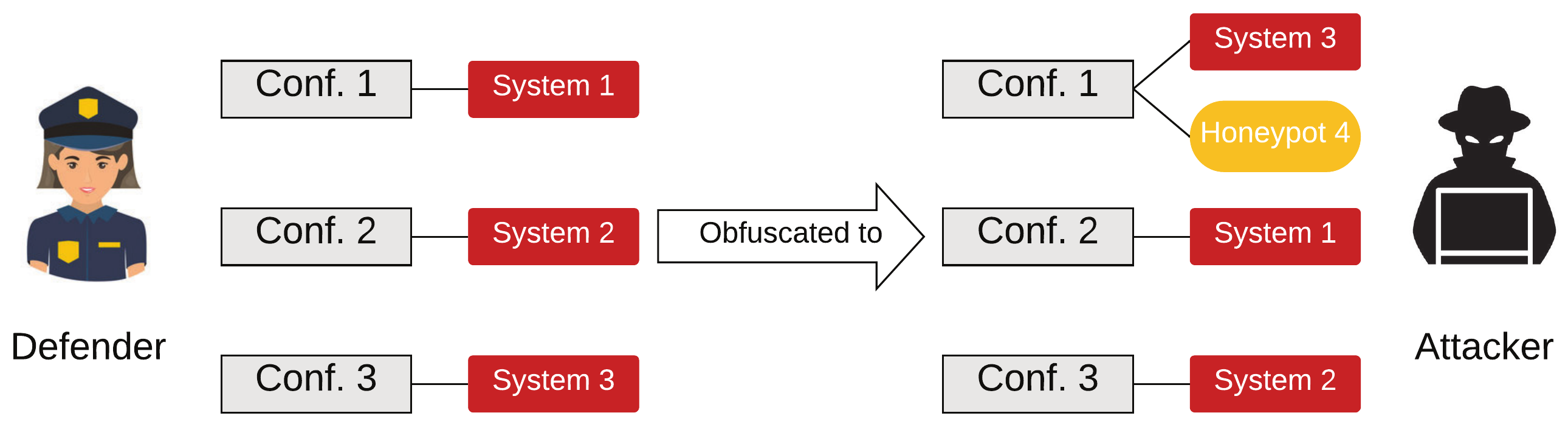}
	\caption{The defender obfuscates configurations of systems}
	\label{fig:obfuscation}
\vspace{-2mm}
\end{figure}


\emph{Step 3}: The attacker estimates the probability with which an observed configuration $k'$ could be a real configuration $k$ using Bayesian inference (Equation \ref{eq:Bayes}),
i.e., estimating:

$q(k_1|k_1)$, $q(k_2|k_1)$ and $q(k_3|k_1)$;

$q(k_1|k_2)$, $q(k_2|k_2)$ and $q(k_3|k_2)$;

$q(k_1|k_3)$, $q(k_2|k_3)$ and $q(k_3|k_3)$.

\emph{Step 4}: Based on the estimation, the attacker uses Equation \ref{eq:utilitygain} to calculate the expected utility gain of selecting each configuration.
Finally, the attacker selects a specific configuration as the target
based on a probability distribution over configurations (Equation \ref{eq:selection}).

\subsection{The defender's strategy}
At each round of the game, the defender first obfuscates the configurations using the differentially private Laplace mechanism,
and deploys systems according to these configurations.

\subsubsection{Step 1: obfuscating configurations}
The obfuscation is described in Algorithm \ref{alg:DP}.
In Line 5, Laplace noise is added to each $|N_k|$ to obfuscate the number of systems with configuration $k$.
In Line 5, $\Delta S$ denotes the sensitivity of the maximum number of systems with a specific configuration.
As this maximum number has a direct impact on the configuration arrangement of these systems,
the sensitivity $\Delta S$ is determined by the maximum number.
Based on the definition of sensitivity, $\Delta S=1$ in this algorithm.

The rationale underpinning Algorithm \ref{alg:DP} is as follows.
Differential privacy is originally designed to preserve data privacy in datasets.
It can guarantee that an individual data record in or out of a dataset has little impact on the analytical output of the dataset.
In other words, an attacker cannot infer whether a data record is in the dataset by making queries to the dataset.
Here we use differential privacy to preserve the privacy of systems in networks.
The privacy of systems in this paper means the configurations of systems.
To map from differential privacy to its meaning in the network configuration,
we treat a network as a dataset and treat a system in the network as a data record in the dataset.
Since differential privacy can guarantee that
an attacker cannot infer whether a given data record is in a dataset,
it can also guarantee that an attacker cannot infer whether a given system is in a network.
Specifically, as we add differentially private Laplace noise on the number of systems in each configuration,
the attacker cannot infer whether a given system is really associated with the shown configuration.
Thus, the system's real configuration is preserved.

\vspace{-2mm}
\begin{algorithm}
\caption{The obfuscation of configurations}
\label{alg:DP}
\textbf{Input}: $|N_1|,...,|N_{|K|}|$;\\
Partition the system set $N$ into $N_1,...,N_{|K|}$ based on the configurations;\\
Calculate the size of each subset: $|N_1|,...,|N_{|K|}|$;\\
\For{$k=1$ to $|K|$}{
    $|N'_k|\leftarrow |N_k|+\lceil Lap(\frac{\Delta S\cdot |K|}{\epsilon})\rceil$;\\
}
\textbf{Output}: $|N'_1|,...,|N'_{|K|}|$;\\
\end{algorithm}
\vspace{-2mm}
After Algorithm \ref{alg:DP} is executed,
the defender receives an output: $|N'_1|,...,|N'_{|K|}|$.
Let $|N'|=\sum_{1\leq k\leq |K|}|N'_k|$,
so that there are three possible situations: $|N'|=|N|$, $|N'|>|N|$ or $|N'|<|N|$.
When $|N'|>|N|$, there are $|N'|-|N|$ void systems,
which can be interpreted as honeypots.
When $|N'|<|N|$, there are $|N|-|N'|$ systems taken offline.

\subsubsection{Step 2: deploying systems}
System deployment will be conducted according to three possible situations: $|N'|=|N|$, $|N'|>|N|$ or $|N'|<|N|$.
This deployment is based not only on the utility of these systems,
but also on the attacker's attack strategy in the previous rounds.
The aim of this deployment is to minimize the defender's expected utility loss
while satisfying the defender's budget constraint $B_d$.

\emph{Situation 1}: When $|N'|=|N|$, the defender's expected utility loss is shown in Equation \ref{eq:utilityloss1},
\begin{equation}\label{eq:utilityloss1}
  U_d=\sum_{1\leq k\leq |K|}(p(k)\cdot\sum_{1\leq i\leq |N_k|} u_i),
\end{equation}
where $p(k)$ is the estimated probability that the attacker will attack the systems with configuration $k$ in the next round.
The estimated probability is the ratio between the number of times that $k$ was attacked to the total number of game rounds.
For example, if the game has been played for $10$ rounds and systems with configuration $k$ have been attacked $3$ times,
the estimated probability is $p(k)=\frac{3}{10}=0.3$.
Initially, before the first round, the probability distribution over the configurations is uniform.
The probability distribution is then updated every $10$ rounds using the method described in the example.
Since the game is an imperfect information game,
the defender is unaware of the real attack probability of the attacker.
The defender can use only her past experience to make a strategy, i.e., past observation of the attacker’s strategies.

During the deployment, the defender's cost is shown in Equation \ref{eq:dcost1},
\begin{equation}\label{eq:dcost1}
  C_d=\sum_{1\leq i\leq |N|}I_i\cdot(c(k',k)),
\end{equation}
where $I_i=1$ if system $i$ is obfuscated from configuration $k$ to be perceived as $k'$, and $I_i=0$ otherwise.

The problem faced by the defender is to identify the deployment configuration
that can minimize her utility loss $U_d$ while satisfying her budget constraint: $C_d\leq B_d$.
Algorithm \ref{alg:deployment1} is developed to solve this problem.

\begin{algorithm}
\caption{Deployment of configurations, $|N'|=|N|$}
\label{alg:deployment1}
\textbf{Input}: $N$ and $|N'_1|,...,|N'_{|K|}|$;\\
Let $|N'_1|=n'_1$, ..., $|N'_{|K|}|=n'_{|K|}$;\\
Initialize $N'_1=...=N'_{|K|}=\emptyset$;\\
Initialize $C_d=0$;\\
Rank $u_1$, ..., $u_{|N|}$ in an increasing order, and supposing that the result is $u_1\leq...\leq u_{|N|}$;\\
\For{$i=1$ to $|N|$}{
    Select a configuration $k'$ with probability proportional to $exp(\frac{\frac{\epsilon}{|N|} U_{k'}}{2\Delta U_d})$;\\
    \If{$C_d+c(k',k)>B_d$}{
        \textbf{continue};\\
    }
    \If{$|N'_{k'}|<n'_{k'}$}{
        $N'_{k'}\leftarrow N'_{k'}\cup \{i\}$;\\
        $C_d\leftarrow C_d+c(k',k)$;\\
    }\Else{
        \textbf{goto} Line 7;\\
    }
}
\textbf{Output}: $N'_1,...,N'_{|K|}$;\\
\end{algorithm}

The idea behind Algorithm \ref{alg:deployment1} involves preferentially obfuscating low-utility systems, with specific probabilities,
to be perceived as those configurations that have a high probability of being attacked.
To a large extent, this idea can prevent high-utility systems from being attacked.
It seems that the high-utility systems may not be protected by the exponential mechanism,
because 1) the exponential mechanism preferentially obfuscates low-utility systems
and 2) the budget constraint $B_d$ may limit the number of obfuscated systems.
However, the exponential mechanism still gives high-utility systems the probability to be obfuscated,
since the exponential mechanism is probabilistic rather than deterministic.
In addition, the budget constraint is a parameter which is set by users.
A larger budget constraint means a larger number of obfuscated systems.

The algorithm starts by ranking the utility of the systems in increasing order (Line 5).
Next, each of the ranked systems $i$ is obfuscated to appear as a configuration, $k'$,
selected using the exponential mechanism,
until all systems are obfuscated or the defender's budget is used up (Lines 6-14).

In the exponential mechanism in Line 7,
$U_{k'}=p(k')\cdot\sum_{1\leq i\leq |N_{k'}|} u_i$ is the defender's expected utility loss on the systems with configuration $k'$.
Moreover, $\Delta U_d=max_{1\leq i\leq |N|}u_i$ is the sensitivity of the defender's expected utility loss,
which is used to to obfuscate system configurations. 
A system configuration with a higher expected utility loss has a larger probability to be obfuscated.
The expected utility loss is used only as a parameter in the exponential mechanism,
and it does not change any properties of the exponential mechanism.
Therefore, the use of the expected utility loss does not violate differential privacy.

The rationale of using the exponential mechanism is described as follows.
According to the definitions of differential privacy,
an obfuscated network $N'$ is interpreted as a dataset $D$,
while a configuration $k'$ in network $N'$ is interpreted as an output $r$ of dataset $D$.
Thus, the utility of a pair of dataset and output, $u(D,r)$,
is equivalent to the defender's expected utility loss of selecting configuration $k'$ from network $N'$, i.e., $U(N',k')$ denoted as $U_{k'}$.
By comparing Definition \ref{Def-Ex} to Line 7 in Algorithm \ref{alg:deployment1},
we can conclude that as the exponential mechanism can guarantee the privacy of the output of a dataset,
it can also guarantee the privacy of the configurations of the systems in a network.

\emph{Situation 2}. When $|N'|>|N|$, there are $|N'|-|N|$ honeypots in the network.
The defender's expected utility loss is shown in Equation \ref{eq:utilityloss2},
\begin{equation}\label{eq:utilityloss2}
  U_d=\sum_{1\leq k\leq |K|}(p(k)\cdot\sum_{1\leq i\leq |N_k|} u_i)-\sum_{1\leq j\leq |N'|-|N|} p(k)\cdot u_j,
\end{equation}
where the second part, $\sum_{1\leq j\leq |N'|-|N|} p(k)\cdot u_j$, indicates the expected utility gain obtained by using honeypots.

The defender's cost is shown in Equation \ref{eq:dcost2},
\begin{equation}\label{eq:dcost2}
  C_d=\sum_{1\leq i\leq |N|}I_i\cdot(c(k',k))+\sum_{1\leq j\leq |N'|-|N|}I_j\cdot h_{k'},
\end{equation}
where the second part, $\sum_{1\leq j\leq |N'|-|N|}I_j\cdot h_{k'}$, indicates the extra cost incurred by deploying honeypots.
Moreover, the privacy budget $\epsilon$ is proportionally increased to $\frac{|N'|}{|N|}\cdot\epsilon$ to cover the extra honeypots.
We use this proportional change of privacy budget, because in Algorithm \ref{alg:deployment2},
the amount of privacy budget is consumed identically in each iteration.

\begin{algorithm}
\caption{Deployment of configurations, $|N'|>|N|$}
\label{alg:deployment2}
\textbf{Input}: $N$ and $|N'_1|,...,|N'_{|K|}|$;\\
Let $|N'_1|=n'_1$, ..., $|N'_{|K|}|=n'_{|K|}$;\\
Let $x_k=p(k)\cdot\sum_{1\leq i\leq |N_k|} u_i$;\\
Initialize $N'_1=...=N'_{|K|}=\emptyset$;\\
Initialize $C_d=0$;\\
Rank $u_1$, ..., $u_{|N|}$ in increasing order, and supposing that the result is $u_1\leq...\leq u_{|N|}$;\\
\For{$i=1$ to $|N'|-|N|$}{
    Select a configuration $k'$ with probability proportional to $exp(\frac{\frac{\epsilon}{|N'|-|N|} U_{k'}}{2\Delta U_d})$;\\
    $C_d\leftarrow C_d+h_{k'}$;\\
    \If{$C_d\leq B_d$}{
        Create a honeypot $i$ with configuration $k'$;\\
        $N'_{k'}\leftarrow N'_{k'}\cup\{i\}$;\\
    }\Else{
        \textbf{goto} Line 8;\\
    }
}
\For{$i=1$ to $|N|$}{
    Select a configuration $k'$ with probability proportional to $exp(\frac{\frac{\epsilon}{|N|} U_{k'}}{2\Delta U_d})$;\\
    \If{$C_d+c(k',k)>B_d$}{
        \textbf{continue};\\
    }
    \If{$|N'_{k'}|<n'_{k'}$}{
        $N'_{k'}\leftarrow N'_{k'}\cup \{i\}$;\\
        $C_d\leftarrow C_d+c(k',k)$;\\
    }\Else{
        \textbf{goto} Line 16;\\
    }
}
\textbf{Output}: $N'_1,...,N'_{|K|}$;\\
\end{algorithm}

The problem faced by the defender is the same as in \emph{Situation 1}.
Algorithm \ref{alg:deployment2} is developed to solve this problem.
In Algorithm \ref{alg:deployment2}, the defender first creates honeypots (Lines 7-14)
and then obfuscates systems (Lines 15-23).
In Line 8, $U_{k'}=p(k')u_{k'}$ while in Line 16, $U_{k'}=p(k')\cdot\sum_{1\leq i\leq |N_{k'}|} u_i$.
In both lines, $\Delta U_d=max_{1\leq i\leq |N|}u_i$.
The defender favors honeypots over obfuscation,
as the use of the former may result in her incurring a utility gain.
Honeypots are created based on the estimated probability with which a configuration will be attacked.
Configurations with the highest estimated probability of being attacked will take priority for use to configure honeypots.

\emph{Situation 3}. When $|N'|<|N|$, there are $|N|-|N'|$ systems taken offline.
The defender's expected utility loss is shown in Equation \ref{eq:utilityloss3},
\begin{equation}\label{eq:utilityloss3}
  U_d=\sum_{1\leq k\leq |K|}\sum_{1\leq i\leq |N_k|}[(p(k)\cdot J_i\cdot u_i)+(J_i-1)^2\cdot l_k],
\end{equation}
where $J_i=0$ if system $i$ is taken offline, and $J_i=1$ otherwise.
The second part, $\sum_{1\leq i\leq |N|-|N'|} (J_i-1)^2\cdot l_k$, indicates the extra utility loss incurred by taking systems offline.
The defender's cost is shown in Equation \ref{eq:dcost3},
\begin{equation}\label{eq:dcost3}
  C_d=\sum_{1\leq i\leq |N|}I_i\cdot(c(k',k)).
\end{equation}
Moreover, akin to \emph{Situation 2}, the privacy budget $\epsilon$ is proportionally decreased to $\frac{|N'|}{|N|}\cdot\epsilon$,
as the systems taken offline need not be protected.
Thus, this part of the privacy budget can be conserved.

The problem for the defender is the same as in \emph{Situation 1}.
Algorithm \ref{alg:deployment3} is developed to solve this problem.
In Algorithm \ref{alg:deployment3}, the defender first takes $|N|-|N'|$ systems offline (Lines 7-9)
and next obfuscates the remaining systems (Lines 10-18).
In Line 11, $U_{k'}=\sum_{1\leq i\leq |N_k|}[(p(k)\cdot J_i\cdot u_i)+(J_i-1)^2\cdot l_k]$
and $\Delta U_d=max_{1\leq i\leq |N|}u_i$.
The defender will be inclined to take those systems offline
that will incur the highest utility loss if attacked.

\begin{algorithm}
\caption{Deployment of configurations, $|N'|<|N|$}
\label{alg:deployment3}
\textbf{Input}: $N$ and $|N'_1|,...,|N'_{|K|}|$;\\
Let $|N'_1|=n'_1$, ..., $|N'_{|K|}|=n'_{|K|}$;\\
Let $x_k=p(k)\cdot\sum_{1\leq i\leq |N_k|} u_i$;\\
Initialize $N'_1=...=N'_{|K|}=\emptyset$;\\
Initialize $C_d=0$;\\
Rank $u_1$, ..., $u_{|N|}$ in increasing order, and supposing that the result is $u_1\leq...\leq u_{|N|}$;\\
\For{$j=1$ to $|N|-|N'|$}{
    $s\leftarrow argMax_{i} [p(k)\cdot u_i-l_k]$;\\
    Take system $s$ offline;\\
}
\For{$i=1$ to $|N'|$}{
    Select a configuration $k'$ with probability proportional to $exp(\frac{\frac{\epsilon}{|N'|} U_{k'}}{2\Delta U_d})$;\\
    \If{$C_d+c(k',k)>B_d$}{
        \textbf{continue};\\
    }
    \If{$|N'_{k'}|<n'_{k'}$}{
        $N'_{k'}\leftarrow N'_{k'}\cup \{i\}$;\\
        $C_d\leftarrow C_d+c(k',k)$;\\
    }\Else{
        \textbf{goto} Line 11;\\
    }
}
\textbf{Output}: $N'_1,...,N'_{|K|}$;\\
\end{algorithm}

\subsection{The attacker's strategy}\label{sub:Bayesian}
In each round of the game, the attacker first estimates the probability
with which an observed configuration $k'$ is a real configuration $k$ using Bayesian inference.
Next, he decides which observed configuration $k'$ should be attacked,
i.e., he will attack those systems associated with configuration $k'$.

\subsubsection{Step 3: probability estimation}
To estimate the probability that an observed configuration $k'$ is a real configuration $k$: $q(k|k')$,
the attacker uses Bayesian inference:

\begin{equation}\label{eq:Bayes}
  q(k|k')=\frac{q(k)\cdot q(k'|k)}{q(k')}.
\end{equation}

In Equation \ref{eq:Bayes}, $q(k)=\frac{|N_k|}{|N|}$, and $q(k')=\frac{|N'_{k'}|}{|N'|}$,
where $q(k)$ is calculated using the attacker's prior knowledge
and $q(k')$ is computed with reference to the attacker's current observations.
In addition, $q(k'|k)$ can be obtained from the attacker's experience, i.e., posterior knowledge.
For example, in previous game rounds, the attacker attacked $10$ systems with configuration $k$.
Among these $10$ systems, four systems are obfuscated to appear as configuration $k'$.
Thus, $q(k'|k)=\frac{4}{10}=0.4$.

\subsubsection{Step 4: target selection}
Based on Equation \ref{eq:Bayes}, the attacker can calculate the expected utility gain of attacking each configuration $k'$:
\begin{equation}\label{eq:utilitygain}
  U^{(k')}_a=\sum_{1\leq k\leq |K|}[q(k|k')\cdot u_k].
\end{equation}
According to Equation \ref{eq:utilitygain}, the attacker uses an $epsilon$-greedy strategy to distribute selection probabilities over the observed configurations:
\begin{equation}\label{eq:selection}
  \mathcal{P}_{k'}=\begin{cases}(1-e)+\frac{e}{|K|}, \mbox{if $U^{(k')}_a$ is the largest}\\
		  \frac{e}{|K|}, \mbox{otherwise}\end{cases}.
\end{equation}
where $e$ is a small positive number in $[0,1]$.
The reason for using the $epsilon$-greedy strategy will be described in \textbf{Remark 2} in Section \ref{sub:attacker analysis}.
Finally, the attacker selects a target $k'$ based on the probability distribution $\boldsymbol{\mathcal{P}}=\{\mathcal{P}_{1},...,\mathcal{P}_{|K|}\}$.

\subsection{Potential application of our approach}
Our approach can be applied against powerful attackers in the real world.
An attacker may execute a suite of scan queries on a network using tools such as NMAP \cite{Lyon09}.
Our DP-based approach returns a mix of true and false results to the attacker’s scan to confuse him.
The false results include not only obfuscated configurations of systems but also additional fake systems, i.e., honey pots.
Moreover, our approach can also take vulnerable systems offline to avoid the attacker’s scan.
For example, if an attacker makes a scan query on one system many times,
he may combine the query results to reason the true configuration of this system.
Thus, this system should be taken offline temporarily.
In summary, our approach can significantly increase the complexity on attackers
to formulate their attack, irrespective of their reasoning power,
so that administrators can have additional time to build defense policies to interdict attackers.

A specific real-world example is an advanced persistent threat (APT) attacker \cite{Tankard11}.
Our approach can be adopted to model interactions between a system manager and an APT attacker.
Typically, an APT attacker continuously scans the vulnerability of the target systems
and studies the defense policy of these systems
so as to steal sensitive information from these systems \cite{Yang19}.
Accordingly, the system manager misleads the APT attacker by obfuscating the true information of these systems,
and formulates a complex defense policy for these systems
which is hardly predictable to the APT attacker.

\section{Theoretical analysis}\label{sec:analysis}
Our theoretical analysis consists of two parts:
analysis of the defender's strategy and analysis of the attacker's strategy.
In the analysis of the defender's strategy, we first analyze the privacy preservation of the defender's strategy.
Then, we analyze the defender's expected utility loss in different situations.
After that, the complexity of the defender's strategy is also analyzed.
In the analysis of the attacker's strategy, we analyze the attacker's optimal strategy and the lower bound of his expected utility gain.
Finally, we give a discussion on the equalibria between the defender and the attacker.

\subsection{Analysis of the defender's strategy}
\subsubsection{Analysis of privacy preservation}
We first prove that the defender's strategy satisfies differential privacy.
Then, we give an upper bound on the number of rounds of the game.
Exceeding this bound, the privacy level of the defender cannot be guaranteed.

\begin{lem}[Sequential Composition Theorem \cite{McSherry200794}]\label{comp1}
Suppose a set of privacy steps
$\mathcal{M}=\{\mathcal{M}_1,...,\mathcal{M}_m\}$
are sequentially performed on a dataset,
and each $\mathcal{M}_i$ provides $\epsilon_i$ privacy guarantee,
then $\mathcal{M}$ will provide
$\sum_{1\leq i\leq m}\epsilon_i$-\emph{differential privacy}.
\end{lem}

\begin{lem}[Parallel Composition Theorem \cite{McSherry200794}]\label{comp2}
Suppose we have a set of privacy step
$\mathcal{M}=\{\mathcal{M}_1,...,\mathcal{M}_m\}$,
if each $\mathcal{M}_i$ provides $\epsilon_{i}$ privacy guarantee on
a disjoint subset of the entire dataset,
then the parallel composition of $\mathcal{M}$ will provide
$\max_{1\leq i\leq m}\epsilon_{i}$-\emph{differential privacy}.
\end{lem}

\begin{thm}\label{thm:DP1}
Algorithm \ref{alg:DP} satisfies $\epsilon$-differential privacy.
\end{thm}
\begin{proof}
Algorithm \ref{alg:DP} consumes a privacy budget $\epsilon$.
The Laplace noise sampled from $Lap(\frac{\Delta S\cdot |K|}{\epsilon})$ is added to $|K|$ steps.
At each step, the average allocated privacy budget is $\frac{\epsilon}{|K|}$.
Thus, for each step, Algorithm \ref{alg:DP} satisfies $\frac{\epsilon}{|K|}$-differential privacy.
As there are $|K|$ steps, based on Lemma \ref{comp1},
Algorithm \ref{alg:DP} satisfies $\epsilon$-differential privacy.
\end{proof}

\begin{thm}\label{thm:DP2}
Algorithms \ref{alg:deployment1}, \ref{alg:deployment2} and \ref{alg:deployment3} satisfy $\epsilon$-differential privacy.
\end{thm}
\begin{proof}
We prove that Algorithm \ref{alg:deployment1} satisfies $\epsilon$-differential privacy.
The proof of Algorithms \ref{alg:deployment2} and \ref{alg:deployment3} is similar.

In Line 7 of Algorithm \ref{alg:deployment1}, the selection of a configuration takes place $|N|$ times.
The privacy budget $\epsilon$ is, thus, consumed in $|N|$ steps.
At each step, the average allocated privacy budget is $\frac{\epsilon}{|N|}$.
Hence, for each step, Algorithm \ref{alg:deployment1} satisfies $\frac{\epsilon}{|N|}$-differential privacy.
As there are $|N|$ steps, based on the sequential composition theorem (Lemma \ref{comp1}),
Algorithm \ref{alg:deployment1} satisfies $\epsilon$-differential privacy.
\end{proof}

The processes of Algorithms \ref{alg:deployment1}, \ref{alg:deployment2} and \ref{alg:deployment3} are similar to the NoisyAverage sampling developed in McSherry’s PINQ framework \cite{McSherry09}.
They, however, are different.
The output value of the NoisyAverage sampling is: $Pr[NoisyAvg(A)=x]\propto\max\limits_{avg(B)=x}exp(-\epsilon\times\frac{|A\oplus B|}{2})$,
where $A$ and $B$ are two datasets and $\oplus$ is for symmetric difference.
In our algorithms, the output value is: $Pr[output=k']\propto exp(\frac{\frac{\epsilon}{|N|}U_{k'}}{2\Delta U_d})$.
The major difference between the NoisyAverage sampling and ours is that
in the NoisyAverage sampling, the output of a value $x$ is based on an operation of $\max\limits_{avg(B)=x}$,
while in our algorithm, the output of a value $k'$ is based on its utility $U_{k'}$.
The operation of $\max\limits_{avg(B)=x}$ may violate the definition of the exponential mechanism and
thus lead the NoisyAverage sampling to fail to satisfy differential privacy.
By comparison, our algorithms strictly follow the definition of the exponential mechanism (ref. page 39, Definition 3.4 in \cite{Dwork14}).
Thus, our algorithms satisfy differential privacy.

\begin{thm}\label{thm:DP}
The defender is guaranteed $\epsilon$-differential privacy.
\end{thm}
\begin{proof}
The defender uses Algorithms \ref{alg:deployment1}, \ref{alg:deployment2} and \ref{alg:deployment3} for deployment of configurations.
The three algorithms are disjoint
because they are applied to three mutually exclusive situations:
$|N'|=|N|$, $|N'|>|N|$ and $|N'|<|N|$.
Therefore, the parallel composition theorem (Lemma \ref{comp2}) can be used here.
Since all of these algorithms satisfy $\epsilon$-differential privacy,
the defender is guaranteed $\epsilon$-differential privacy.
\end{proof}

\begin{cor}\label{cor:strong}
The attacker cannot deduce the exact configurations of the systems.
\end{cor}
\begin{proof}
According to Theorem \ref{thm:DP}, the defender is guaranteed $\epsilon$-differential privacy.
Based on the description of differential privacy in Section \ref{sub:DP},
differential privacy guarantees that the attacker cannot tell whether a particular system is associated with configuration $k$,
because whether or not the system is associated with configuration $k$ makes little difference to the output observations.
Since the attacker cannot deduce the configuration of each individual system,
he cannot deduce the exact configurations of the systems in the network.
\end{proof}

\textbf{Remark 1.} Corollary \ref{cor:strong} guarantees the privacy of the defender in a single game round.
However, if no bound on the number of rounds is set,
the defender's privacy may be leaked \cite{Dwork14},
i.e., that the attacker figures out the real configurations of the systems.
This is because in differential privacy,
a privacy budget is used to control the privacy level.
Every time the system configuration is obfuscated and released,
the privacy budget is partially consumed.
Once the privacy budget is used up,
differential privacy cannot guarantee the privacy of the defender anymore.
To guarantee the privacy level of the defender,
we need to set a bound on the number of rounds.

\begin{definition}[KL-Divergence \cite{Dwork14}]\label{Def:KL}
The KL-Divergence between two random variables $Y$ and $Z$ taking values from the same domain is defined to be:
\begin{equation}
D(Y||Z)=\mathbb{E}_{y\sim Y}\left[ln\frac{Pr(Y=y)}{Pr(Z=y)}\right].
\end{equation}
\end{definition}

\begin{definition}[Max Divergence \cite{Dwork14}]\label{Def:Max}
The Max Divergence between two random variables $Y$ and $Z$ taking values from the same domain is defined to be:
\begin{equation}
D_{\infty}(Y||Z)=\max\limits_{S\subseteq Supp(Y)}\left[ln\frac{Pr(Y\in S)}{Pr(Z\in S)}\right].
\end{equation}
\end{definition}

\begin{lem}[\cite{Dwork14}]\label{lem:DPDivergence}
A mechanism $\mathcal{M}$ is $\epsilon$-differentially private
if and only if on every two neighboring datasets $x$ and $x'$,
$D_{\infty}(\mathcal{M}(x)||\mathcal{M}(x'))\leq\epsilon$
and $D_{\infty}(\mathcal{M}(x')||\mathcal{M}(x))\leq\epsilon$.
\end{lem}

\begin{lem}[\cite{Dwork14}]\label{lem:KLMax}
Suppose that random variables $Y$ and $Z$ satisfy $D_{\infty}(Y||Z)\leq\epsilon$
and $D_{\infty}(Z||Y)\leq\epsilon$.
Then, $D(Y||Z)\leq\epsilon\cdot (e^{\epsilon}-1)$.
\end{lem}

\begin{thm}\label{thm:bound}
Given that the defender's privacy level is $\epsilon$ at each single round,
to guarantee the defender's overall privacy level to be $\epsilon'$,
the upper bound of the number of rounds is $\frac{\epsilon '\cdot(e^{\epsilon'}-1)}{\epsilon\cdot(e^{\epsilon}-1)}$.
\end{thm}
\begin{proof}
Let the upper bound of the number of rounds be $k$ and
the view of the attacker in the $k$ rounds be $v=(v_1,...,v_k)$.
We have $D(Y||Z)=ln\left[\frac{Pr(Y=v)}{Pr(Z=v)}\right]=ln\left[\prod^k_{i=1}\frac{Pr(Y_i=v_i)}{Pr(Z_i=v_i)}\right]=\sum^k_{i=1}ln\left[\frac{Pr(Y_i=v_i)}{Pr(Z_i=v_i)}\right]=\sum^k_{i=1}D(Y_i||Z_i)$.
As the defender is guaranteed $\epsilon$-differential privacy at each single round,
based on Lemma \ref{lem:DPDivergence},
we have $D_{\infty}(Y_i||Z_i)\leq\epsilon$ and $D_{\infty}(Z_i||Y_i)\leq\epsilon$.
Based on this result, according to Lemma \ref{lem:KLMax},
we have $D(Y_i||Z_i)\leq\epsilon\cdot(e^{\epsilon}-1)$.
Thus, we have $D(Y||Z)=\sum^k_{i=1}D(Y_i||Z_i)\leq k\cdot\epsilon\cdot(e^{\epsilon}-1)$.

To guarantee the defender's overall privacy level to be $\epsilon '$,
according to Lemma \ref{lem:DPDivergence},
we have $D_{\infty}(Y||Z)\leq\epsilon'$ and $D_{\infty}(Z||Y)\leq\epsilon'$.
By using Lemma \ref{lem:KLMax},
we have $D(Y||Z)\leq\epsilon'\cdot(e^{\epsilon'}-1)$.
Finally, since $D(Y||Z)\leq k\cdot\epsilon\cdot(e^{\epsilon}-1)$ and $D(Y||Z)\leq\epsilon'\cdot(e^{\epsilon'}-1)$,
the upper bound $k$ is limited by $\frac{\epsilon'\cdot(e^{\epsilon'}-1)}{\epsilon\cdot(e^{\epsilon}-1)}$.
\end{proof}

\subsubsection{Analysis of the deployment algorithms}
Here, we compute the expected utility loss of the defender in three situations:
$|N'|=|N|$, $|N'|>|N|$ and $|N'|<|N|$.

\begin{thm}\label{thm:DUtility1}
In Situation 1 ($|N'|=|N|$), the defender's expected utility loss is $\sum_{1\leq i\leq |N|}[u_i\cdot\sum_{1\leq k'\leq |K|}p_{k'}\mathcal{P}_{k'}]$,
where $p_{k'}=\frac{exp(\frac{\frac{\epsilon}{|N|} U_{k'}}{2\Delta U_d})}{\sum_{1\leq k\leq |K|}exp(\frac{\frac{\epsilon}{|N|} U_{k}}{2\Delta U_d})}$,
if the defender uses Algorithm \ref{alg:deployment1} to deploy configurations
and the attacker uses Equation \ref{eq:selection} to select the target.
\end{thm}
\begin{proof}
In Algorithm \ref{alg:deployment1}, the defender deploys the systems based on 1) the utility of each system
and 2) the estimated probability that the attacker will select each configuration to attack. 
By using the exponential mechanism,
each system $i$ is obfuscated to appear as configuration $k'$ with probability
$\frac{exp(\frac{\frac{\epsilon}{|N|} U_{k'}}{2\Delta U_d})}{\sum_{1\leq k\leq |K|}exp(\frac{\frac{\epsilon}{|N|} U_{k}}{2\Delta U_d})}$.
The attacker thus selects configuration $k'$ as a target with probability $\mathcal{P}_{k'}$.

System $i$ is attacked only when the defender obfuscates $i$ to appear as configuration $k'$ and
the attacker selects $k'$ as a target.
As there are $|K|$ configurations,
the defender's expected utility loss on system $i$ is $u_i\cdot\sum_{1\leq k'\leq |K|}p_{k'}\mathcal{P}_{k'}$.
Since there are $|N|$ systems,
the defender's total expected utility loss is $\sum_{1\leq i\leq |N|}[u_i\cdot\sum_{1\leq k'\leq |K|}p_{k'}\mathcal{P}_{k'}]$.
\end{proof}

\begin{cor}\label{cor:DUtility1}
In Situation 1, the upper bound of the defender's utility loss is $\sum_{1\leq i\leq |N'_l|}u_i$,
where $|N'_l|=max_{1\leq k'\leq |K|}|N'_{k'}|$. \\
$\sum_{1\leq i\leq |N'_l|}u_i$ represents the utility sum of $|N'_l|$ systems, which have the highest utility among all the systems.
\end{cor}
\begin{proof}
In Algorithm \ref{alg:deployment1}, the defender deploys each system in order of increasing utility.
Therefore, it is possible for the defender to deploy the $|N'_l|$ systems in configuration $l$,
where $|N'_l|=max_{1\leq k'\leq |K|}|N'_{k'}|$ and the $|N'_l|$ systems have the highest utility among all the systems.
When configuration $l$ is selected by the attacker as a target,
the utility of all systems in $N'_l$ is lost;
this amounts to $\sum_{1\leq i\leq |N'_l|}u_i$.
As this is the worst case for the defender in Situation 1,
the utility loss in this case is the upper bound.
\end{proof}

Similarly, we can also draw the following conclusions.

\begin{thm}\label{thm:DUtility2}
In Situation 2 ($|N'|>|N|$), the defender's expected utility loss is $\sum_{1\leq i\leq |N'|}[\frac{2|N|-|N'|}{|N'|}(u_i\cdot\sum_{1\leq k'\leq |K|}p_{k'}\mathcal{P}_{k'})]$,
where $p_{k'}=\frac{exp(\frac{\frac{\epsilon}{|N|} U_{k'}}{2\Delta U_d})}{\sum_{1\leq k\leq |K|}exp(\frac{\frac{\epsilon}{|N|} U_{k}}{2\Delta U_d})}$,
if the defender uses Algorithm \ref{alg:deployment2} to deploy configurations
and the attacker uses Equation \ref{eq:selection} to select the target.
\end{thm}
\begin{proof}
By comparing the results of Theorems \ref{thm:DUtility1} and \ref{thm:DUtility2},
the difference is that in Theorem \ref{thm:DUtility2}, there is an extra coefficient, $\frac{2|N|-|N'|}{|N'|}$.
In Situation 2, there are $|N'|-|N|$ honeypots.
Hence, a system $i$ could be a honeypot, with probability $\frac{|N'|-|N|}{|N'|}$,
or a genuine system, with probability $\frac{|N|}{|N'|}$.
When system $i$ is attacked,
the expected utility loss on system $i$ is $\frac{|N|}{|N'|}[u_i\cdot\sum_{1\leq k'\leq |K|}p(k')\mathcal{P}(k')]-\frac{|N'|-|N|}{|N'|}(u_i\cdot\sum_{1\leq k'\leq |K|}p_{k'}\mathcal{P}_{k'})$,
where $\frac{|N'|-|N|}{|N'|}(u_i\cdot\sum_{1\leq k'\leq |K|}p_{k'}\mathcal{P}_{k'})$ means a utility gain if system $i$ is a honeypot.
As there are $|N'|$ systems, including honeypots,
the expected utility loss is\\ $\sum_{1\leq i\leq |N'|}[\frac{2|N|-|N'|}{|N'|}(u_i\cdot\sum_{1\leq k'\leq |K|}p_{k'}\mathcal{P}_{k'})]$.
\end{proof}

\begin{cor}\label{cor:DUtility2}
In Situation 2, the upper bound of the defender's utility loss is $\sum_{1\leq i\leq |N'_l|}u_i$,
where $|N'_l|=max_{1\leq k'\leq |K|}|N'_{k'}|$. \\
$\sum_{1\leq i\leq |N'_l|}u_i$ represents the utility sum of the $|N'_l|$ systems that have the highest utility among all the systems.
\end{cor}
\begin{proof}
Although Situation 2 is different from Situation 1,
the worst case in Situation 2 is the same as that in Situation 1.
This is because in Situation 2, all genuine systems are still in the network.
In the worst case, the attacker attacks all of the systems with the highest utility.
The result, thus, is the same as that in Situation 1.
\end{proof}

\begin{thm}\label{thm:DUtility3}
In Situation 3 ($|N'|<|N|$), the defender's expected utility loss is
\\$\sum_{1\leq i\leq |N'|}(u_i\cdot\sum_{1\leq k'\leq |K|}p_{k'}\mathcal{P}_{k'})-\sum_{1\leq j\leq |N|-|N'|}(u_j-l_j)$,
where $p_{k'}=\frac{exp(\frac{\frac{\epsilon}{|N|} U_{k'}}{2\Delta U_d})}{\sum_{1\leq k\leq |K|}exp(\frac{\frac{\epsilon}{|N|} U_{k}}{2\Delta U_d})}$,
if the defender uses Algorithm \ref{alg:deployment3} to deploy configurations
and the attacker uses Equation \ref{eq:selection} to select the target.
\end{thm}
\begin{proof}
By comparing the results of Theorems \ref{thm:DUtility1} and \ref{thm:DUtility3},
we can see that there is an extra part in Theorem \ref{thm:DUtility3}, $\sum_{1\leq j\leq |N|-|N'|}(u_j-l_j)$.
This part means that $|N|-|N'|$ systems are taken offline.
Thus, the utility of these systems will not be lost.
\end{proof}

\begin{cor}\label{cor:DUtility3}
In Situation 3, the upper bound of the defender's utility loss is $\sum_{1\leq i\leq |N'_l|}u_i$,
where $|N'_l|=max_{1\leq k'\leq |K|}|N'_{k'}|$.\\
$\sum_{1\leq i\leq |N'_l|}u_i$ represents the utility sum of the $|N'_l|$ systems
that have the highest utility among the remaining $|N'|$ systems.
\end{cor}
\begin{proof}
In Situation 3, $|N|-|N'|$ systems are taken offline.
These systems, therefore, will not be attacked.
Hence, in the worst case, the $|N'_l|$ systems with the highest utility among the remaining $|N'|$ systems are attacked.
\end{proof}

\subsubsection{Analysis of the complexity of the algorithms}
The proposed approach includes four algorithms.
The analysis of their complexity is given as follows.

\begin{thm}\label{thm:Agl1Complexity}
The computational complexity of Algorithm \ref{alg:DP} is $O(|N|\cdot|K|)$,
where $|N|$ is the number of systems and $|K|$ is the number of configurations.
\end{thm}
\begin{proof}
In Line 3 of Algorithm \ref{alg:DP}, $|N|$ systems are partitioned into $|K|$ categories.
This partition implicitly involves a nested loop,
where the number of iterations is $|N|\cdot|K|$.
In Lines 4-5 of Algorithm \ref{alg:DP}, the number of iterations in this loop is $|K|$.
Thus, the overall number of iterations is $|N|\cdot|K|+|K|=(|N|+1)\cdot|K|$.
The complexity of Algorithm \ref{alg:DP} is $O(|N|\cdot|K|)$.
\end{proof}

\begin{thm}\label{thm:Agl2Complexity}
The computational complexity of Algorithm \ref{alg:deployment1} is $O(|N|^2)$,
where $|N|$ is the number of systems.
\end{thm}
\begin{proof}
In Line 5 of Algorithm \ref{alg:deployment1}, $|N|$ systems are ranked in an increasing order based on their utilities.
The number of iterations in this ranking is $\frac{|N|\cdot(|N|-1)}{2}$.
Moreover, in Lines 7 to 14, the number of iterations in this loop is $|N|$.
Thus, the overall number of iterations is $\frac{|N|\cdot(|N|-1)}{2}+|N|=\frac{|N|\cdot(|N|+1)}{2}$.
The complexity of Algorithm \ref{alg:deployment1}, therefore, is $O(|N|^2)$.
\end{proof}

\begin{thm}\label{thm:Agl3Complexity}
The computational complexity of Algorithm \ref{alg:deployment2} is $O(|N|^2)$,
where $|N|$ is the number of systems.
\end{thm}
\begin{proof}
In Line 6 of Algorithm \ref{alg:deployment2}, $|N|$ systems are ranked in an increasing order based on their utilities.
The number of iterations in this ranking is $\frac{|N|\cdot(|N|-1)}{2}$.
Moreover, in Lines 7 to 14, the number of iterations in this loop is $|N'|-|N|$,
where $|N'|$ is the number of systems after obfuscation.
Also, in Lines 15 to 23, the number of iterations in this loop is $|N|$.
Thus, the overall number of iterations is $\frac{|N|\cdot(|N|-1)}{2}+|N'|$.
Since $|N'|$ and $|N|$ are in the same scale,
the complexity of Algorithm \ref{alg:deployment1} is $O(|N|^2)$.
\end{proof}

\begin{thm}\label{thm:Agl4Complexity}
The computational complexity of Algorithm \ref{alg:deployment3} is $O(|N|^2)$,
where $|N|$ is the number of systems.
\end{thm}
\begin{proof}
In Line 6 of Algorithm \ref{alg:deployment3}, $|N|$ systems are ranked in an increasing order based on their utilities.
The number of iterations in this ranking is $\frac{|N|\cdot(|N|-1)}{2}$.
Moreover, in Lines 7 to 9, the number of iterations in this loop is $|N|-|N'|$,
where $|N'|$ is the number of systems after obfuscation.
Also, in Lines 10 to 18, the number of iterations in this loop is $|N'|$.
Thus, the overall number of iterations is $\frac{|N|\cdot(|N|-1)}{2}+|N|=\frac{|N|\cdot(|N|+1)}{2}$.
The complexity of Algorithm \ref{alg:deployment1}, therefore, is $O(|N|^2)$.
\end{proof}

\subsection{Analysis of the attacker's strategy}\label{sub:attacker analysis}
We first analyze the attacker's optimal strategy
and then compute the lower bound of his expected utility gain.

\textbf{Remark 2} (the attacker's expected utility gain).
As the cyber deception game is a zero-sum game in Situations 1 ($|N'|=|N|$) and 2 ($|N'|>|N|$),
the defender's expected utility loss is the attacker's expected utility gain.
In Situation 3 ($|N'|<|N|$), since the attacker does not attack the offline systems,
his expected utility gain is based only on the utility of the remaining $|N'|$ systems,
which is $\sum_{1\leq i\leq |N'|}(u_i\cdot\sum_{1\leq k'\leq |K|}p_{k'}\mathcal{P}_{k'})$.

\begin{thm}\label{thm:AStrategy}
The attacker's optimal strategy is the solution to the following problem:
given $U^{(1)}_a$, ..., $U^{(|K|)}_a$, find a probability distribution, $\mathcal{P}_1$, ..., $\mathcal{P}_{|K|}$,
which can maximize $\sum_{1\leq k'\leq |K|}\mathcal{P}_{k'}\cdot U^{(k')}_a$.
\end{thm}
\begin{proof}
According to the discussion in \textbf{Remark 2},
the major component of the attacker's expected utility gain is $\sum_{1\leq i\leq |N'|}(u_i\cdot\sum_{1\leq k'\leq |K|}p_{k'}\mathcal{P}_{k'})$.
In this component, the only uncertain part is $\mathcal{P}(k')$.
The calculation of $\mathcal{P}(k')$ is based on Equations \ref{eq:Bayes} and \ref{eq:utilitygain}.

In Equation \ref{eq:Bayes}, $q(k|k')=\frac{q(k)\cdot q(k'|k)}{q(k')}$, $q(k)$ and $q(k')$ are known to the attacker,
while $q(k'|k)$ can be calculated by the attacker based on the defender's strategy.
Hence, Equation \ref{eq:utilitygain}, $U^{(k')}_a=\sum_{1\leq k\leq |K|}[q(k|k')\cdot u_k]$, can also be computed by the attacker.
Since the attacker's aim is to maximize his expected utility gain,
he needs to identify the appropriate probability distribution, $\mathcal{P}_{1}$, ..., $\mathcal{P}_{|K|}$, for use as his attack strategy
in order to maximize $\sum_{1\leq k'\leq |K|}\mathcal{P}_{k'}\cdot U^{(k')}_a$.
\end{proof}

\textbf{Remark 3} (The attacker's risk).
According to Theorem \ref{thm:AStrategy}, the attacker's optimal strategy should be to attack configuration $k'$ with probability $1$,
where $U^{(k')}_a$ is the maximum among $U^{(1)}_a$, ..., $U^{(|K|)}_a$.
However, such a deterministic strategy is highly predictable to the defender.
The attacker has to use a mixed strategy
that is as close as possible to the deterministic strategy.
The $epsilon$-greedy strategy is a good solution,
as it is similar to the deterministic strategy while also exhibits a reasonable degree of randomness.
In the $epsilon$-greedy strategy, the maximum $U^{(k')}_a$ is given a high probability,
while the remainder shares the remaining probability.

\begin{thm}\label{thm:AUtiligy}
Let $U^{(1)}_a\geq U^{(2)}_a\geq... U^{(|K|)}_a$,
the lower bound of the attacker's utility gain is $\frac{1}{|K|}\sum_{1\leq k'\leq |K|}U^{k'}_a$.
\end{thm}
\begin{proof}
By using the $epsilon$-greedy strategy,
the attacker's expected utility gain is \\
$U^{(1)}_a[(1-e)+\frac{e}{|K|}]+\frac{e}{|K|}\sum_{2\leq k'\leq |K|}U^{(k')}_a$\\
$=(1-e)U^{(1)}_a+\frac{e}{|K|}\sum_{1\leq k'\leq |K|}U^{(k')}_a$.

This expected utility gain decreases monotonically with the increase of the value of $e$.
Thus, when $e=1$, the expected utility reaches the minimum: $\frac{1}{|K|}\sum_{1\leq k'\leq |K|}U^{k'}_a$.
\end{proof}

\subsection{Analysis of the equilibria}
We provide a characterization of equilibria
and leave the deep investigation of equilibria in this highly dynamic security game as one of our future studies.

In our approach, the number of systems may change at each round
due to the introduction of Laplace noise (recall Algorithm \ref{alg:DP}).
Therefore, the available strategies to the defender also change at each round,
where each strategy is interpreted as an assignment of systems to configurations, i.e., obfuscation.
By comparison, the available strategies to the attacker is fixed in the game,
because 1) the number of configurations is fixed;
and 2) the attacker selects only one configuration to attack at each round,
i.e., attacking the systems associated with a configuration.

Formally, let $S$ be the set of all potentially attainable strategies of the defender,
and $S^t\subset S$ be the set of the defender's strategies at round $t$.
Correspondingly, let $A^t\in\mathbb{R}^{m_t\times n}$ be the payoff matrix at round $t$,
where $m_t=|S^t|$ is the number of the defender's strategies at round $t$
and $n$ is the number of the attacker's strategies.
Let $\Delta^t_d$ be the set of probability distributions over the defender's available strategies at round $t$
and $\Delta_a$ be the set of probability distributions over the attacker's available strategies.
Then, at round $t$, if the defender uses a probability distribution $\mathbf{x}\in\Delta^t_d$
and the attacker uses a probability distribution $\mathbf{y}\in\Delta_a$,
the defender's utility loss is $\mathbf{x}A^t\mathbf{y}$.
The defender wishes to minimize the quantity of $\mathbf{x}A^t\mathbf{y}$,
while the attacker wants to maximize it.
This is a \emph{minimax} problem.
The solution of this problem is an equilibrium of the game \cite{Gilpin08}.
However, the defender's strategy set $S^t$ at each round $t$ is randomly generated using the differentially private Laplace mechanism
and $S^t$ is unknown to the attacker.
Hence, the attacker is difficult to precisely compute the solution
which can force the defender not to unilaterally change her strategy.
This finding has also been demonstrated in our experimental results
which will be given in detail in Section \ref{sec:experiment}.

\section{Experiment and Analysis}\label{sec:experiment}
\subsection{Experimental setup}
In the experiments, three defender approaches are evaluated,
which are our approach, referred to as \emph{DP-based},
a deterministic greedy approach \cite{Schlenker18}, referred to as \emph{Greedy},
and a mixed greedy approach, referred to as \emph{Greedy-Mixed}.
The \emph{Greedy} approach obfuscates systems using a greedy and deterministic strategy to minimize the defender's expected utility loss.
The \emph{Greedy-Mixed} approach obfuscates systems using a greedy and mixed strategy,
where a high utility system is obfuscated to a low utility system with a certain probability.

Two metrics are used to evaluate the three approaches: the attacker's utility gain and the defender's cost.
For the first metric, the attacker's utility gain, its amount is the same as the amount of the defender's utility loss.
We use the metric attacker’s utility gain by following reference \cite{Schlenker18}.
Certainly, using the defender’s utility loss as a metric will reach the same results.
For the second metric, the defender's cost, it is the defender’s deployment cost
used to obfuscate configurations of systems.

Against the three defender approaches, two attacker approaches are adopted,
which are the Bayesian inference approach and the deep reinforcement learning (DRL) approach.
The Bayesian inference approach has been described in Section \ref{sub:Bayesian}.
The DRL approach involves a deep neural network
which consists of an input layer, two hidden layers and an output layer.
The input is the state of the attacker,
which is defined as the selected configurations in the last eight rounds
and the corresponding utility received by attacking the systems in these configurations.
The output is the selected configuration
that the attacker will attack in this round.
Moreover, each of the hidden layer has ten neurons.
The hyper-parameters used in DRL are listed as follows:
learning rate $\alpha=0.1$, discount factor $\gamma=0.9$,
$epsilon$-greedy selection probability $epsilon=0.9$,
$batch\_size=32$ and $memory\_size=2000$.

These two techniques, Bayesian inference and deep reinforcement learning,
are powerful enough and prevalent in cybersecurity \cite{Xie10,Nguyen19b,Chivukula20}.
By comparison, deep reinforcement learning is more powerful than Bayesian inference,
but it is more complex than Bayesian inference and it needs a long training period.
Therefore, we take both of them into our evaluation.

The experiments are conducted in four scenarios.

\noindent\textbf{Scenario 1}: The numbers of systems and configurations vary,
  but the average number of systems associated with each configuration is fixed.
  This scenario is used to evaluate the scalability of the three approaches.
  It simulates the real-world situation that
  different organizations have different scales of networks.

\noindent\textbf{Scenario 2}: The number of systems is fixed,
  but the number of configurations varies.
  Thus, the average number of systems associated with each configuration also varies.
 This scenario is used to evaluate the adaptivity of the three approaches.
  It simulates the real-world situation that
  an organization updates the configurations of systems from time to time.

\noindent\textbf{Scenario 3}: The number of both systems and configurations is fixed,
  but the value of the defender's deployment budget $B_d$ varies.
  This scenario is used to evaluate the impact of deployment budget value on the performance of the three approaches.
  It simulates the real-world situation that
  different organizations have different deployment budgets
  or an organization adjusts its deployment budget sometimes.

\noindent\textbf{Scenario 4}: The number of both the systems and configurations is fixed,
  but the value of the privacy budget $\epsilon$ varies.
  This scenario is used to evaluate the impact of the privacy budget value on the performance of our approach.
  It simulates the real-world situation that
  different organizations have different requirements of privacy level
  and thus requires different privacy budgets.


In a given experimental setting, as arbitrarily changing the number of systems may incur security issues,
in each round we use the differential privacy mechanism (Algorithm \ref{alg:DP})
to compute the number of systems in each configuration.
Therefore, the number of systems is strategically determined and their privacy can be guaranteed.

The value of $\Delta S$ is set to $1$.
The utility of a configuration, $u_k$, is uniformly drawn from $[1,20]$.
The defender's obfuscation cost, $c(k',k)$, is set to $0.1\cdot u'_k$.
The defender's honeypot deployment cost, $h_k$, is set to $0.2\cdot u_k$.
The defender's utility loss, $l_k$, for taking a system with configuration $k$ offline is set to $0.15\cdot u_k$.
These settings are experimentally set up to yield the best results.
Here, ``best'' means the most representative.
Increasing these costs will increase both the cost and utility loss of the defender.
On one hand, these costs are used by the defender to obfuscate configurations of systems and deploy honeypots.
Thus, increasing these costs will increase the defender’s cost.
On the other hand, increasing these costs will result in the defender using up her deployment budget too early.
The defender then cannot perform obfuscation or deployment.
Thus, her utility loss will increase.
On the contrary, decreasing these costs will decrease both the cost and utility loss of the defender.
However, excessively decreasing these costs will render the experimental results meaningless.
Therefore, we set representative parameter values to
balance the defender’s performance and the meanings of the experimental results.
The experimental results are obtained by averaging $1000$ rounds of the game.
\subsection{Experimental results}
\subsubsection{Scenario 1}
\begin{figure}[ht]
\vspace{-6mm}
\centering
	\begin{minipage}{0.45\textwidth}
   \subfigure[\scriptsize{The attacker's utility in different scales of networks}]{
    \includegraphics[width=0.45\textwidth, height=3cm]{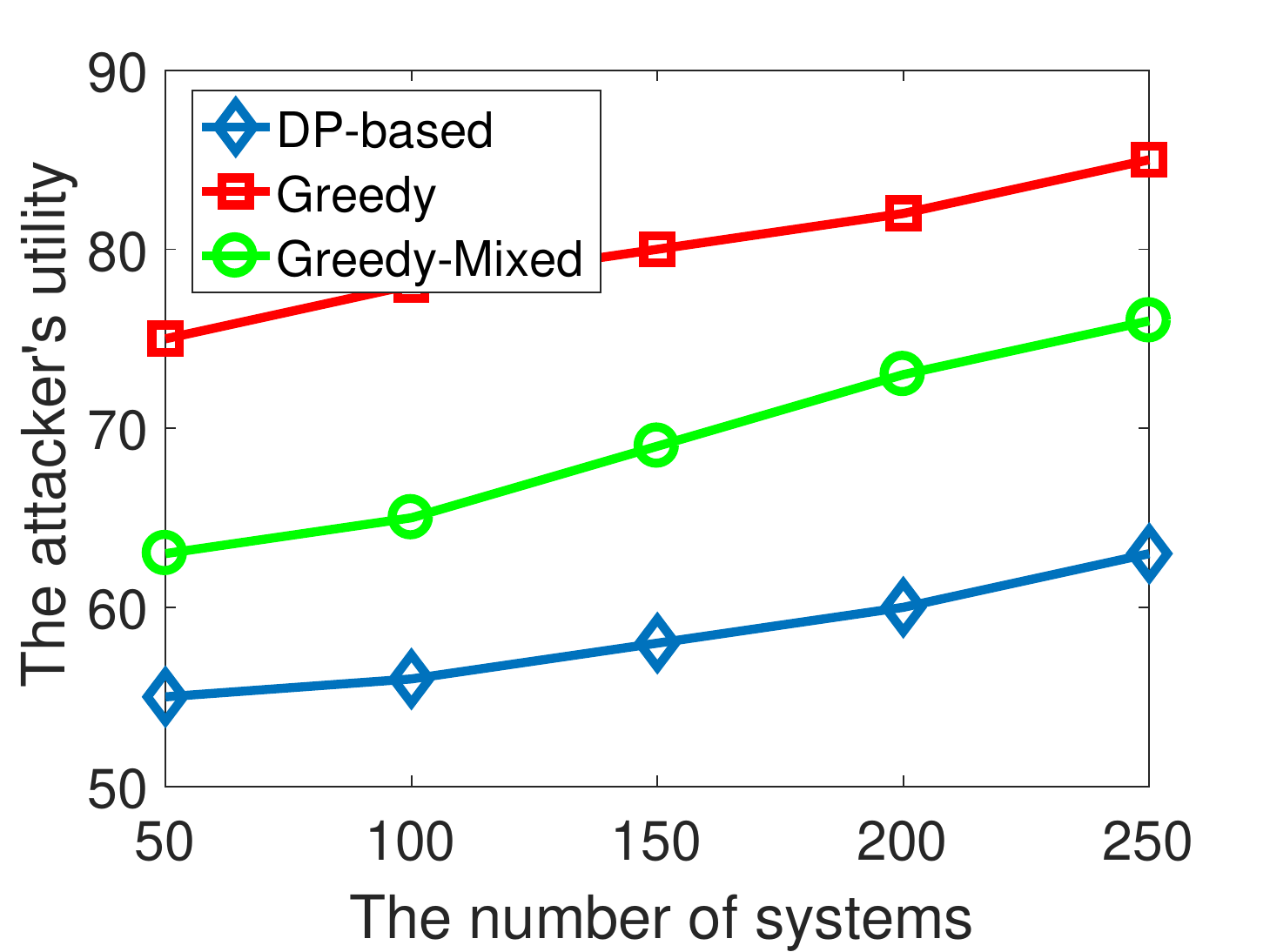}
			\label{fig:S1Utility}}
    \subfigure[\scriptsize{The defender's cost in different scales of networks}]{
    \includegraphics[width=0.45\textwidth, height=3cm]{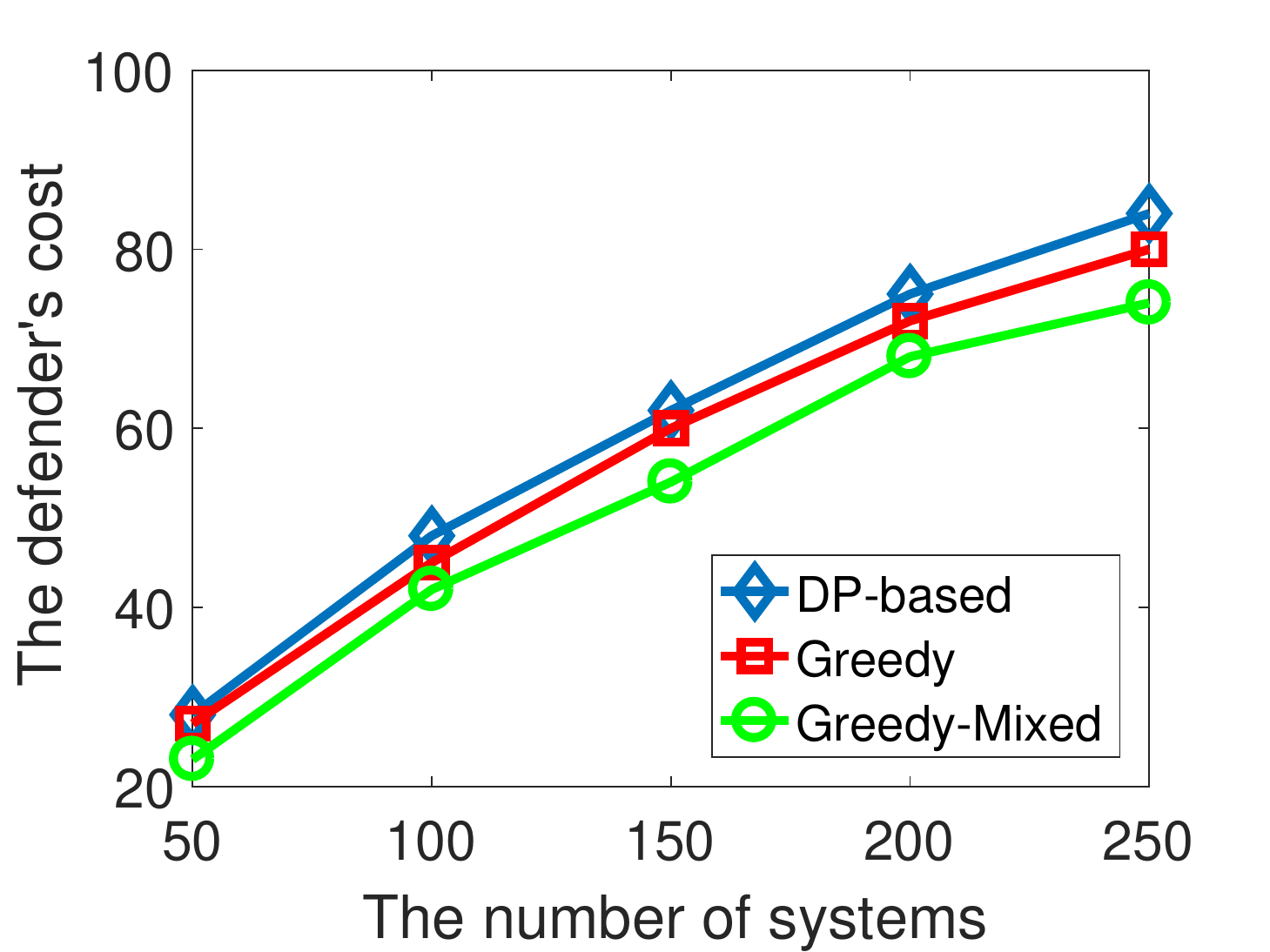}
			\label{fig:S1Cost}}\\[2ex]
    \end{minipage}
    \vspace{-5mm}
	\caption{The three approaches' performance in Scenario 1}
\vspace{-2mm}
	\label{fig:S1}
\end{figure}

Fig. \ref{fig:S1} demonstrates the performance of the three approaches in Scenario 1.
The number of systems varies from $50$ to $250$,
and the number of configurations varies from $5$ to $25$, accordingly.
The privacy budget $\epsilon$ is fixed at $0.3$.
The cost budget $B_d$ is fixed at $1000$ for each round.

As the number of systems increases, across the three approaches,
the attacker's utility gain increases slightly and linearly
while the defender's cost increases sub-linearly.
Along with the increase in the number of systems, the number of configurations also increases.
More configurations gives the attacker more choices.
By using Bayesian inference, the attacker can choose a configuration with high utility, to attack.
Thus, the attacker's utility gain increases linearly.
As the number of systems increases, according to Algorithms \ref{alg:deployment1}, \ref{alg:deployment2} and \ref{alg:deployment3},
the defender deals with more systems.
Thus, the defender's cost inevitably increases provided that the budget $B_d$ is large enough.

Comparing our \emph{DP-based} approach to the \emph{Greedy} and \emph{Greedy-Mixed} approaches,
the attacker in the \emph{DP-based} approach achieves about $30\%$ and $12\%$ less utility than in the \emph{Greedy} and \emph{Greedy-Mixed} approaches, respectively.
Moreover, the defender in the \emph{DP-based} approach incurs about $3\%$ and $5\%$ more cost than in the \emph{Greedy} and \emph{Greedy-Mixed} approaches, respectively.
In the \emph{DP-based} approach, the defender can not only obfuscate systems,
but can also deploy honeypots to attract the attacker.
The cost of deploying a honeypot exceeds that of obfuscating a system.
However, a honeypot results in a negative utility to the attacker,
while obfuscating a system can only reduce the attacker's utility gain.
By comparison, the defender in the \emph{Greedy} and \emph{Greedy-Mixed} approaches only obfuscates systems.
The difference between defender's costs among the three approaches is negligible.
This demonstrates that in our \emph{DP-based} approach,
the defender uses almost the same cost as the other two approaches
to achieve much better results, i.e., lowering the attacker’s utility gain.

\begin{figure}[ht]
\vspace{-6mm}
\centering
	\begin{minipage}{0.51\textwidth}
    \subfigure[\scriptsize{The attacker's utility as game progresses in \emph{DP-based}}]{
    \includegraphics[width=0.45\textwidth, height=3cm]{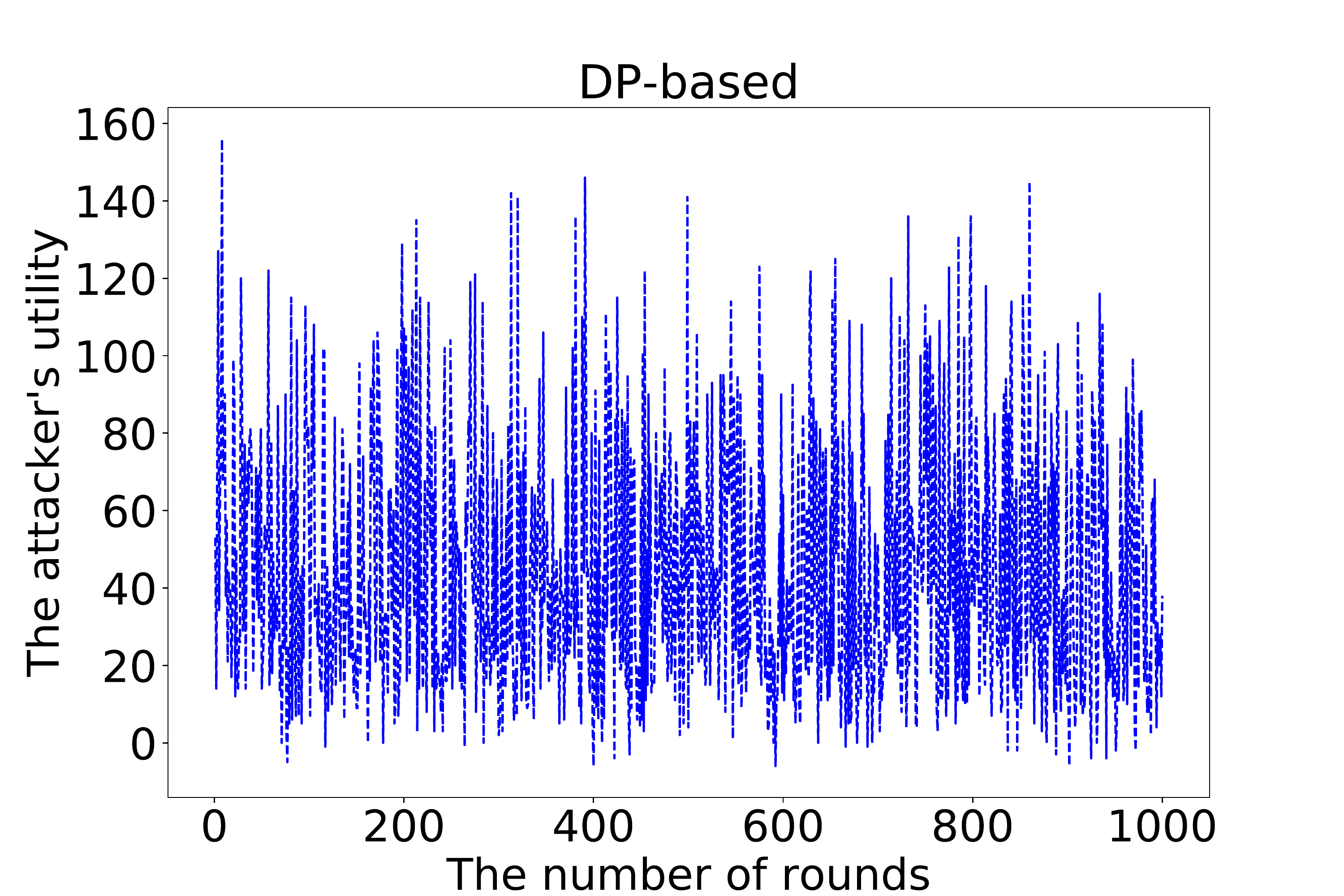}
			\label{fig:S1DPUtility}}
	    \subfigure[\scriptsize{The defender's cost as game progresses in \emph{DP-based}}]{
    \includegraphics[width=0.45\textwidth, height=3cm]{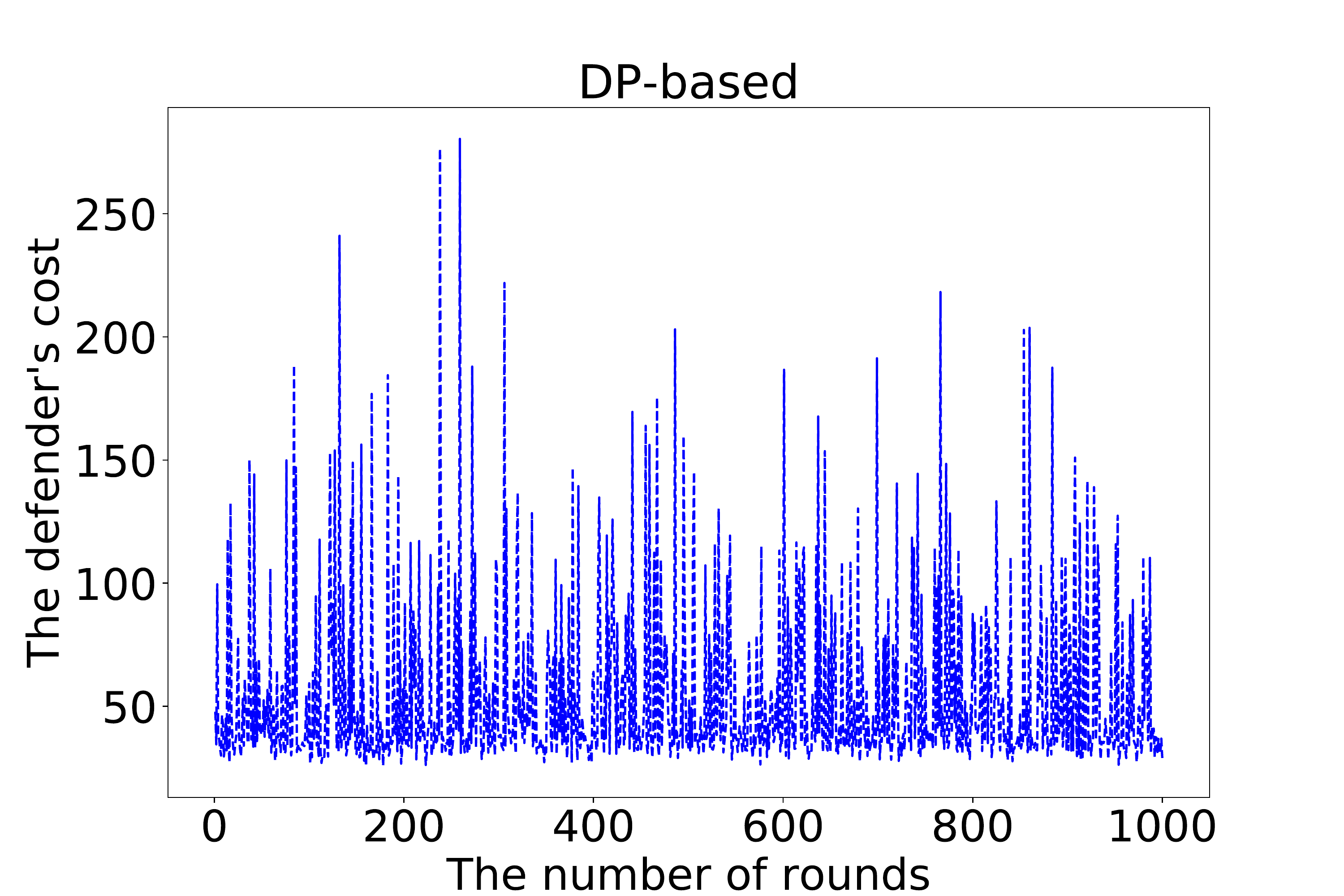}
			\label{fig:S1DPCost}}\\[2ex]
			
        \subfigure[\scriptsize{The attacker's utility as game progresses in \emph{Greedy}}]{
    \includegraphics[width=0.45\textwidth, height=3cm]{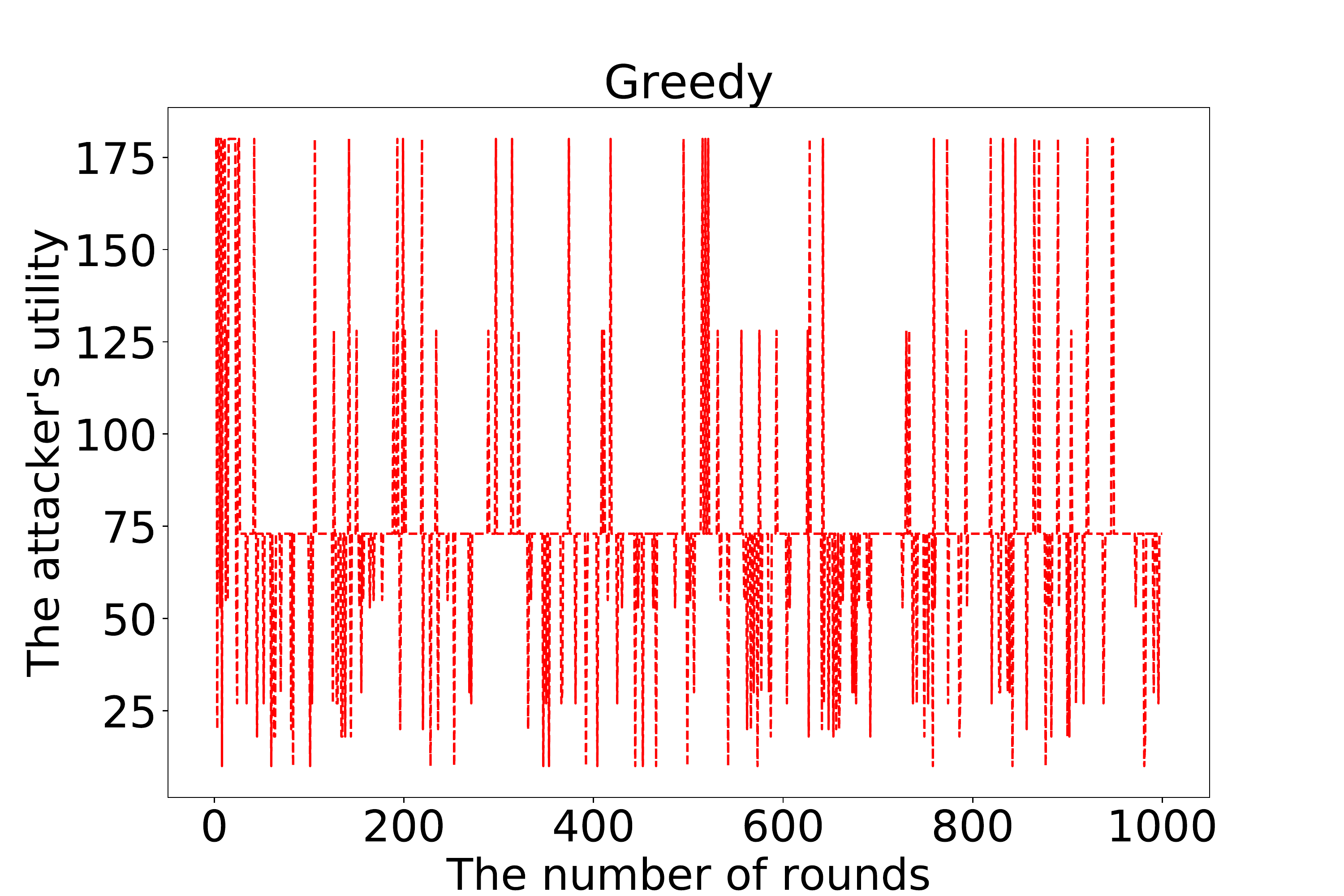}
			\label{fig:S1GreedyUtility}}
		    \subfigure[\scriptsize{The defender's cost as game progresses in \emph{Greedy}}]{
    \includegraphics[width=0.45\textwidth, height=3cm]{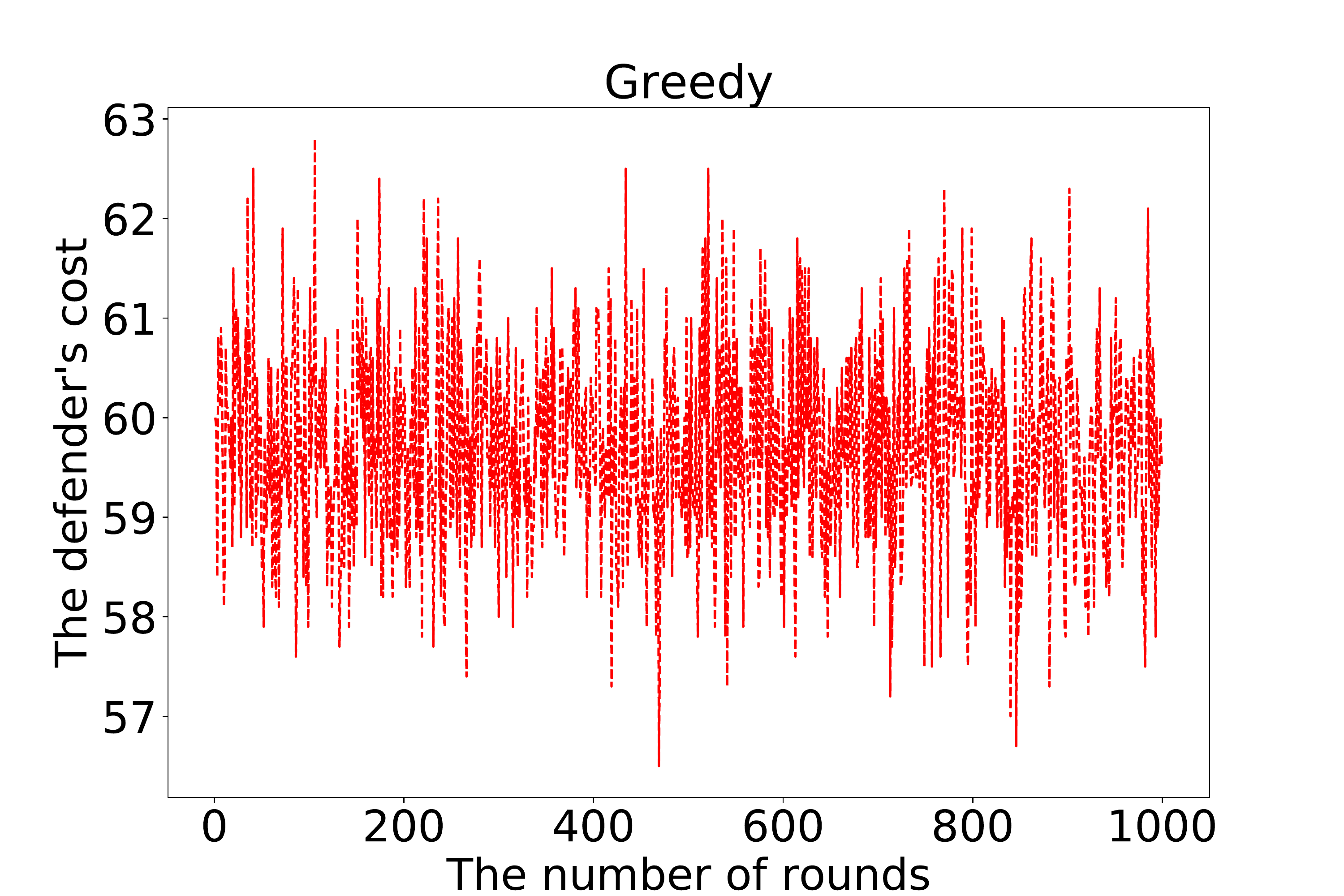}
			\label{fig:S1GreedyCost}}\\[2ex]
			
    \subfigure[\scriptsize{The attacker's utility as game progresses in \emph{Greedy-Mixed}}]{
    \includegraphics[width=0.45\textwidth, height=3cm]{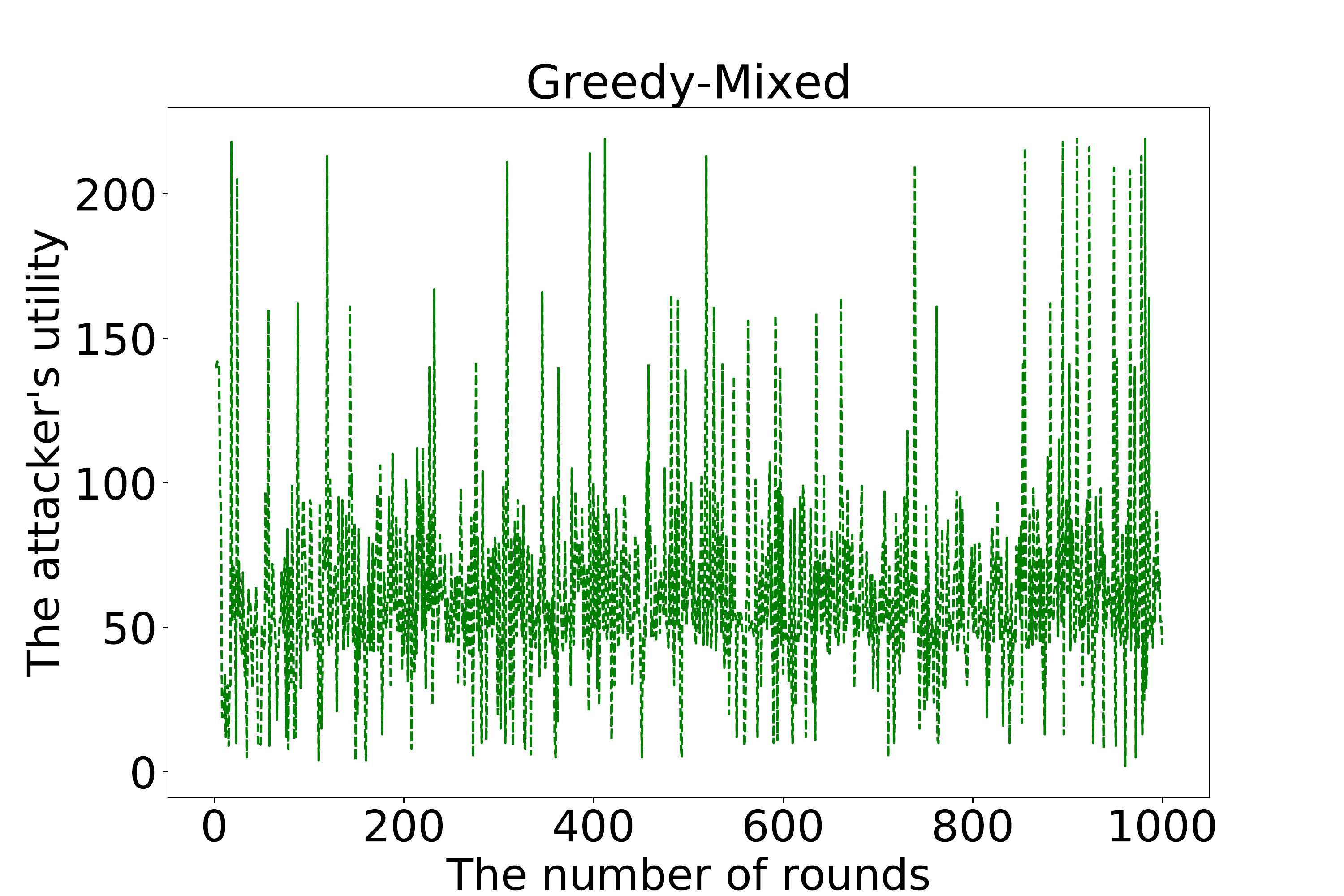}
			\label{fig:S1GreedyMixedUtility}}
    \subfigure[\scriptsize{The defender's cost as game progresses in \emph{Greedy-Mixed}}]{
    \includegraphics[width=0.45\textwidth, height=3cm]{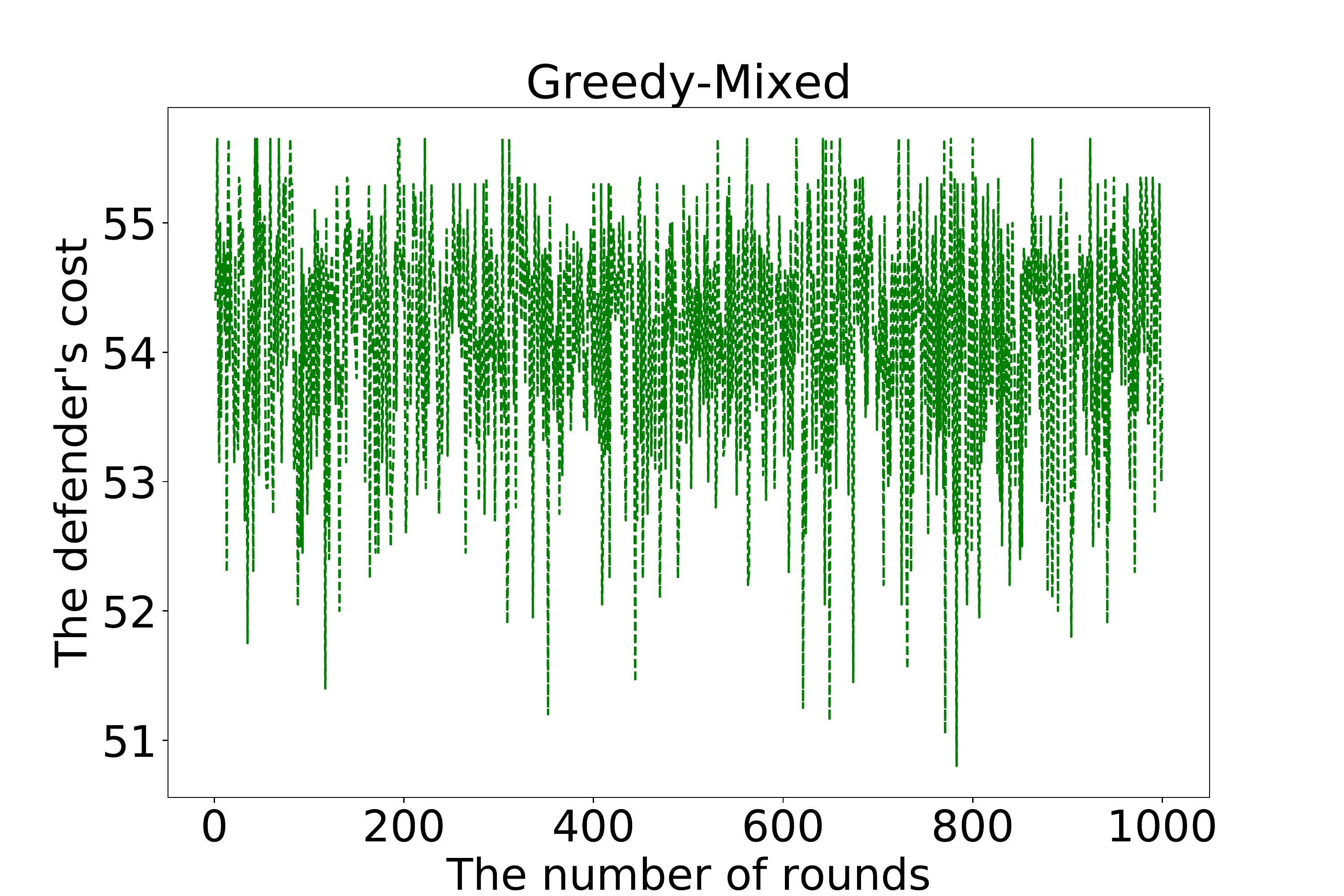}
			\label{fig:S1GreedyMixedCost}}\\
    \end{minipage}
    \vspace{-4mm}
	\caption{The three approaches' performance as game progress in Scenario 1}
\vspace{-2mm}
	\label{fig:S1Round}
\end{figure}

Fig. \ref{fig:S1Round} shows the variation in the attacker's utility gain
and the defender's cost in the three approaches as the game progresses in Scenario 1.
The number of systems is fixed at $100$,
and the number of configurations is fixed at $10$.

In Fig. \ref{fig:S1GreedyUtility}, which depicts the \emph{Greedy} approach, there are a number of plateaus.
These plateaus indicate the steadiness of the attacker's utility gain,
which implies that the attacker has predicted the defender's strategy.
The attacker can thus adopt an optimal strategy to maximize his utility gain.
By comparison, in Fig. \ref{fig:S1DPUtility} and Fig. \ref{fig:S1GreedyMixedUtility},
which depicts the \emph{DP-based} and \emph{Greedy-Mixed} approaches, respectively, there is no plateau.
This means that the attacker cannot predict the defender's strategy
and adopts only random strategies.
This finding also demonstrates that equilibria may not exist between the defender and attacker
if the defender uses our \emph{DP-based} approach,
as the attacker’s utility fluctuates all the time with no equilibrium.
However, when the defender uses the \emph{Greedy} approach
which does not change the available strategies of the defender,
the attacker can predict the defender’s strategy and may reach an equilibrium.
Particularly, by comparing Figs. \ref{fig:S1DPUtility}, \ref{fig:S1GreedyUtility} and \ref{fig:S1GreedyMixedUtility},
we can see that the shape of Fig. \ref{fig:S1GreedyMixedUtility} is more similar to Fig. \ref{fig:S1GreedyUtility} than Fig. \ref{fig:S1DPUtility}.
This implies that although the defender in the \emph{Greedy-Mixed} approach uses a mixed strategy,
the attacker can still predict the defender's strategy to some extent.

By comparing Figs. \ref{fig:S1DPCost}, \ref{fig:S1GreedyCost} and \ref{fig:S1GreedyMixedCost},
we can see that the variation of the defender's cost in the \emph{Greedy} and \emph{Greedy-Mixed} approaches
is more stable than the \emph{DP-based} approach.
This is due to the fact that in the \emph{DP-based} approach,
the defender can deploy honeypots which incurs extra cost.
However, as shown in Fig. \ref{fig:S1Cost},
the defender's average cost in the \emph{DP-based} approach still stays at a relatively low level
in comparison with the \emph{Greedy} and \emph{Greedy-Mixed} approaches.


\begin{figure}[ht]
\vspace{-4mm}
\centering
	\begin{minipage}{0.51\textwidth}
	\subfigure[\scriptsize{The attacker's utility with DRL in different scales of networks}]{
	\includegraphics[width=0.45\textwidth, height=3cm]{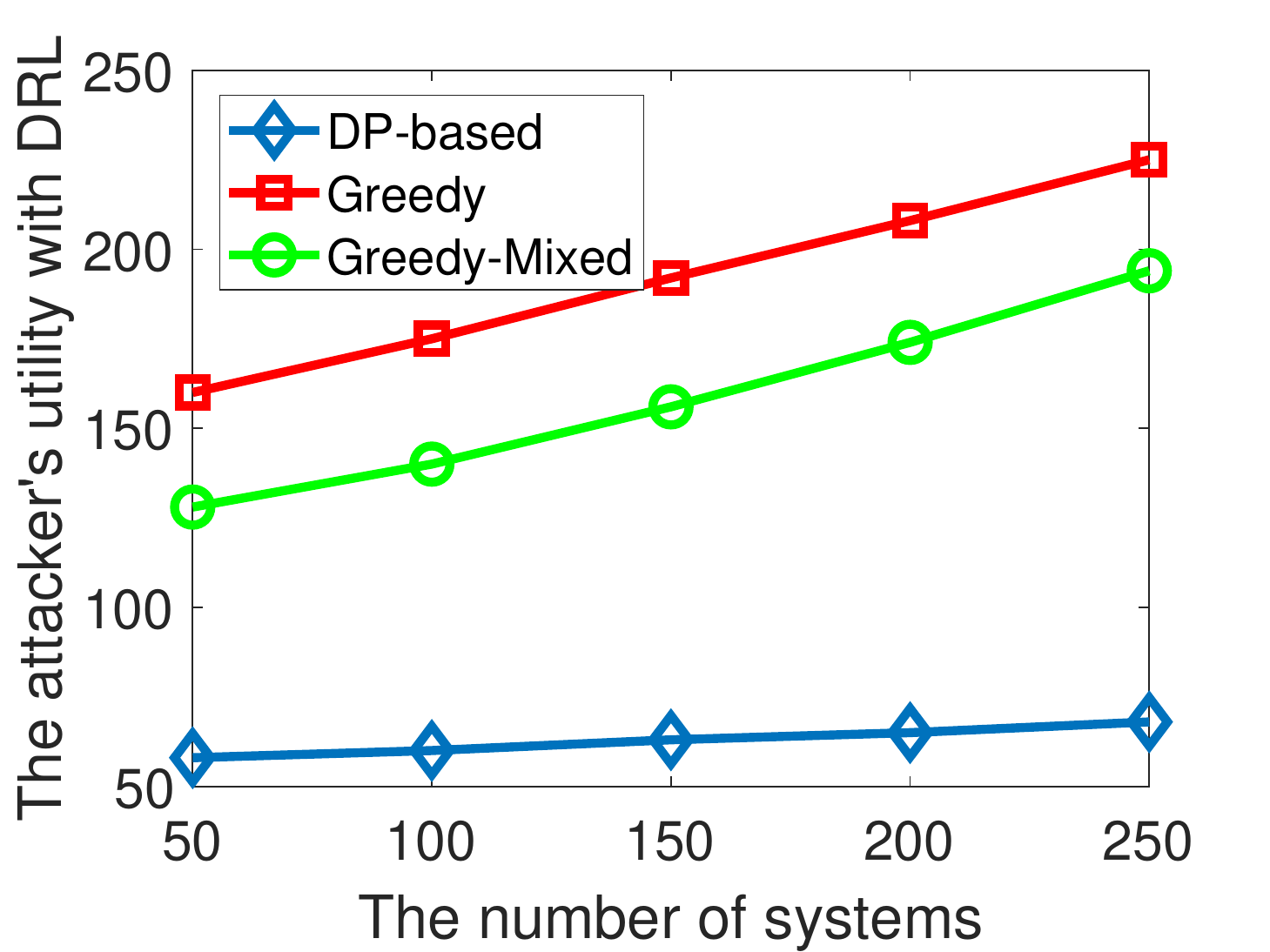}
	\label{fig:S1DRL}}
    \subfigure[\scriptsize{The attacker's utility with DRL as game progresses in \emph{DP-based}}]{
    \includegraphics[width=0.45\textwidth, height=3cm]{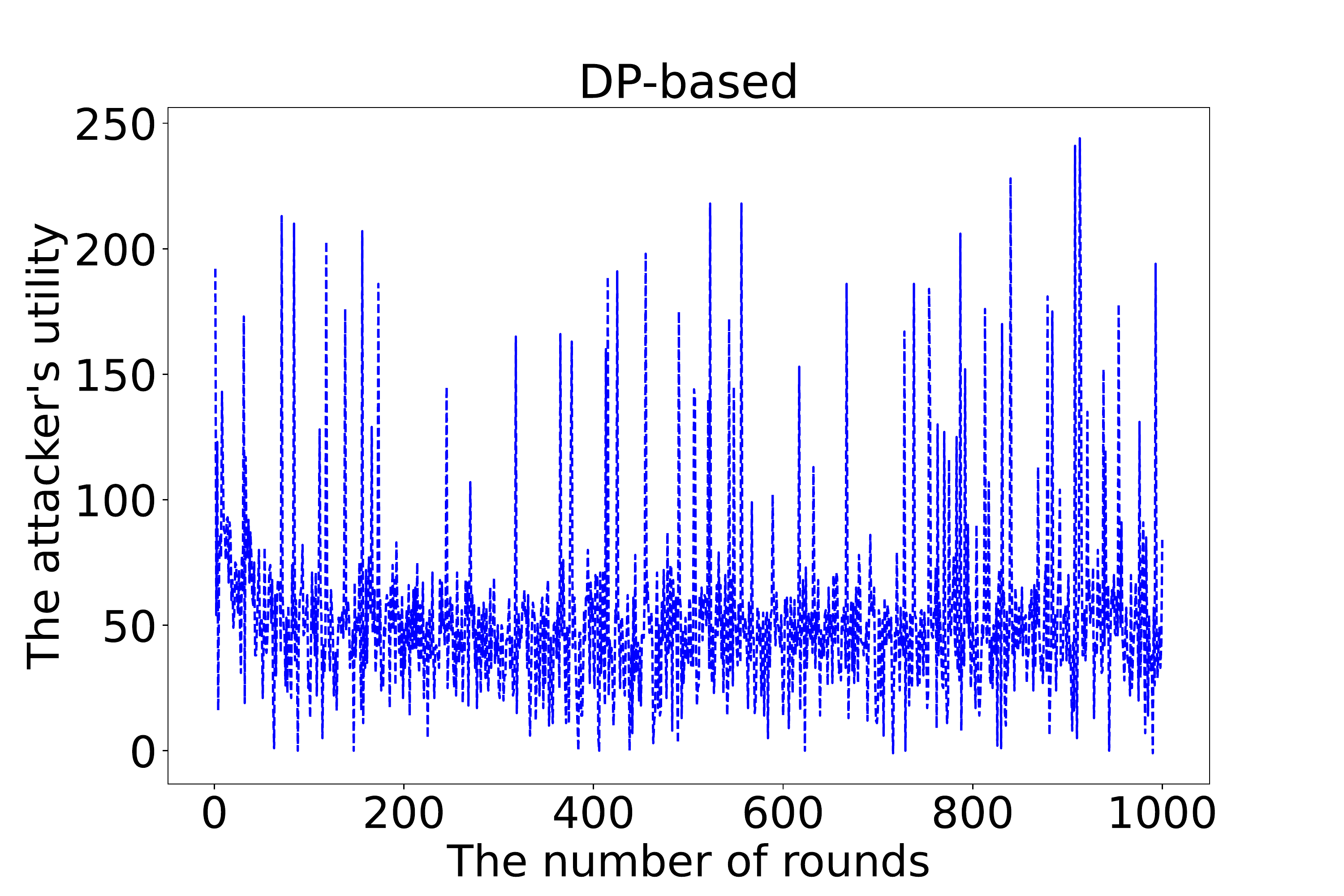}
			\label{fig:S1DPUtilityDRL}}\\[2ex]
        \subfigure[\scriptsize{The attacker's utility with DRL as game progresses in \emph{Greedy}}]{
    \includegraphics[width=0.45\textwidth, height=3cm]{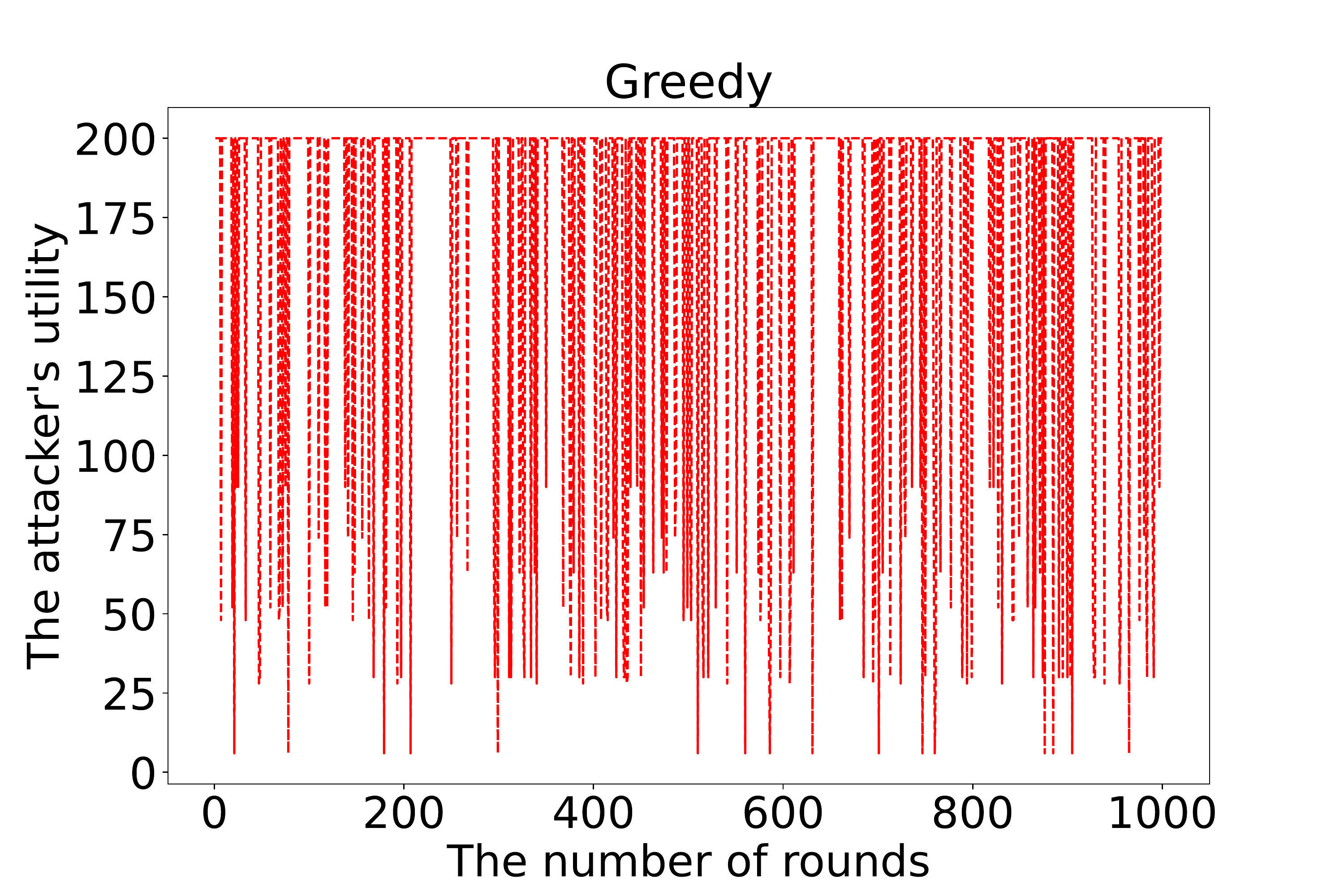}
			\label{fig:S1GreedyUtilityDRL}}
    \subfigure[\scriptsize{The attacker's utility with DRL as game progresses in \emph{Greedy-Mixed}}]{
    \includegraphics[width=0.45\textwidth, height=3cm]{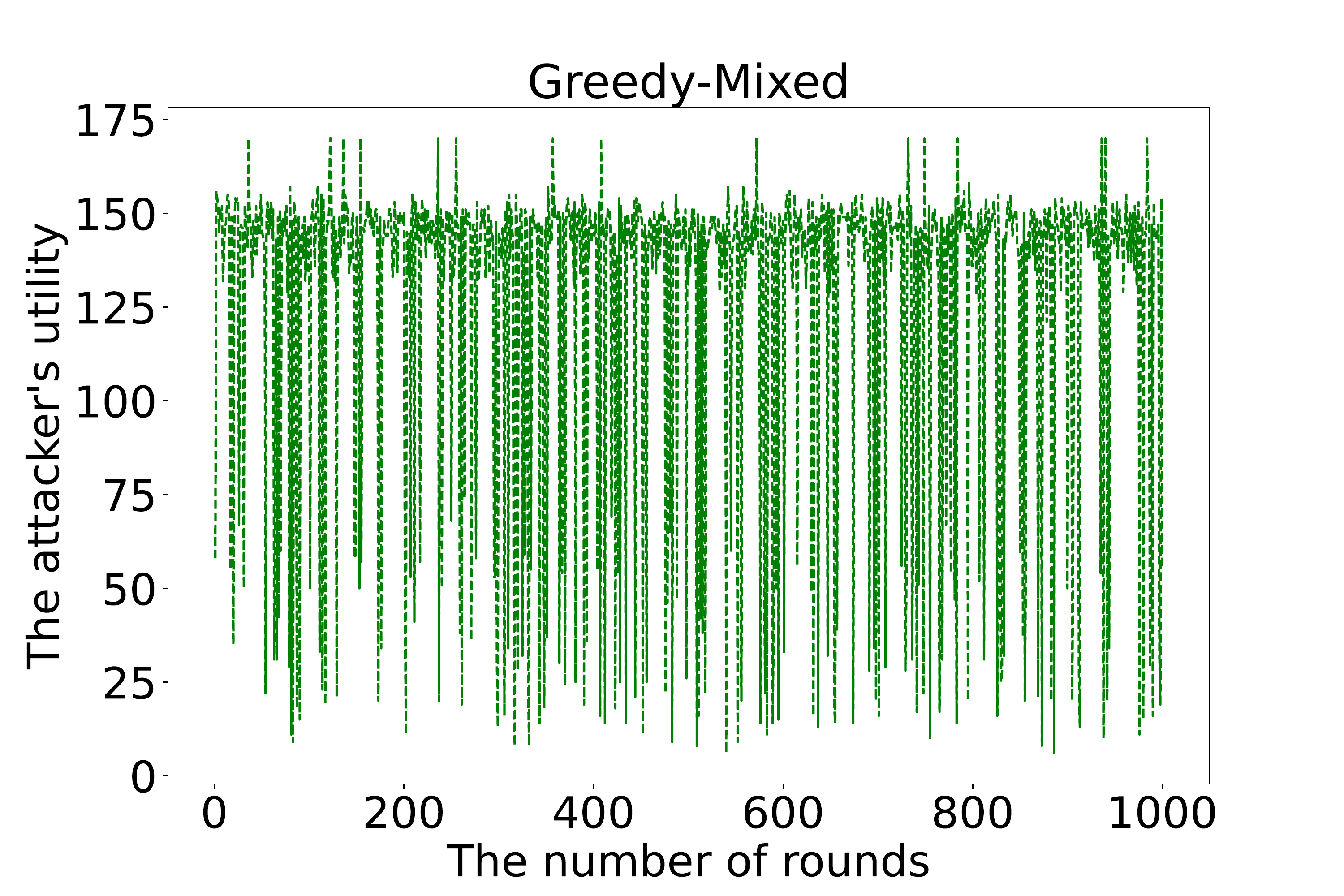}
			\label{fig:S1GreedyMixedUtilityDRL}}
    \end{minipage}
    \vspace{-4mm}
	\caption{The three approaches against the DRL attacker in Scenario 1}
\vspace{-2mm}
	\label{fig:S1RoundDRL}
\end{figure}

Fig. \ref{fig:S1RoundDRL} demonstrates the performance of the three defender approaches against the DRL attacker.
Compared to the Bayesian inference attacker (Figs. \ref{fig:S1} and \ref{fig:S1Round}), the DRL attacker can obtain more utility,
when the defender adopts either the \emph{Greedy} or the \emph{Greedy-Mixed} approach.
This is because the \emph{Greedy} approach is deterministic
and thus the \emph{Greedy} defender's strategies are easy to be learned by the DRL attacker.
Although the \emph{Greedy-Mixed} approach introduces randomization to some extent,
the major component of the approach is still greedy.
Therefore, it is still not difficult for a powerful DRL attacker to learn
the \emph{Greedy-Mixed} defender's strategies.
However, when the defender employs our \emph{DP-based} approach,
the two types of attackers obtain almost the same utility.
To explain, our \emph{DP-based} approach introduces differentially private random noise into the configurations.
As analyzed in Section V-A, with the protection of differential privacy,
it is hard for an attacker to deduce the \emph{DP-based} defender's strategies
irrespective of the attacker's reasoning power.

Against the two types of attackers,
the defender uses almost the same cost.
This is because in the \emph{Greedy} and \emph{Greedy-Mixed} approaches,
defender's strategies are independent of the attacker's strategies.
Thus, the defender's cost is independent of the attacker's types.
In our \emph{DP-based} approach, the defender's strategies do take the attacker's strategies into consideration.
However, due to the use of differential privacy mechanisms,
the utility loss of the defender against the two attacker approaches is almost the same as shown in Figs. \ref{fig:S1Utility} and \ref{fig:S1DRL},
given that the defender's utility loss is identical to the attacker's utility gain.
Hence, according to Equations \ref{eq:utilityloss1}, \ref{eq:utilityloss2}, \ref{eq:utilityloss3}
and Algorithms \ref{alg:deployment1}, \ref{alg:deployment2}, \ref{alg:deployment3},
as the utility loss of the defender stays steady,
the defender's strategies are not affected much.
Thus, the defender's cost remains almost the same.
For simplicity, the figures regarding the defender's cost spent against the DRL attacker are not presented.
Moreover, in the remaining scenarios, the performance variation tendency of the three defender approaches against the DLR attacker
is similar to the tendency shown in Fig. \ref{fig:S1RoundDRL}.
For simplicity and clarity, they are not included in the paper.

\subsubsection{Scenario 2}
\begin{figure}[ht]
\vspace{-4mm}
\centering
	\begin{minipage}{0.45\textwidth}
   \subfigure[\scriptsize{The attacker's utility with different numbers of configurations}]{
    \includegraphics[width=0.45\textwidth, height=3cm]{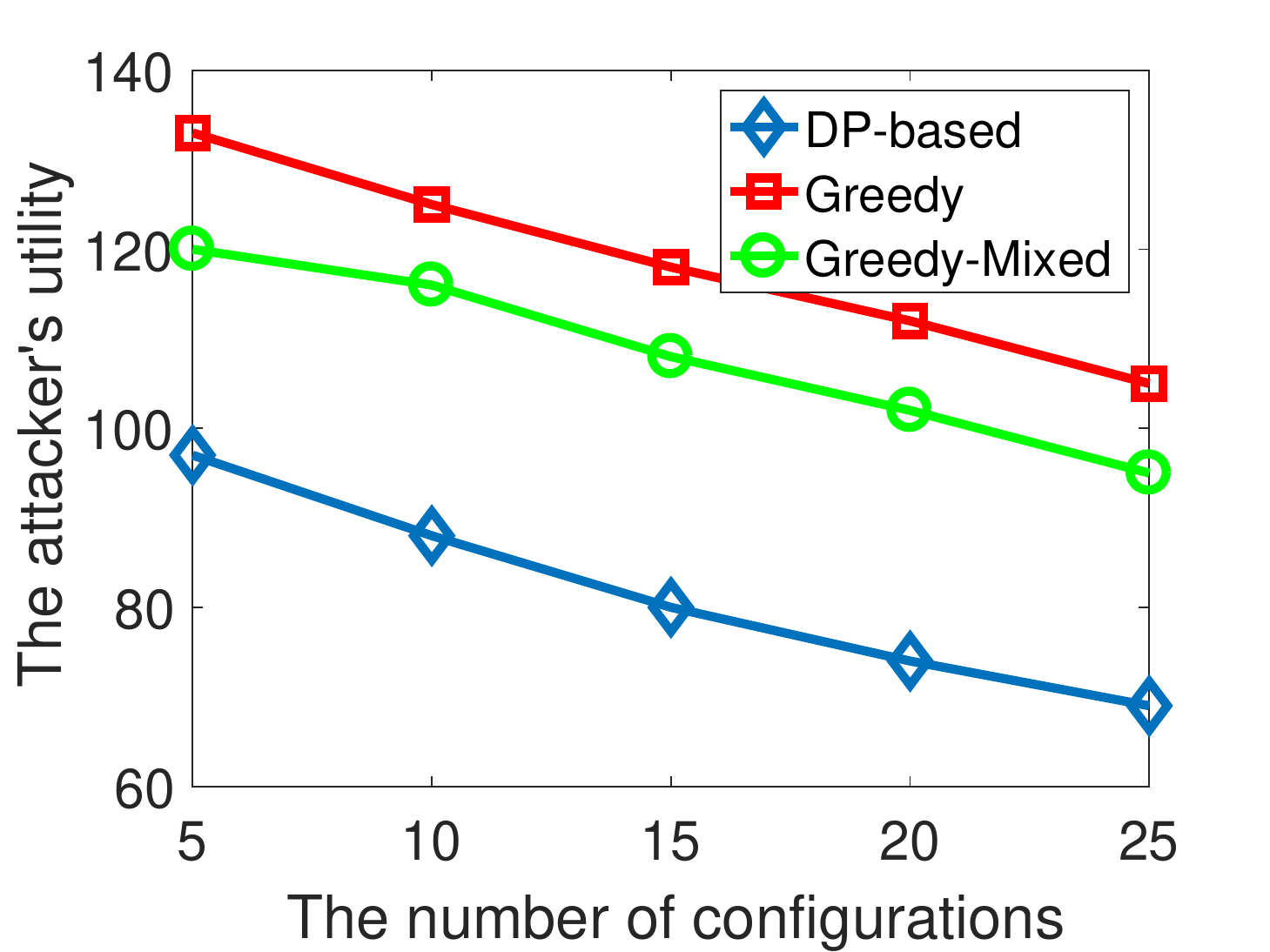}
			\label{fig:S2Utility}}
   \subfigure[\scriptsize{The defender's cost with different numbers of configurations}]{
    \includegraphics[width=0.45\textwidth, height=3cm]{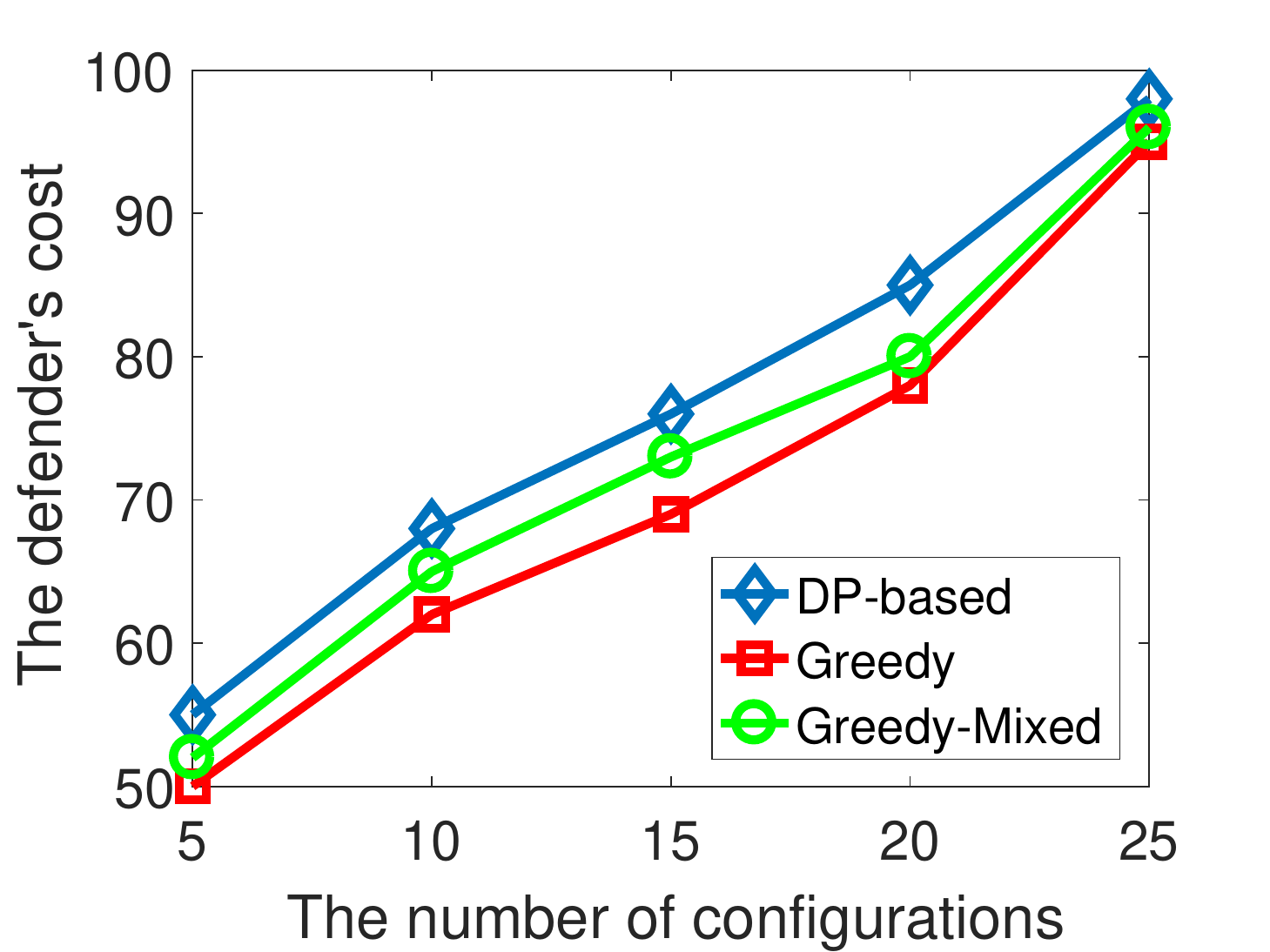}
			\label{fig:S2Cost}}\\[2ex]
    \end{minipage}
    \vspace{-4mm}
	\caption{The three approaches' performance in Scenario 2}
\vspace{-2mm}
	\label{fig:S2}
\end{figure}

Fig. \ref{fig:S2} demonstrates the performance of the three approaches in Scenario 2.
The number of systems is fixed at $150$,
and the number of configurations varies from $5$ to $25$.
The privacy budget $\epsilon$ is fixed at $0.3$.
The cost budget $B_d$ is fixed at $1000$ for each round.

With the increase of the number of configurations,
the attacker's utility gain decreases while the defender's cost increases.
In Scenario 2, since the number of systems is fixed,
as the number of configurations increases, the average number of systems associated with each configuration decreases.
As discussed in Fig. \ref{fig:S1}, the attacker's utility gain is based on the number of systems associated with the attacked configuration.
Thus, the attacker's utility gain decreases.
Moreover, as the number of configurations increases, the defender is more likely to obfuscate a system.
Thus, the defender's cost increases.
For example, in our \emph{DP-based} approach, when there are two configurations: $1$ and $2$,
the defender may obfuscate a system from configuration $1$ to appear as $2$ with probability $0.5$.
However, when there are three configurations: $1$, $2$ and $3$,
the defender may obfuscate a system from configuration $1$ to appear as $2$ or $3$ with the same probability of $0.33$, or $0.66$ altogether.
This example shows that as the number of configurations increases,
the probability that the defender will obfuscate a system will increase.

\begin{figure}[ht]
\vspace{-6mm}
\centering
	\begin{minipage}{0.51\textwidth}
    \subfigure[\scriptsize{The attacker's utility as game progresses in \emph{DP-based}}]{
    \includegraphics[width=0.45\textwidth, height=3cm]{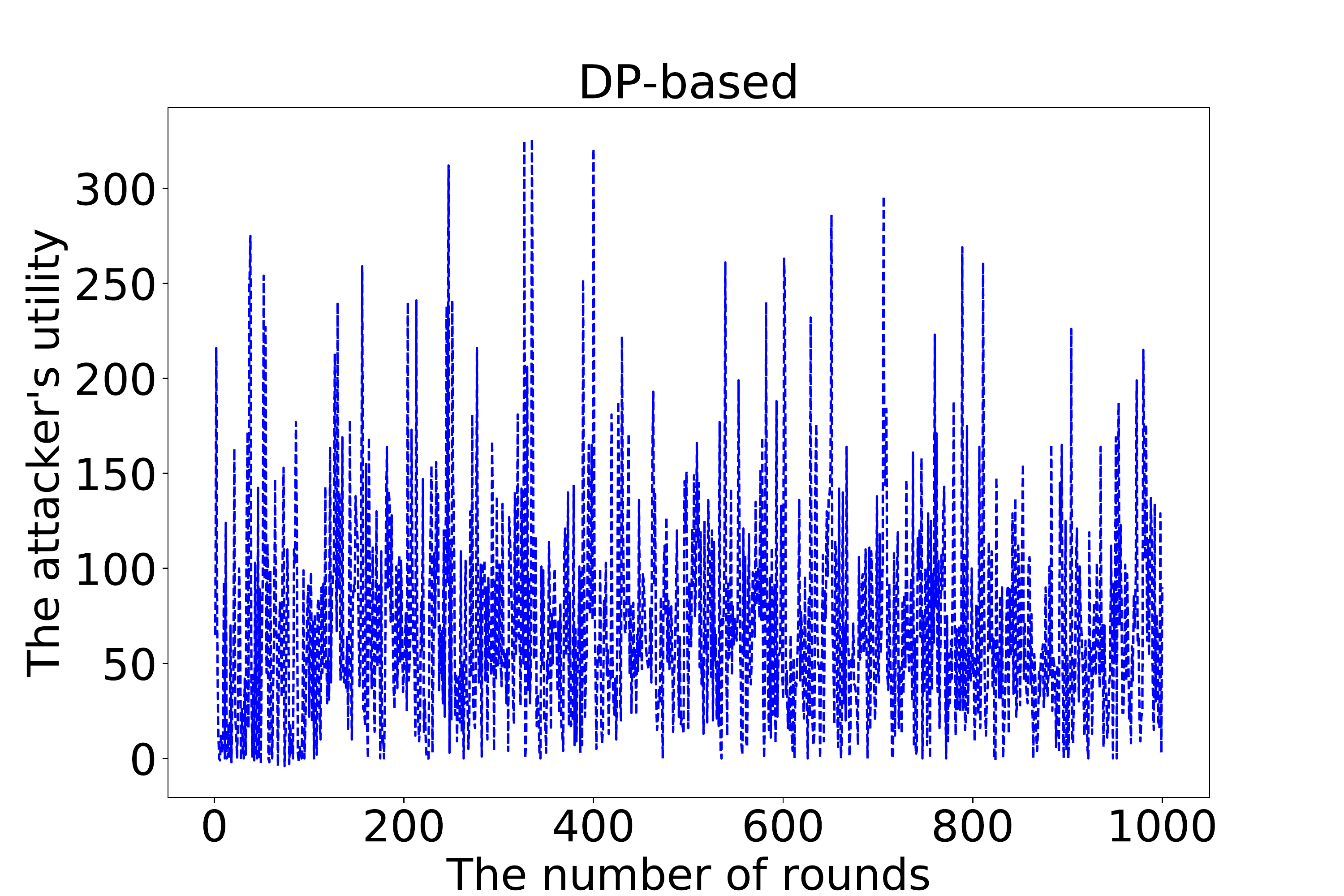}
			\label{fig:S2DPUtility}}
	\subfigure[\scriptsize{The defender's cost as game progresses in \emph{DP-based}}]{
    \includegraphics[width=0.45\textwidth, height=3cm]{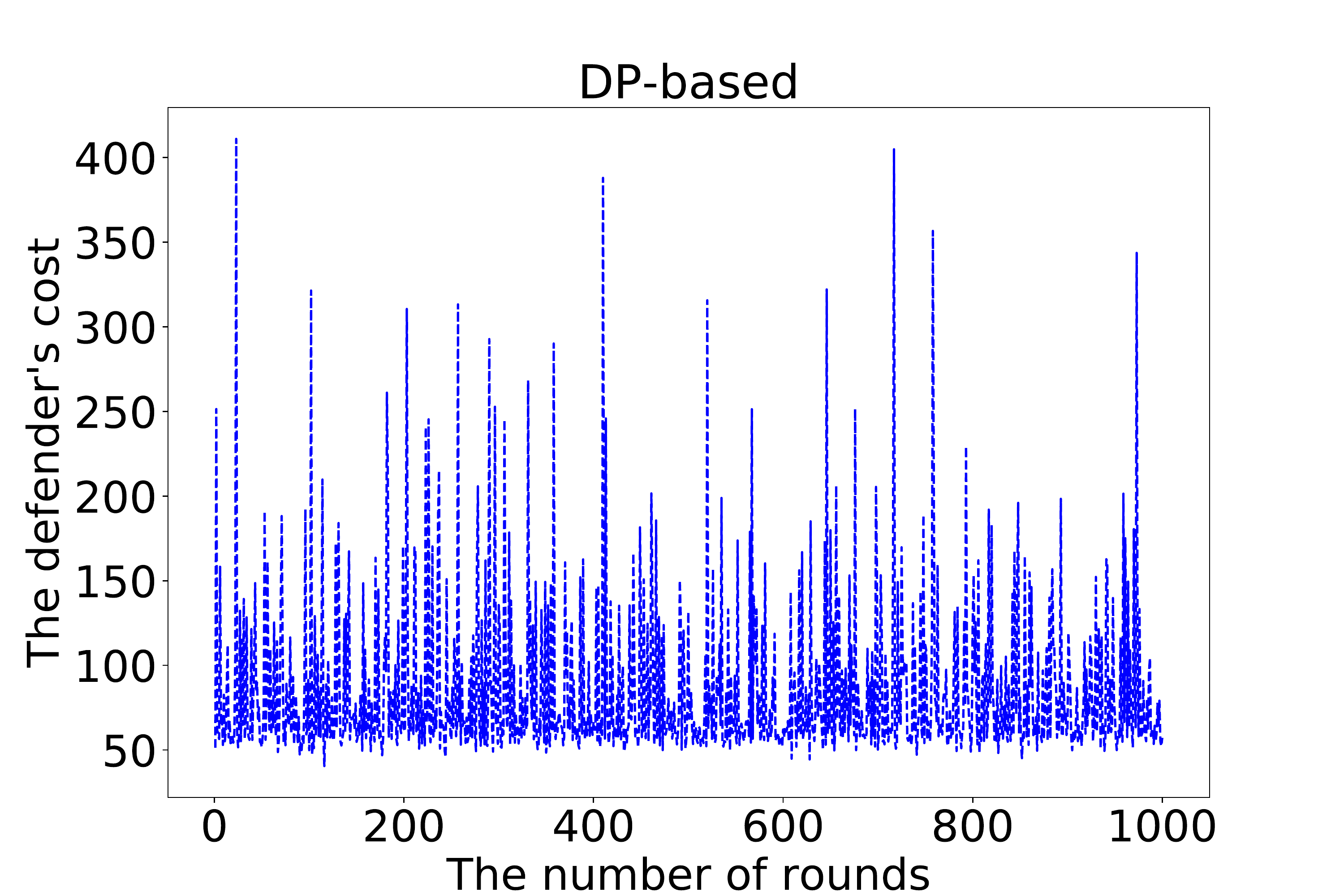}
			\label{fig:S2DPCost}}\\[2ex]
			
    \subfigure[\scriptsize{The attacker's utility as game progresses in \emph{Greedy}}]{
    \includegraphics[width=0.45\textwidth, height=3cm]{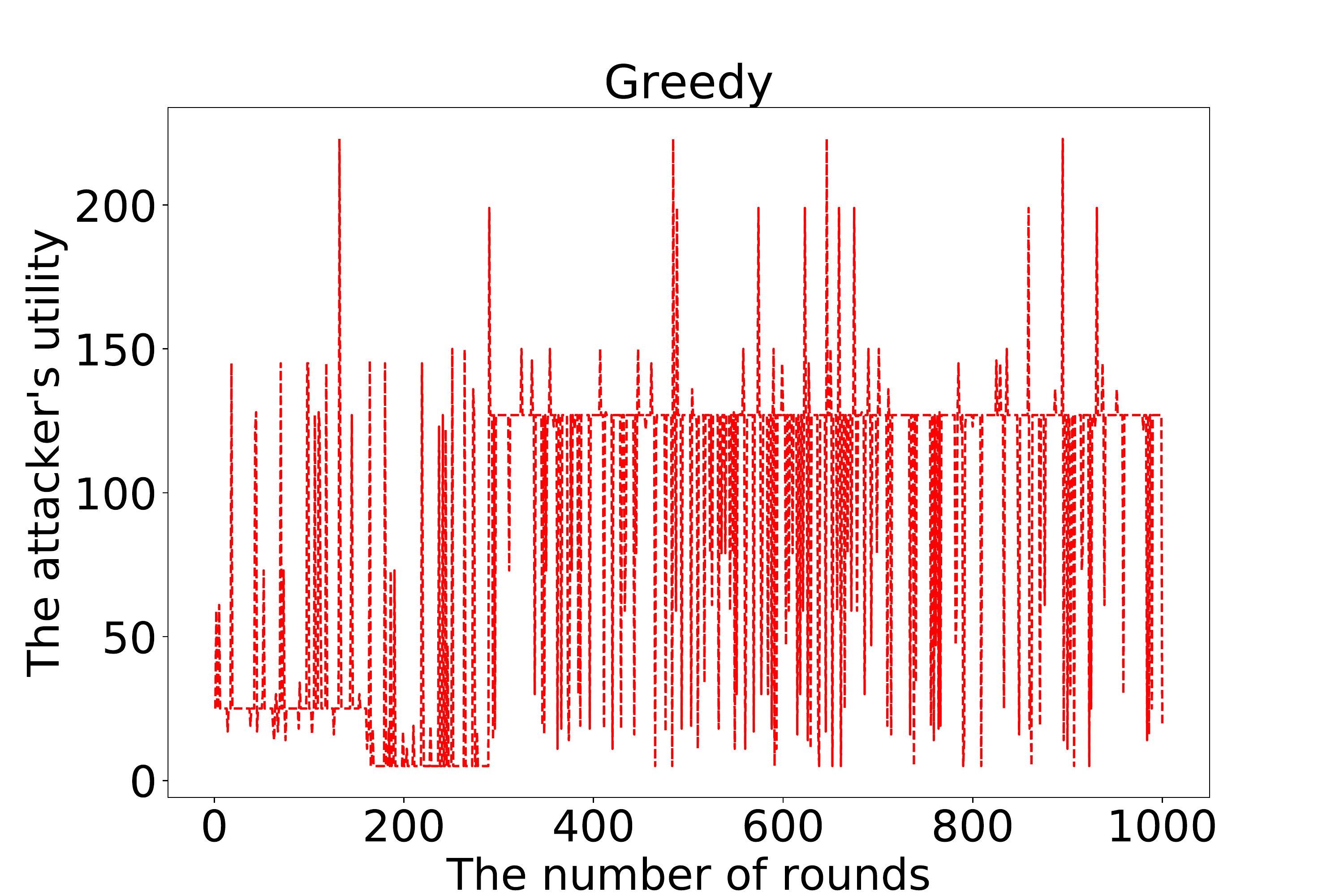}
			\label{fig:S2GreedyUtility}}
	\subfigure[\scriptsize{The defender's cost as game progresses in \emph{Greedy}}]{
    \includegraphics[width=0.45\textwidth, height=3cm]{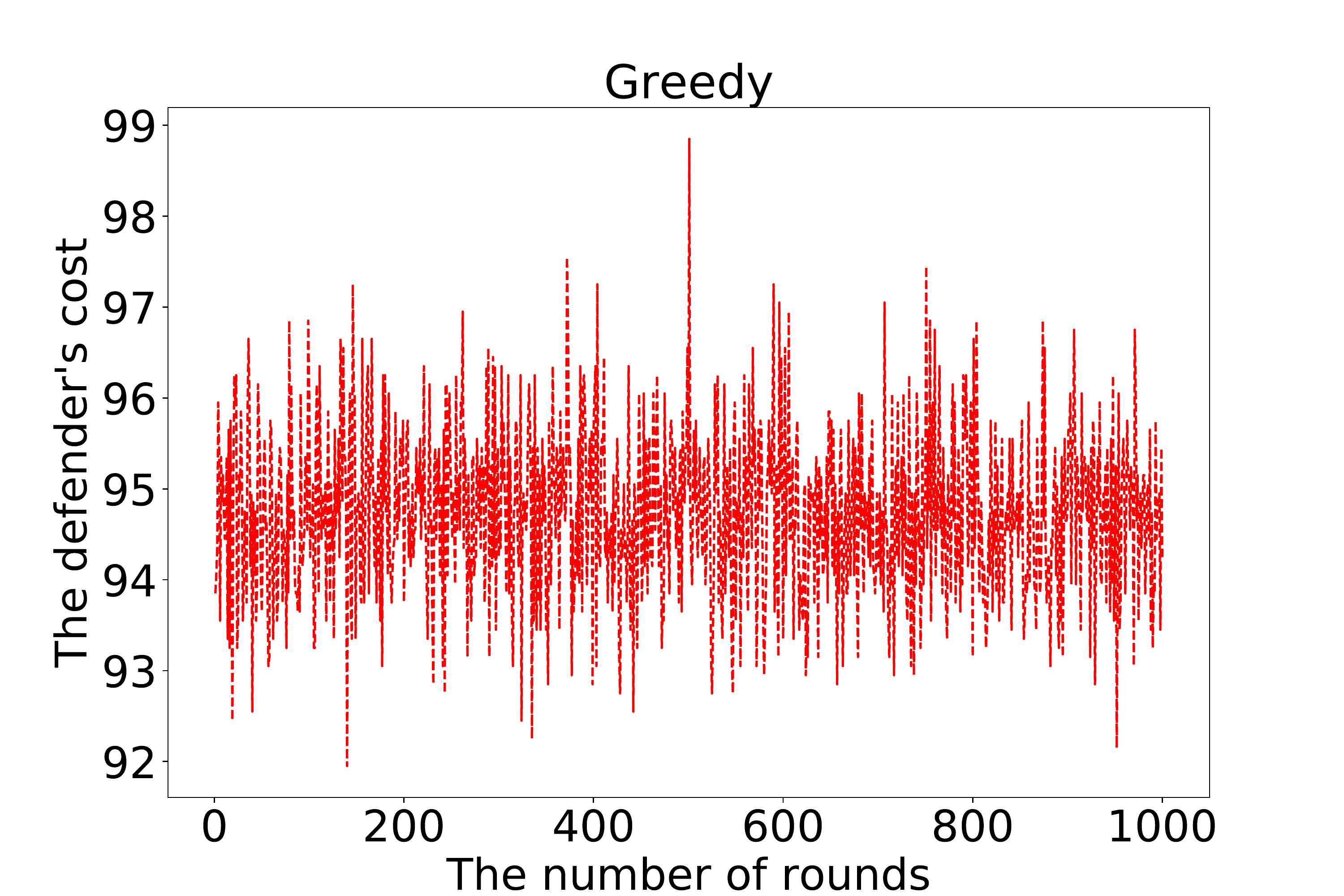}
			\label{fig:S2GreedyCost}}\\[2ex]
			
    \subfigure[\scriptsize{The attacker's utility as game progresses in \emph{Greedy-Mixed}}]{
    \includegraphics[width=0.45\textwidth, height=3cm]{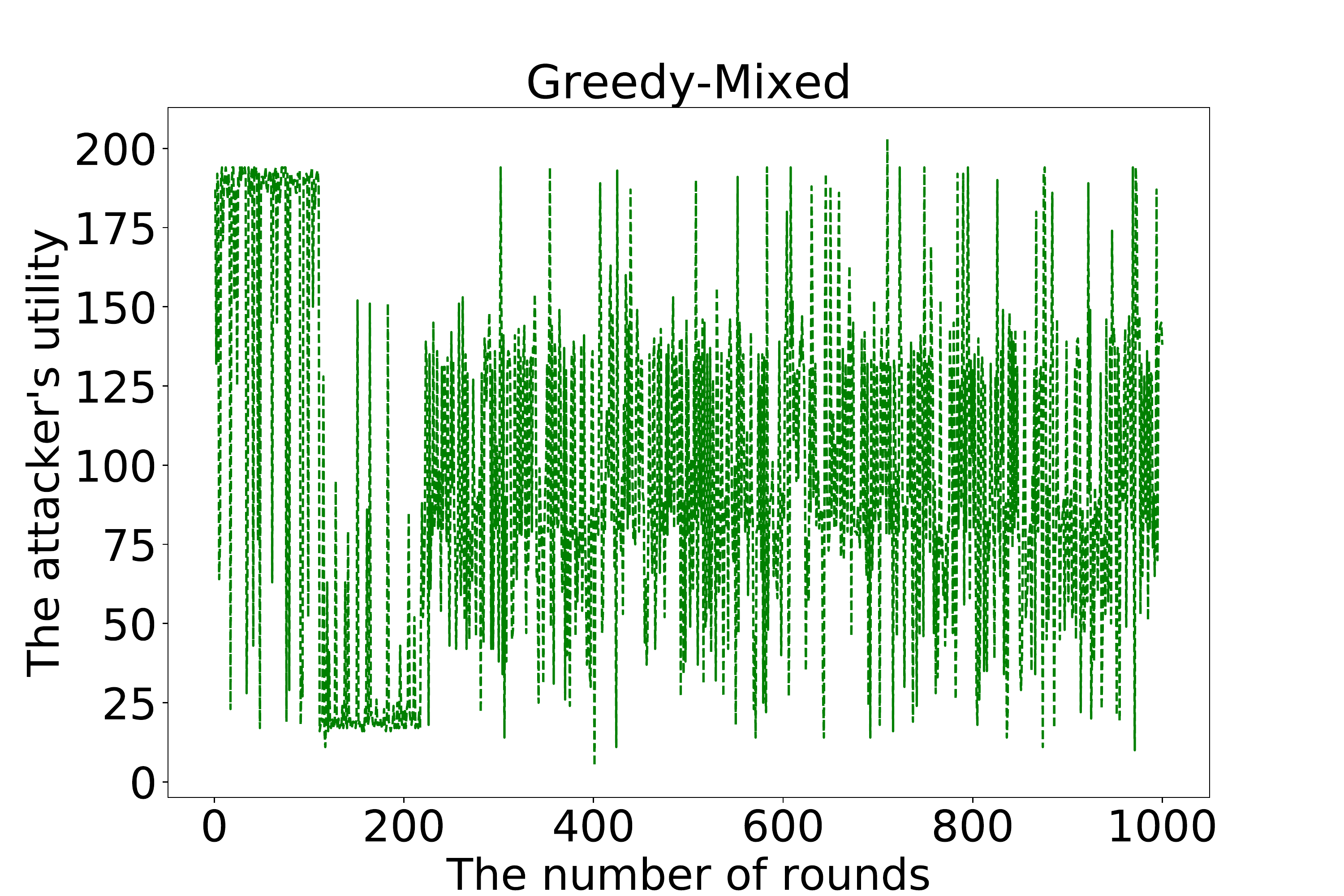}
			\label{fig:S2GreedyMixedUtility}}
    \subfigure[\scriptsize{The defender's cost as game progresses in \emph{Greedy-Mixed}}]{
    \includegraphics[width=0.45\textwidth, height=3cm]{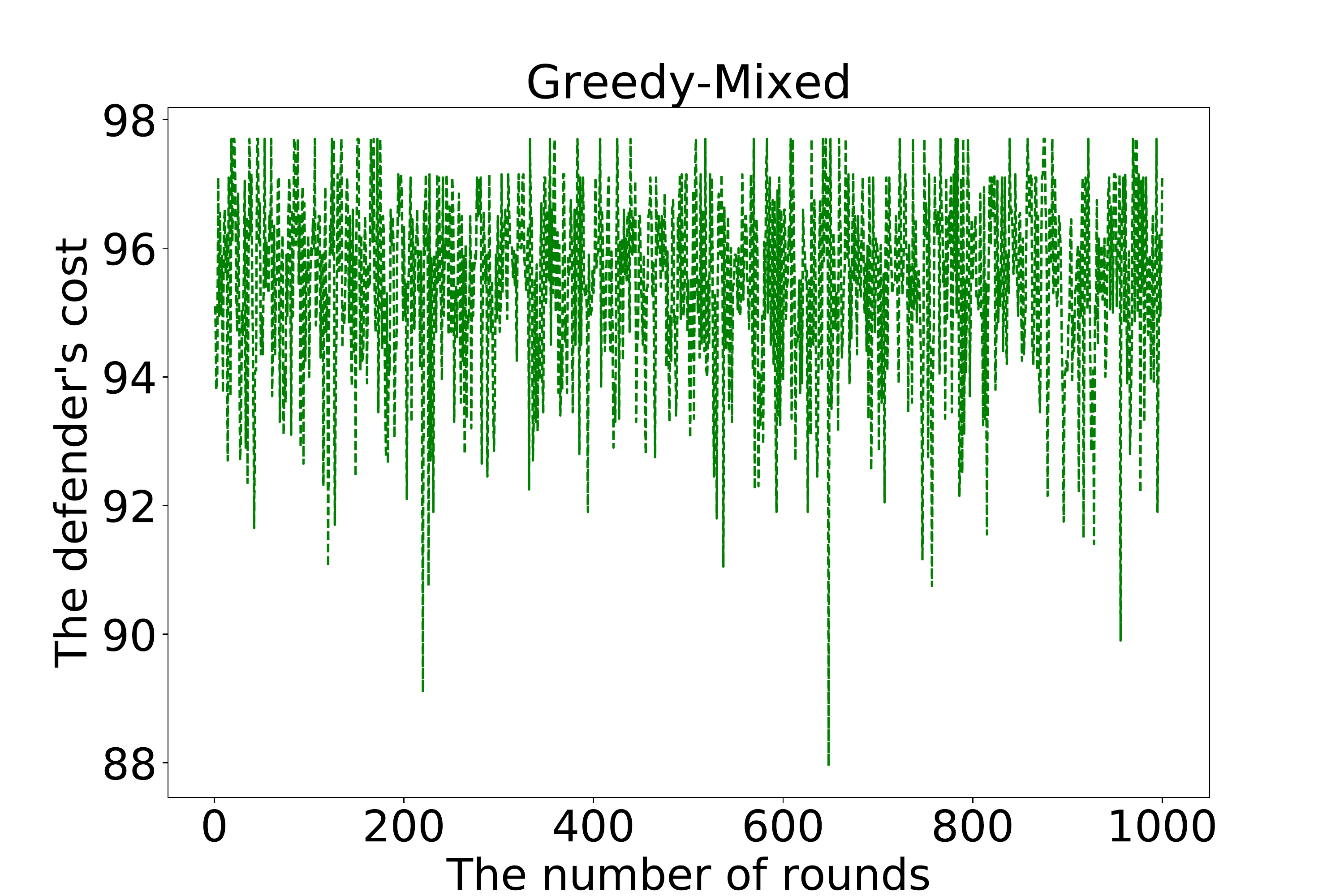}
			\label{fig:S2GreedyMixedCost}}\\
    \end{minipage}
    \vspace{-4mm}
	\caption{The three approaches' performance as game progress in Scenario 2}
\vspace{-2mm}
	\label{fig:S2Round}
\end{figure}
Fig. \ref{fig:S2Round} shows the variation in the attacker's utility gain
and the defender's cost in the three approaches as the game progresses in Scenario 2.
The number of systems is fixed at $150$,
and the number of configurations is fixed at $25$.
Fig. \ref{fig:S2Round} has a similar trend to Fig. \ref{fig:S1Round},
because Scenario 2 has a similar setting to Scenario 1,
where budget $B_d$ is large enough to cover the whole obfuscation process.

\subsubsection{Scenario 3}
\begin{figure}[ht]
\vspace{-4mm}
\centering
	\begin{minipage}{0.45\textwidth}
   \subfigure[\scriptsize{The attacker's utility with different values of budget $B_d$}]{
    \includegraphics[width=0.45\textwidth, height=3cm]{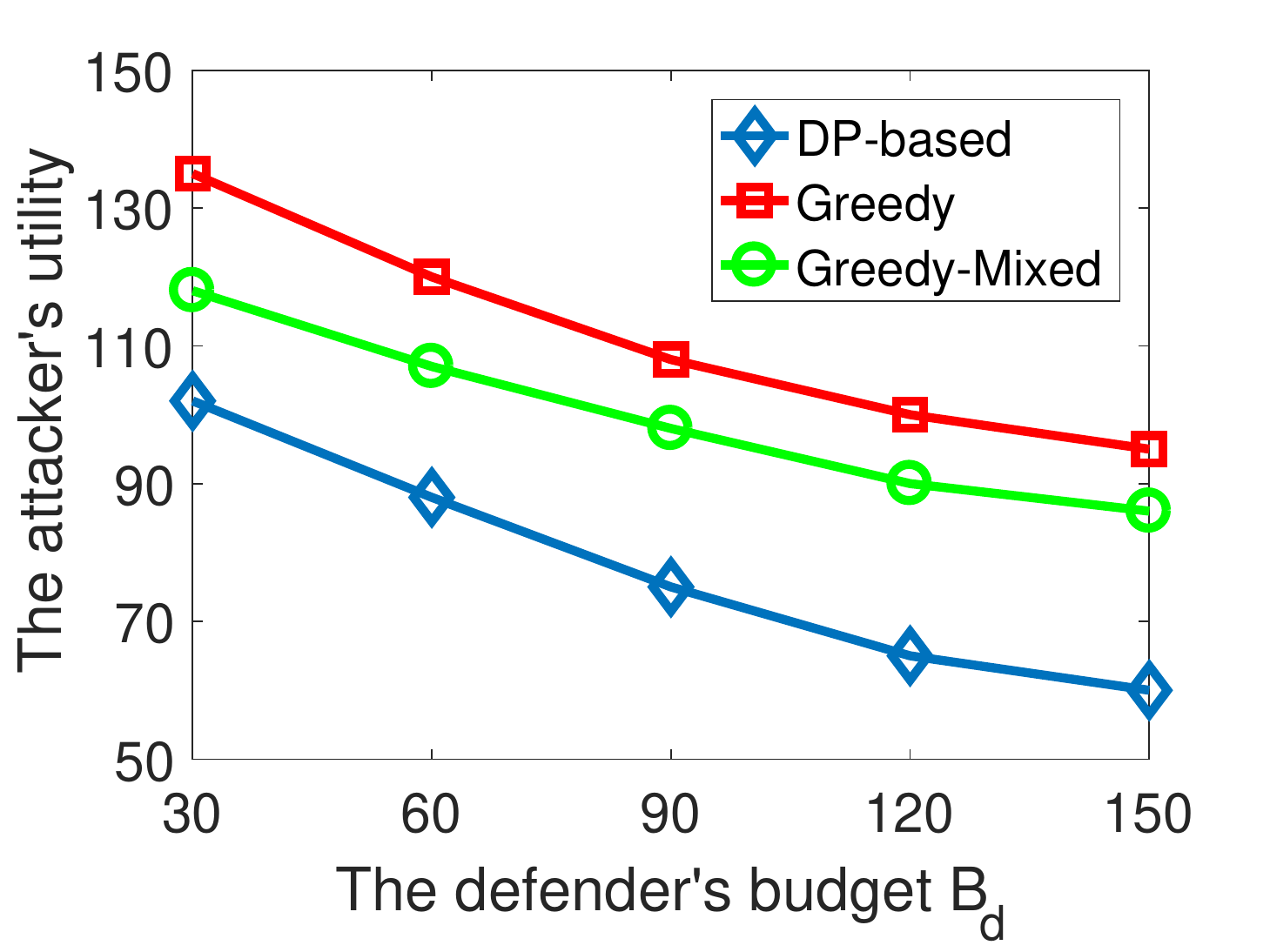}
			\label{fig:S3Utility}}
    \subfigure[\scriptsize{The defender's cost with different values of budget $B_d$}]{
    \includegraphics[width=0.45\textwidth, height=3cm]{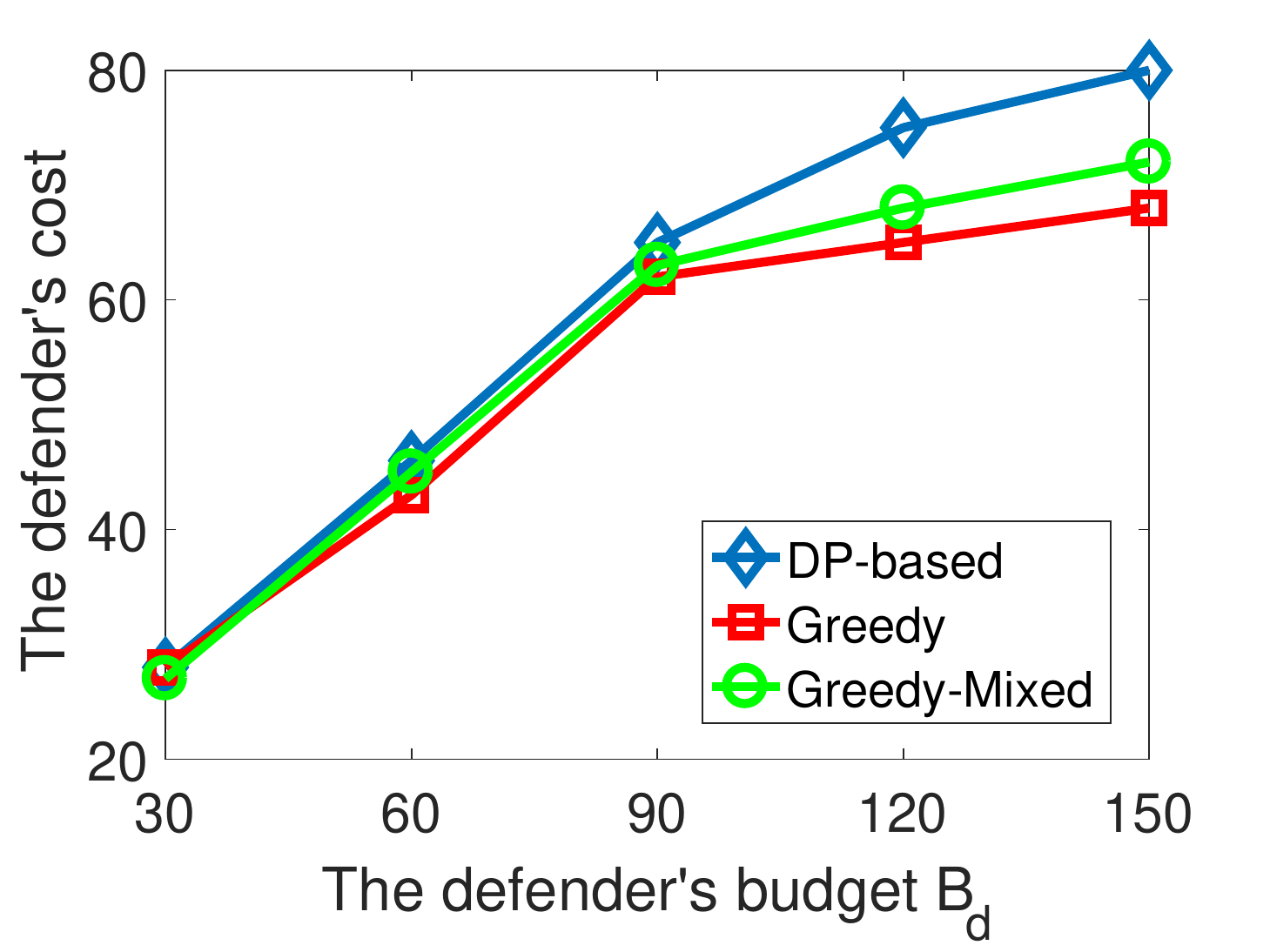}
			\label{fig:S3Cost}}\\[2ex]
    \end{minipage}
    \vspace{-4mm}
	\caption{The three approaches' performance in Scenario 3}
\vspace{-4mm}
	\label{fig:S3}
\end{figure}

Fig. \ref{fig:S3} demonstrates the performance of the three approaches in Scenario 3.
The numbers of systems and configurations are fixed at $100$ and $10$, respectively.
The privacy budget $\epsilon$ is fixed at $0.3$.
The cost budget $B_d$ for each round varies from $30$ to $150$.

When the defender's budget $B_d$ is tight (less than $90$),
the defender's cost can be properly controlled (Fig. \ref{fig:S3Cost}).
However, the attacker will gain high utility (Fig. \ref{fig:S3Utility}).
Due to the budget limitation (less than $90$), the defender obfuscates only a few systems,
meaning that the attacker can easily select high-utility systems.
By contrast, when the defender's budget $B_d$ is large enough (greater than $90$),
she can obfuscate almost all necessary systems,
which increases the difficulty experienced by the attacker when attempting to select high-utility systems.

\begin{figure}[ht]
\vspace{-4mm}
\centering
	\begin{minipage}{0.45\textwidth}
   \subfigure[\scriptsize{The attacker's utility as game progresses in \emph{DP-based}, $B_d=60$}]{
    \includegraphics[width=0.45\textwidth, height=3cm]{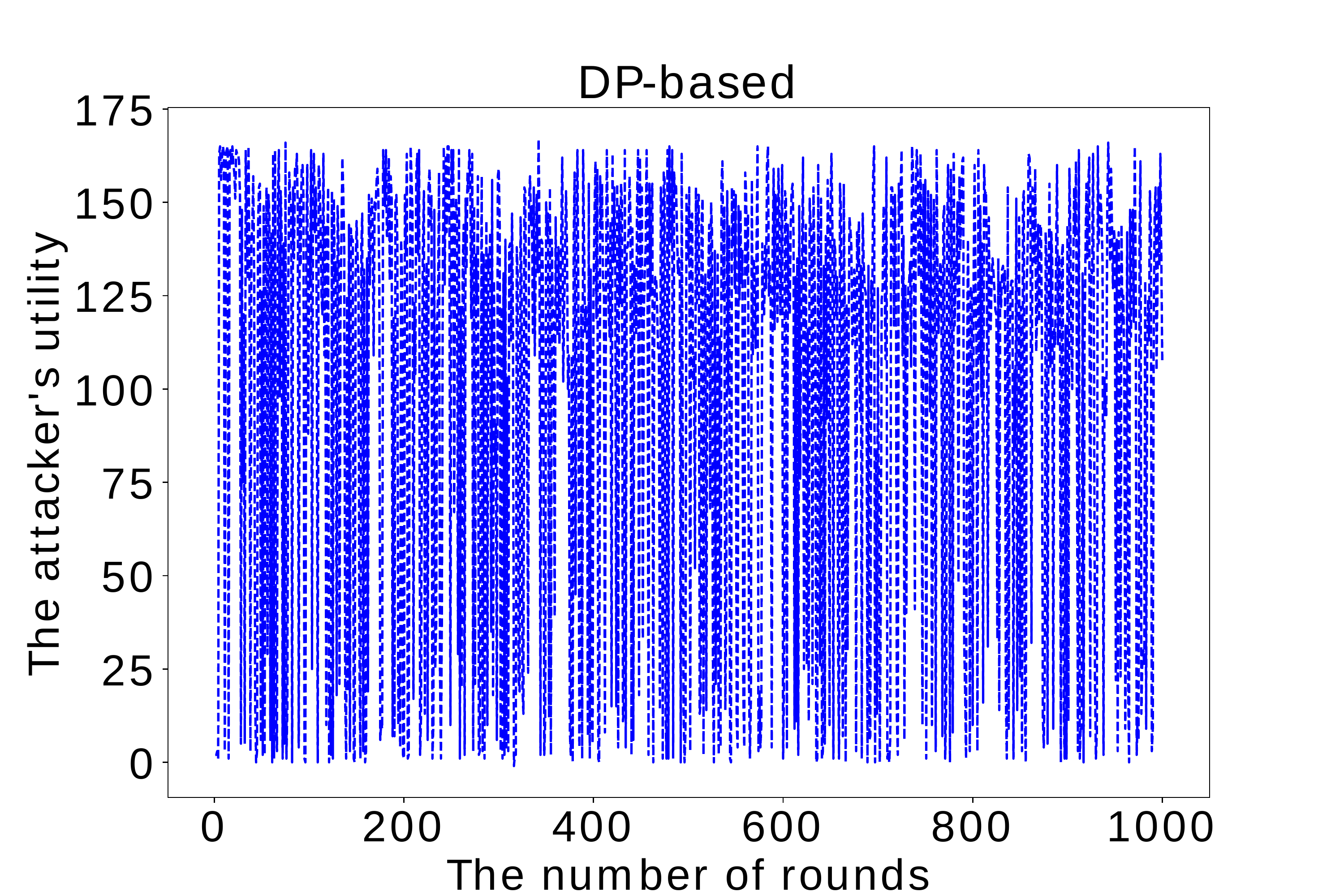}
			\label{fig:S3Utility60}}
    \subfigure[\scriptsize{The attacker's utility as game progresses in \emph{Greedy}, $B_d=150$}]{
    \includegraphics[width=0.45\textwidth, height=3cm]{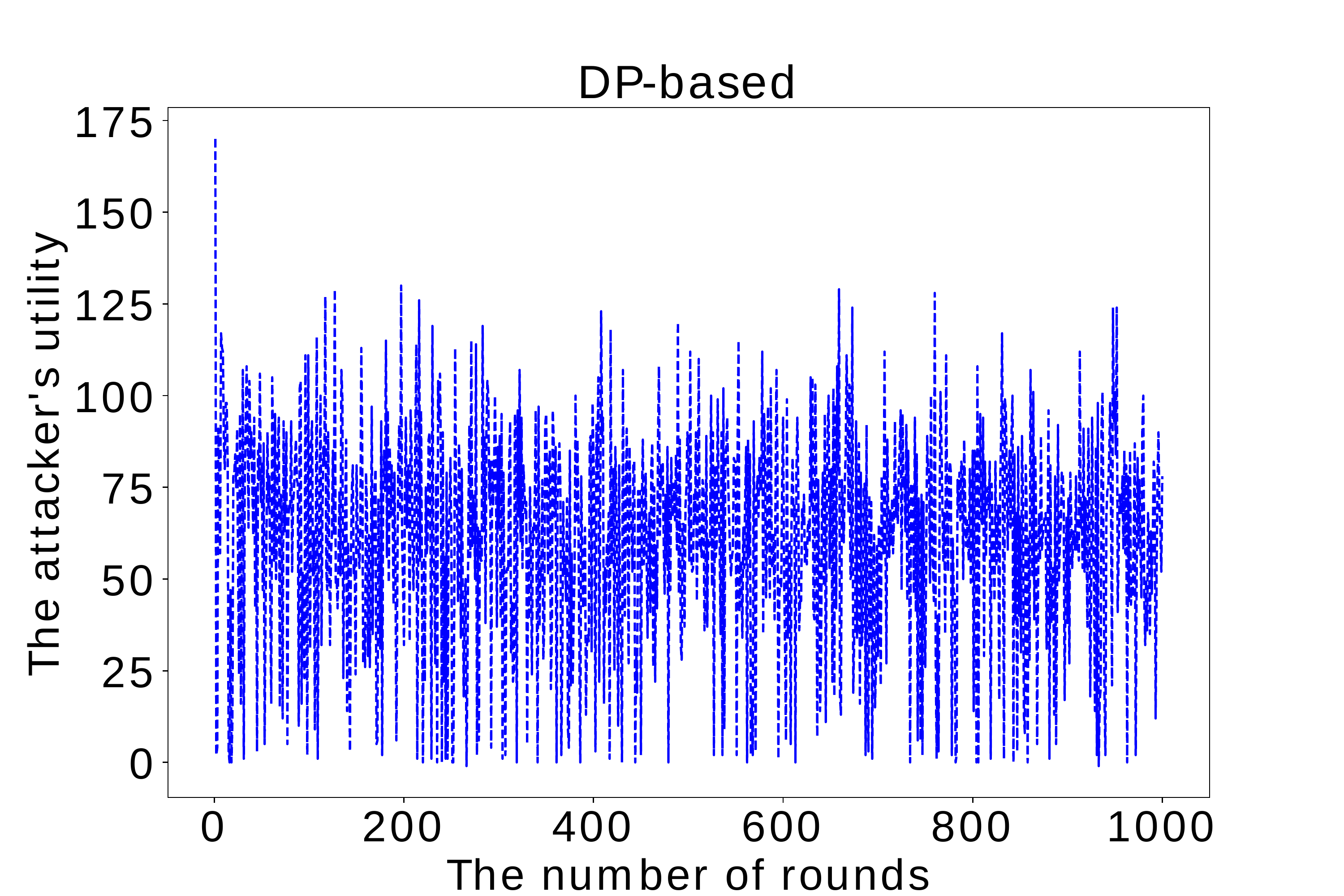}
			\label{fig:S3Utility150}}\\[2ex]

    \subfigure[\scriptsize{The defender's cost as game progresses in \emph{DP-based}, $B_d=60$}]{
    \includegraphics[width=0.45\textwidth, height=3cm]{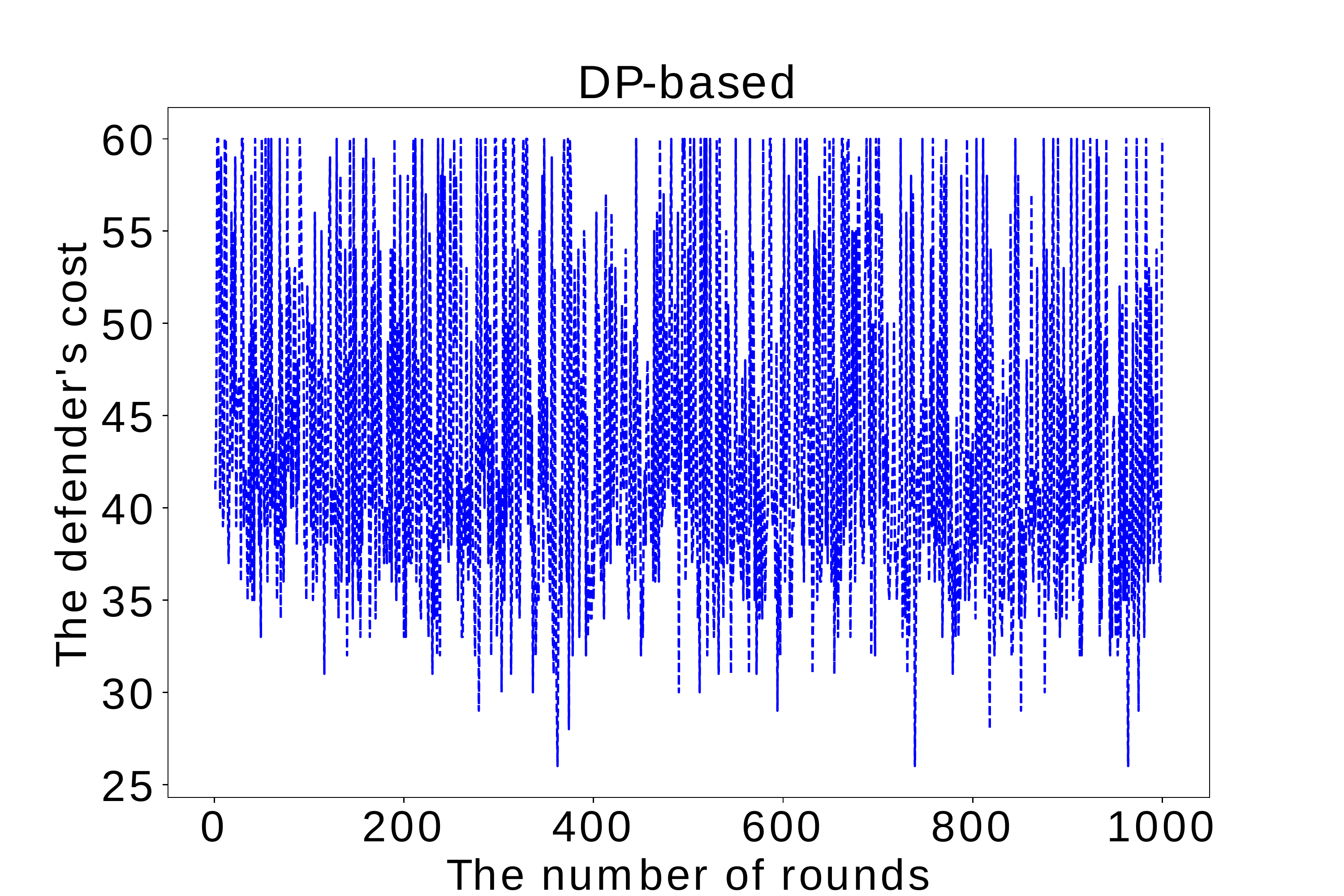}
			\label{fig:S3Cost60}}
    \subfigure[\scriptsize{The defender's cost as game progresses in \emph{Greedy}, $B_d=150$}]{
    \includegraphics[width=0.45\textwidth, height=3cm]{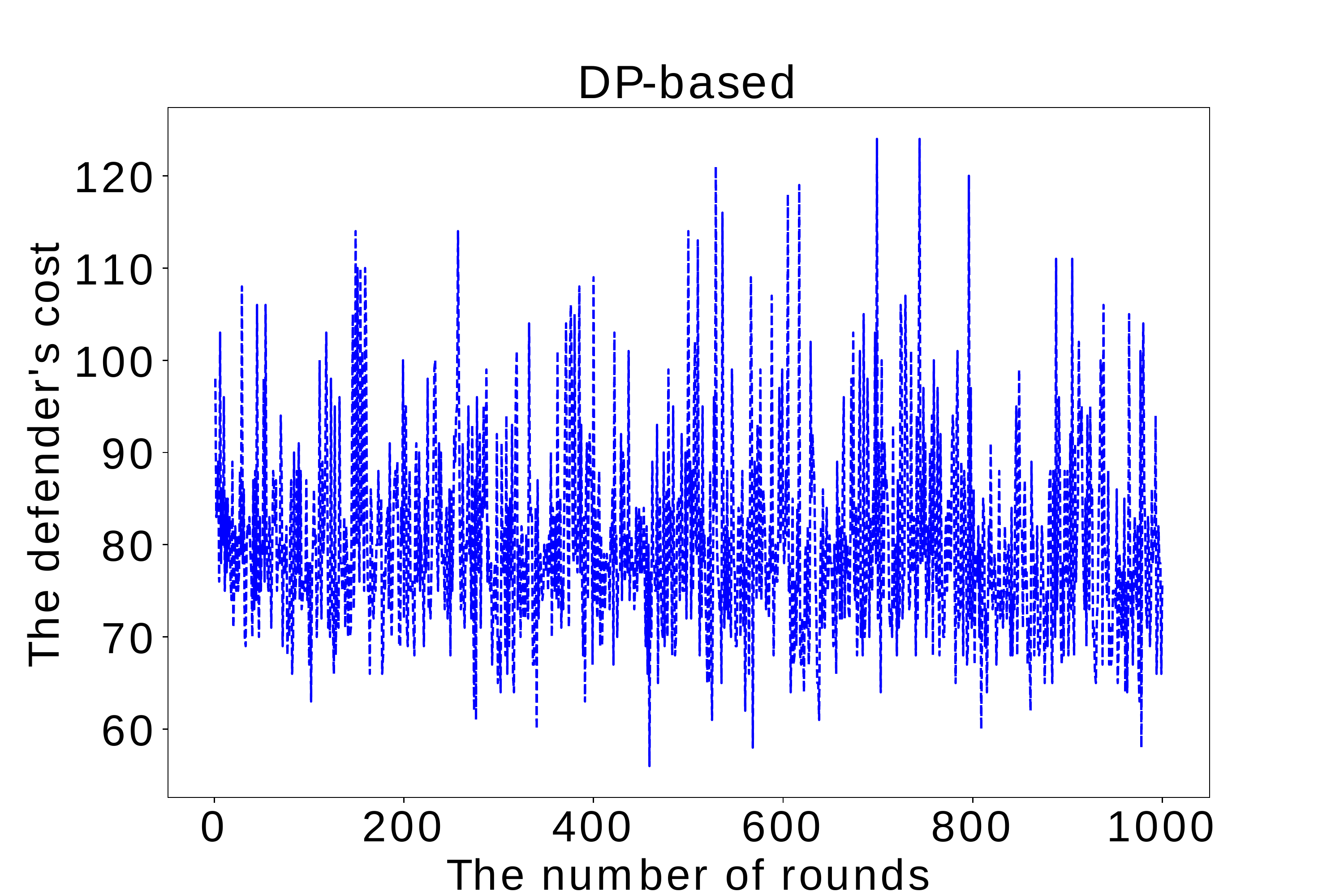}
			\label{fig:S3Cost150}}\\[2ex]
    \end{minipage}
    \vspace{-4mm}
	\caption{Performance of the \emph{DP-based} approach with different cost budgets}
\vspace{-2mm}
	\label{fig:S3Budget}
\end{figure}
Fig. \ref{fig:S3Budget} shows the details of the situations of $B_d=60$ and $B_d=150$ in our \emph{DP-based} approach.
The details in the \emph{Greedy} and \emph{Greedy-Mixed} approaches are similar to the \emph{DP-based} approach,
which thus are not presented.

By comparing Fig. \ref{fig:S3Utility60} and Fig. \ref{fig:S3Utility150},
the attacker's utility significantly reduces from $88$ to $57$ on average,
when increasing the defender's budget from $60$ to $150$.
Moreover, in Fig. \ref{fig:S3Cost60}, the defender's cost is confined to $60$,
while in Fig. \ref{fig:S3Cost150}, there is no confinement on the defender's cost.
This is because budget $B_d=60$ is insufficient for the defender to obfuscate all the necessary systems in most rounds.
Therefore, budget $B_d=60$ becomes a confinement for the defender.
When the budget is increased to $150$,
budget $B_d=150$ is enough to cover the obfuscation of all the necessary systems.
Thus, the confinement disappears.

\subsubsection{Scenario 4}
\begin{figure}[ht]
\vspace{-4mm}
\centering
	\begin{minipage}{0.45\textwidth}
   \subfigure[\scriptsize{The attacker's utility with different values of privacy budget $\epsilon$}]{
    \includegraphics[width=0.45\textwidth, height=3cm]{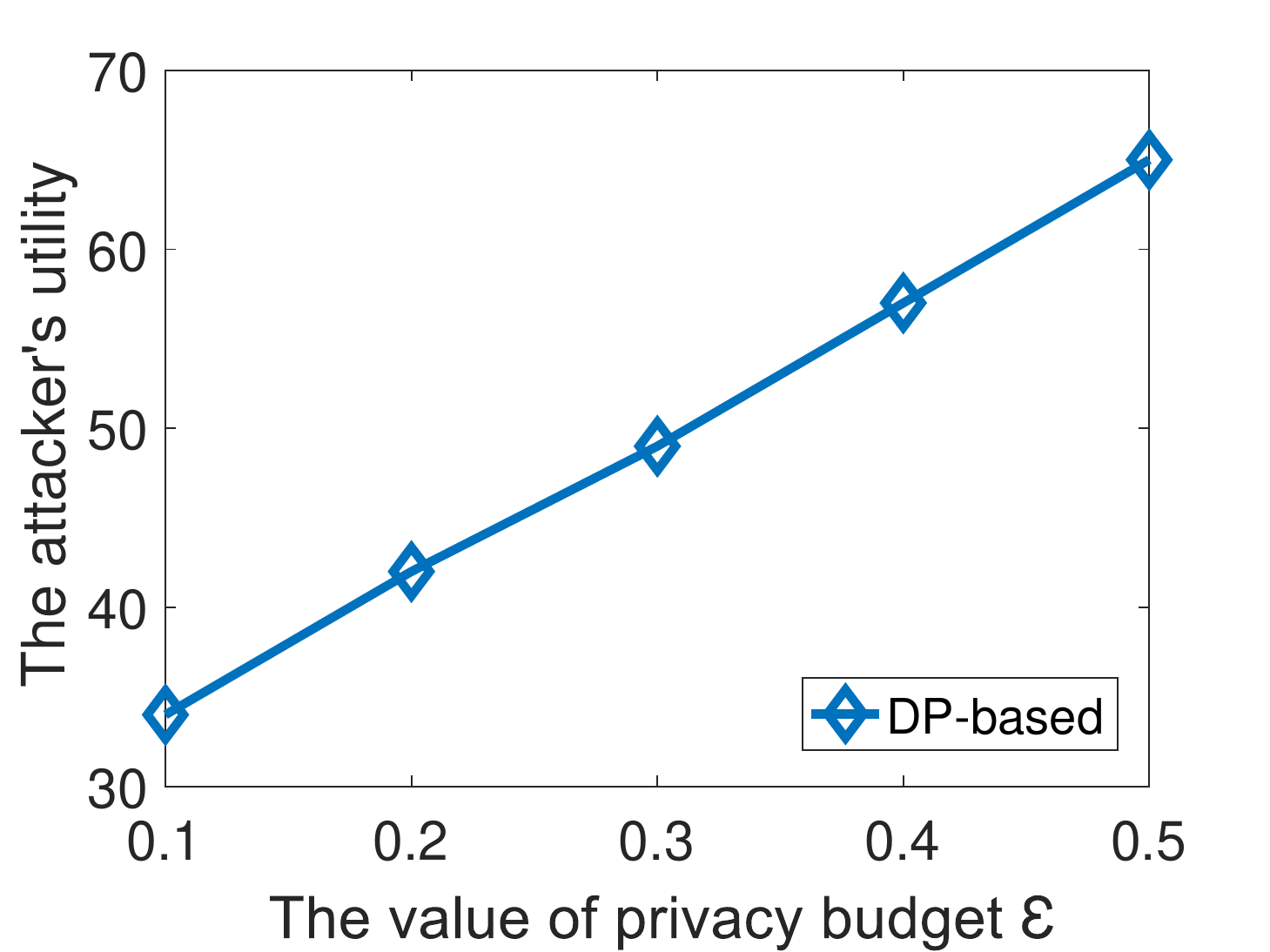}
			\label{fig:S4Utility}}
    \subfigure[\scriptsize{The defender's cost with different values of privacy budget $\epsilon$}]{
    \includegraphics[width=0.45\textwidth, height=3cm]{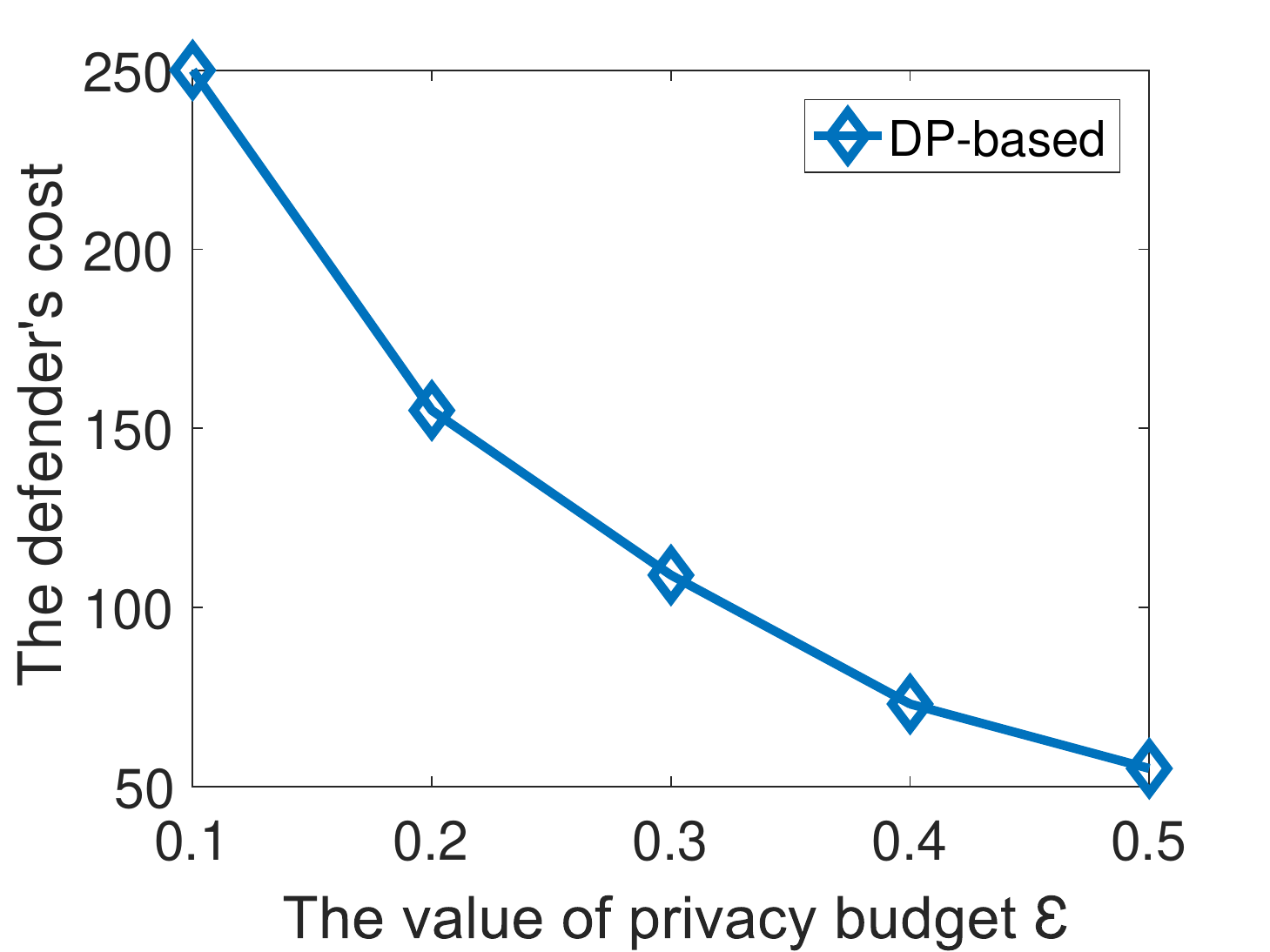}
			\label{fig:S4Cost}}\\[2ex]
    \end{minipage}
    \vspace{-4mm}
	\caption{Our approach's performance in Scenario 4}
\vspace{-2mm}
	\label{fig:S4}
\end{figure}

Fig. \ref{fig:S4} demonstrates the performance of our approach in Scenario 4.
The number of systems is fixed at $150$,
and the number of configurations is fixed at $15$.
The privacy budget $\epsilon$ various from $0.1$ to $0.5$.
The cost budget $B_d$ is fixed at $1000$ for each round.

According to Definition \ref{Def-LA} in Sub-section \ref{sub:DP}, a smaller $\epsilon$ denotes a larger Laplace noise.
In Line 5 of Algorithm \ref{alg:DP}, the number of systems associated with each configuration is adjusted by adding Laplace noise.
Therefore, a larger Laplace noise implies a larger change in the number of systems.
On one hand, if the Laplace noise is positive,
the number of systems increases ($|N'|>|N|$).
The defender thus incurs extra cost to deploy the added systems, i.e., honeypots.
These honeypots will attract the attacker and decrease his utility gain.
On the other hand, if the Laplace noise is negative,
the number of systems decreases ($|N'|<|N|$).
The defender sacrifices some utility to take $|N|-|N'|$ systems offline.
These systems can avoid the attacker,
and the attacker's utility gain reduces.

\begin{figure}[ht]
\vspace{-4mm}
\centering
	\begin{minipage}{0.45\textwidth}
   \subfigure[\scriptsize{The attacker's utility as game progresses}]{
    \includegraphics[width=0.45\textwidth, height=3cm]{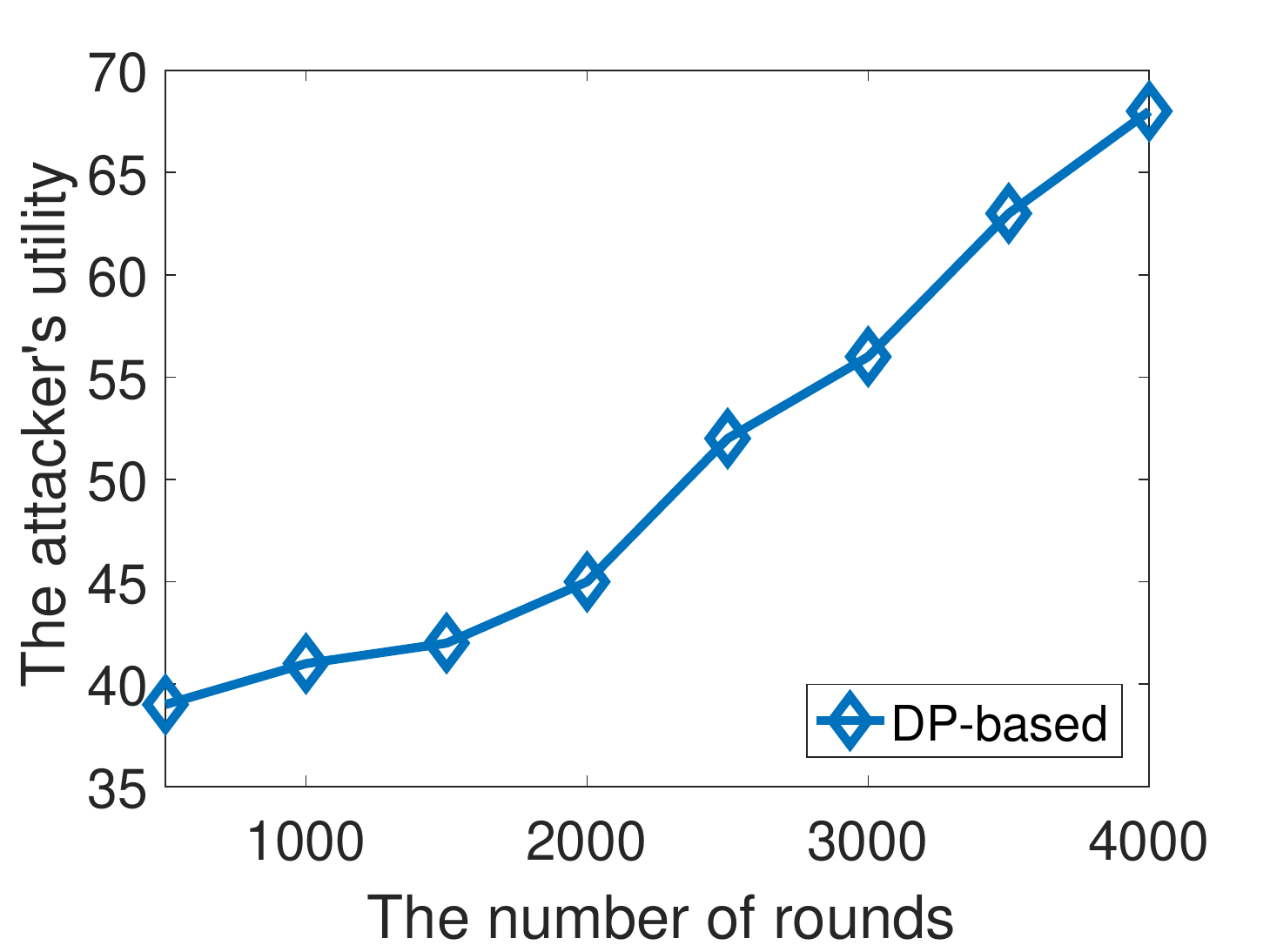}
			\label{fig:S4UtilityRounds}}
    \subfigure[\scriptsize{The defender's cost as game progresses}]{
    \includegraphics[width=0.45\textwidth, height=3cm]{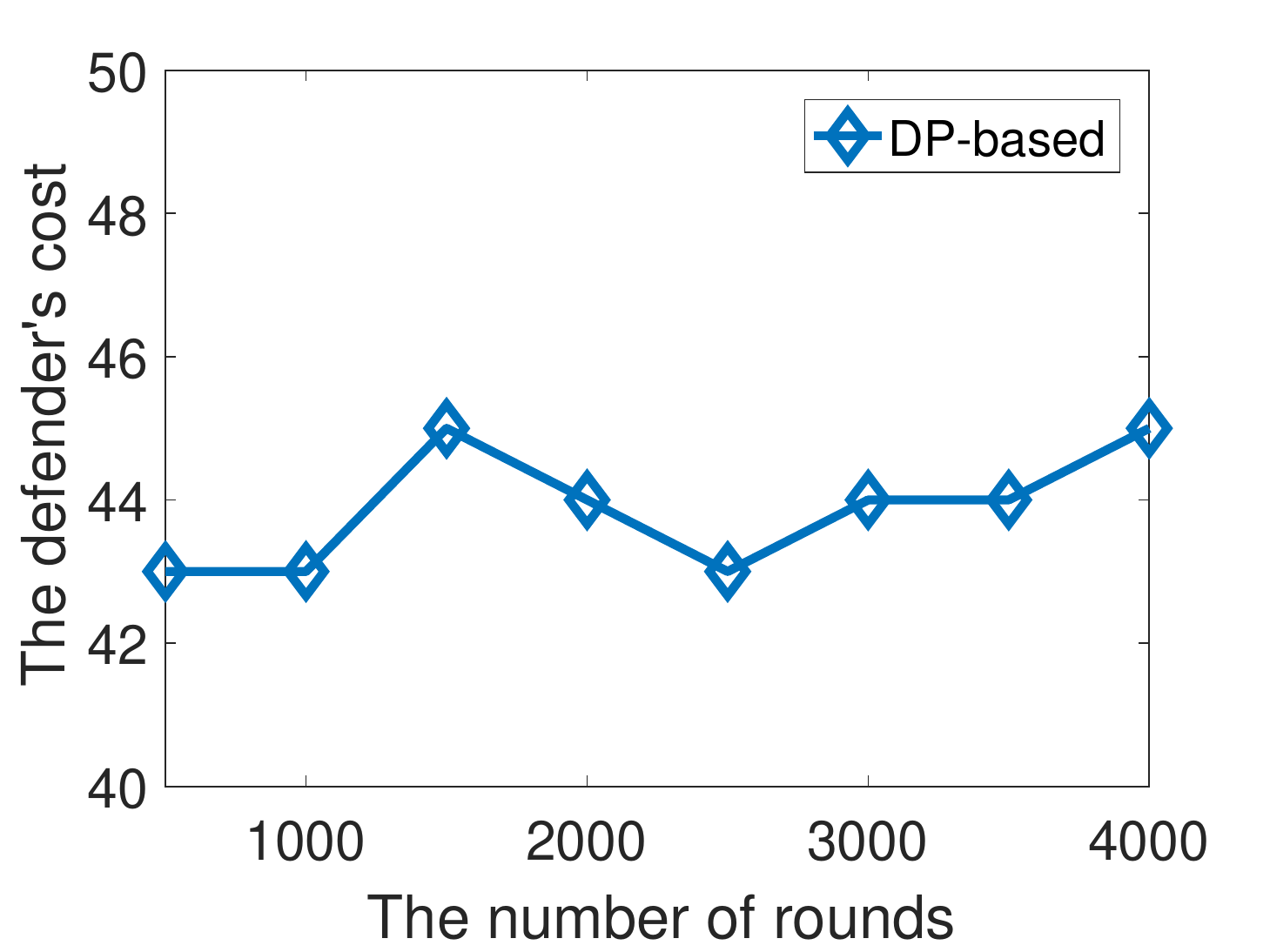}
			\label{fig:S4CostRounds}}\\[2ex]
    \end{minipage}
    \vspace{-4mm}
	\caption{Our approach with different number of game rounds in Scenario 4}
\vspace{-2mm}
	\label{fig:S4Rounds}
\end{figure}
Fig. \ref{fig:S4Rounds} demonstrates the performance of our approach with different number of game rounds,
where the privacy budget $\epsilon$ is fixed at $0.3$.
Based on the discussion in Theorem \ref{thm:bound}, to guarantee a given privacy level,
an upper bound of the number of rounds must be set.
In Fig. \ref{fig:S4UtilityRounds}, with the increasing number of rounds,
the attacker's utility gain increases gradually,
especially after $2000$ rounds.
This implies that when the game is played in a large number of rounds,
the attacker can infer the true configurations of systems. 
This experimental result empirically prove our theoretical result.
In Fig. \ref{fig:S4CostRounds}, the defender's cost is not affected by the number of rounds,
since both the number of systems and configurations is fixed.

\subsection{Summary}
According to the experimental results, due to the adoption of the differential privacy technique,
our \emph{DP-based} approach outperforms the \emph{Greedy} and \emph{Greedy-mixed} approaches in various scenarios
by significantly reducing the attacker's utility gain by about $30\%\sim 40\%$.
In terms of the overhead resulting from adopting the differential privacy technique,
the defender in our \emph{DP-based} approach incurs about $5\%\sim 8\%$ more cost than in the \emph{Greedy} and \emph{Greedy-mixed} approaches.
Moreover, the privacy budget $\epsilon$ can be used to balance the attacker's utility gain and the defender's cost.
A small value of $\epsilon$ means a low utility gain for the attacker, but a high cost to the defender.
The setting of the $\epsilon$ value is left to users.

Furthermore, the running times of the three approaches are almost the same.
Thus, they are not shown graphically in the experiments.
Based on the theoretical analysis in Section \ref{sec:analysis}, the complexity of our \emph{DP-based} approach is $O(|N|^2)$,
where $|N|$ is the number of systems.
In the \emph{Greedy} and \emph{Greedy-mixed} approaches, there is a ranking process to rank the utility of systems.
Since we use bubble sort to implement the \emph{Greedy} and \emph{Greedy-mixed} approaches,
the complexity of both approaches is $O(|N|^2)$
which is the same as our \emph{DP-based} approach.


\vspace{-1mm}
\section{Conclusion and Future Work}\label{sec:conclusion}
This paper proposes a novel differentially private game theoretic approach to cyber deception.
Our approach is the first to adopt the differential privacy technique,
which can efficiently handle any change in the number of systems
and complex strategies potentially adopted by an attacker.
Compared to the benchmark approaches,
our approach is more stable and achieves much better performance
in various scenarios with slightly higher cost.

In the future, as mentioned in Section \ref{sec:analysis},
we will deeply investigate the equilibria in highly dynamic security games.
We also intend to improve our approach by introducing multiple defenders and multiple attackers.
In this paper, we consider only one defender and one attacker.
Future games will be very interesting when multiple defenders and attackers are involved.
Moreover, we will develop a general prototype of our approach to allow the state of the art to progress faster.
Specifically, to develop a prototype, we intend to follow Jajodia et al.’s work \cite{Jajodia17}.
We will first use Cauldron software \cite{Jajodia16} in a real network environment
to scan for vulnerabilities on each node in the network.
Then we will replicate the network into a larger network with the NS2 network simulator \cite{Issariyakul08}.
After that, our approach will be implemented on the simulated network.

\bibliographystyle{IEEEtran}
{\small \bibliography{references}}

\begin{IEEEbiography}[{\includegraphics[scale=0.25]{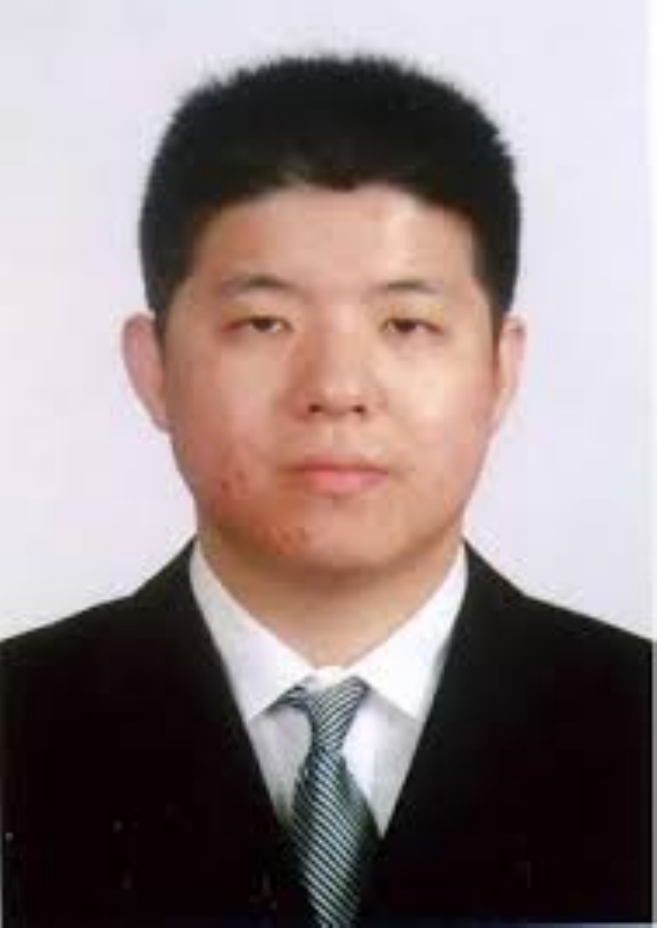}}]{Dayong Ye}
received his MSc and PhD degrees both from University of Wollongong, Australia, in 2009 and 2013, respectively. Now, he is a research fellow of Cyber-security at University of Technology, Sydney, Australia. His research interests focus on differential privacy, privacy preserving, and multi-agent systems.
\end{IEEEbiography}
\vfill
\vspace{-19mm}
\begin{IEEEbiography}[{\includegraphics[scale=0.67]{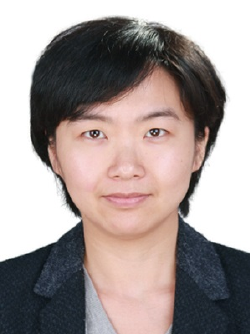}}]{A/Prof. Tianqing Zhu}
received the BEng and MEng degrees from Wuhan University, China, in 2000 and 2004, respectively, and the PhD degree in computer science from Deakin University, Australia, in 2014. Dr Tianqing Zhu is currently an associate professor at the School of Computer Science in University of Technology Sydney, Australia. 
Her research interests include privacy preserving, data mining and network security.
\end{IEEEbiography}
\vfill
\vspace{-19mm}
\begin{IEEEbiography}[{\includegraphics[scale=0.3]{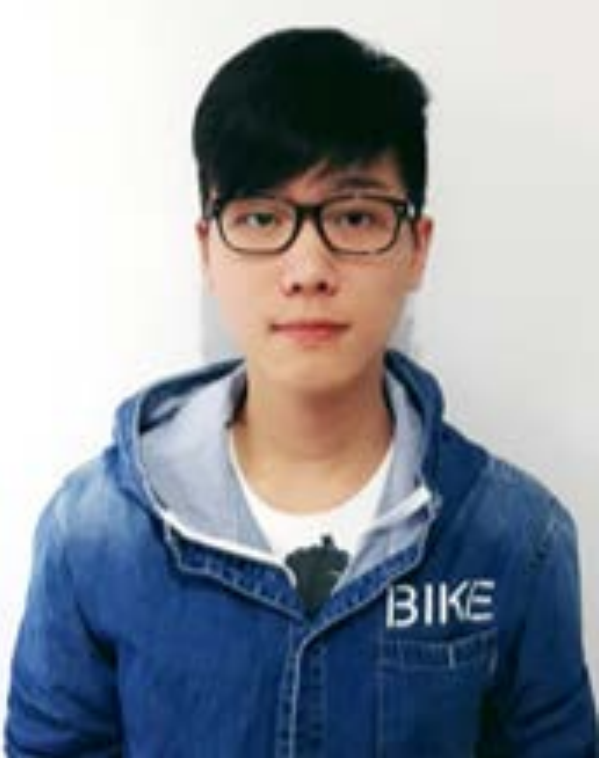}}]{Sheng Shen}
is pursuing his PhD at the School of Computer Science, University of Technology Sydney. He received the Bachelor of Engineering (Honors) degree in Information and Communication Technology from the University of Technology Sydney in 2017, and Master of Information Technology degree in University of Sydney in 2018. His current research interests include data privacy preserving, differential privacy and federated learning.
\end{IEEEbiography}
\vfill
\vspace{-19mm}
\begin{IEEEbiography}[{\includegraphics[scale=0.09]{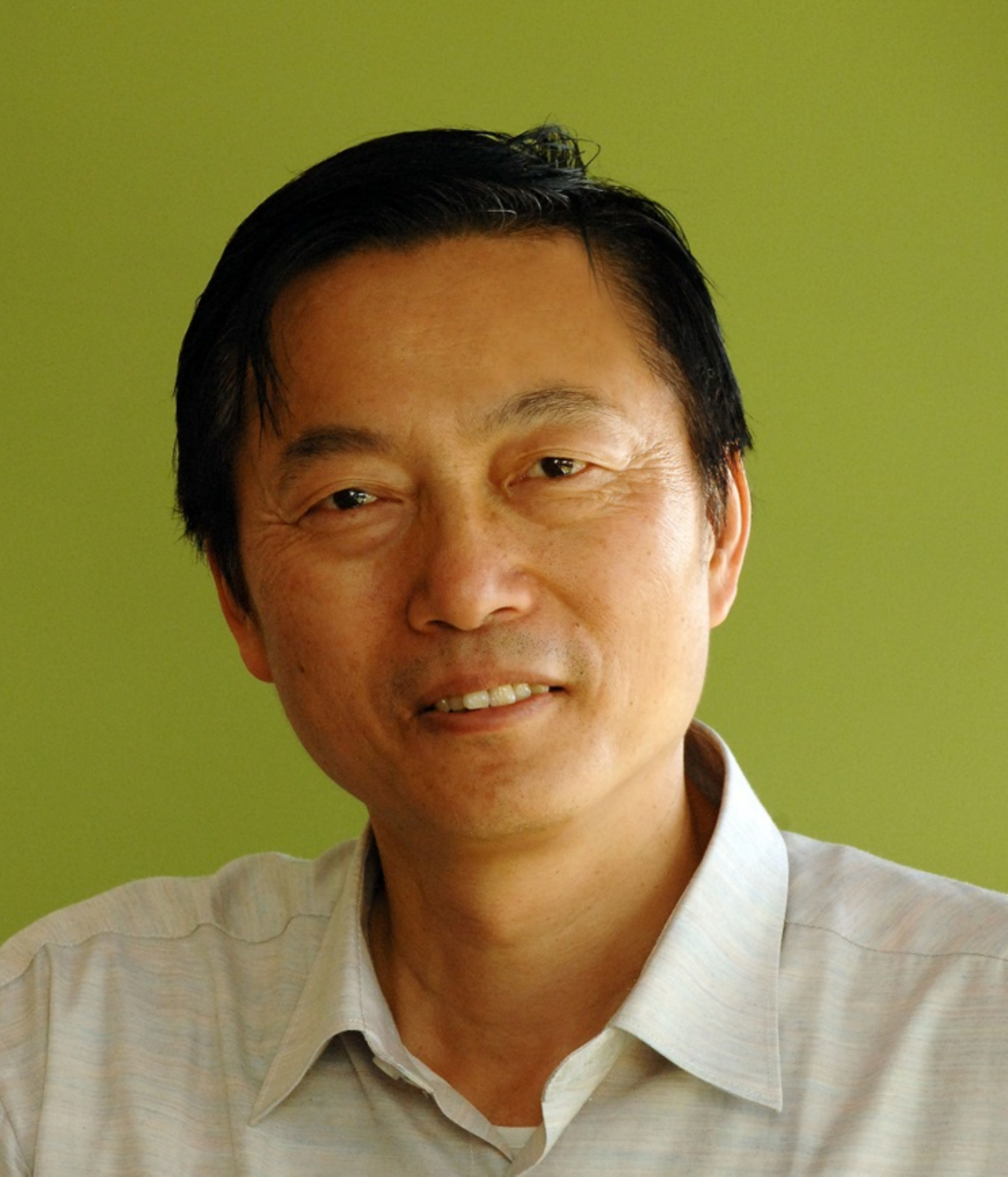}}]{Prof. Wanlei Zhou}
received the BEng and MEng degrees from Harbin Institute of Technology, Harbin, China in 1982 and 1984, respectively, and the PhD degree from The Australian National University, Canberra, Australia, in 1991, all in Computer Science and Engineering. He also received a DSc degree from Deakin University in 2002. He is currently the Head of School of School of Computer Science in University of Technology Sydney, Australia.
His research interests include distributed systems, network security, and privacy preserving.
\end{IEEEbiography}
\vfill

\end{document}